%% file: revised_revised_manuscript.tex
\documentclass[draft]{agujournal2019}
\usepackage[utf8]{inputenc}
\usepackage{url} 
\usepackage{lineno}
\usepackage{soul}

\usepackage{booktabs}
\usepackage{csquotes}
\usepackage{nicefrac}
\usepackage{amsmath,amsthm,amssymb,bbm}
\usepackage{rotating}
\usepackage{algorithm, algorithmic}
\usepackage{multirow}
\usepackage{subcaption}
\usepackage{upgreek}
\usepackage{tabularx}
\usepackage{float}
\usepackage{ragged2e}

\input{lettres_maths.tex}

\input{theoremes.tex}
\newcommand{\thinrule}{\specialrule{0.001em}{0em}{0em}}
\draftfalse

\journalname{Journal of Advances in Modeling Earth Systems (JAMES)}

\begin{document}

\title{FaIRGP: A Bayesian Energy Balance Model \\ for Surface Temperatures Emulation}

\authors{Shahine Bouabid\affil{1}, Dino Sejdinovic\affil{2}, Duncan Watson-Parris\affil{3}}

\affiliation{1}{Department of Statistics, University of Oxford, Oxford, UK}
\affiliation{2}{School of CMS \& AIML, University of Adelaide, Adelaide, Australia}
\affiliation{3}{Scripps Institution of Oceanography and Halicioğlu Data Science Institute, University of California, San Diego, US}

\correspondingauthor{Shahine Bouabid}{shahine.bouabid@stats.ox.ac.uk}

\justifying

\begin{keypoints}
    \item We introduce FaIRGP, a Bayesian machine learning emulator for global and local mean surface temperatures that builds upon a physically based simple climate model
    \item The model improves upon both purely physically-driven and purely data-driven baseline emulators on several metrics across realistic future scenarios.
    \item The model is fully mathematically tractable, which makes it a convenient and easy-to-use probabilistic tool for the emulation of surface temperatures, but also for downstream applications such as detection and attribution or precipitation emulation.
\end{keypoints}

\begin{abstract}
    Emulators, or reduced complexity climate models, are surrogate Earth system models that produce projections of key climate quantities with minimal computational resources. Using time-series modelling or more advanced machine learning techniques, data-driven emulators have emerged as a promising avenue of research, producing spatially resolved climate responses that are visually indistinguishable from state-of-the-art Earth system models. Yet, their lack of physical interpretability limits their wider adoption. In this work, we introduce FaIRGP, a data-driven emulator that satisfies the physical temperature response equations of an energy balance model. The result is an emulator that \textit{(i)} enjoys the flexibility of statistical machine learning models and can learn from data, and \textit{(ii)} has a robust physical grounding with interpretable parameters that can be used to make inference about the climate system. Further, our Bayesian approach allows a principled and mathematically tractable uncertainty quantification. Our model demonstrates skillful emulation of global mean surface temperature and spatial surface temperatures across realistic future scenarios. Its ability to learn from data allows it to outperform energy balance models, while its robust physical foundation safeguards against the pitfalls of purely data-driven models. We also illustrate how FaIRGP can be used to obtain estimates of top-of-atmosphere radiative forcing and discuss the benefits of its mathematical tractability for applications such as detection and attribution or precipitation emulation. We hope that this work will contribute to widening the adoption of data-driven methods in climate emulation.
\end{abstract}

\section*{Plain Language Summary}
Emulators are simplified climate models that can be used to rapidly explore climate scenarios --- they can run in less than a minute on an average computer. They are key tools used by the Intergovernmental Panel on Climate Change to explore the diversity of possible future climates. Data-driven emulators use advanced machine learning techniques to produce climate predictions that look very similar to the predictions of complex climate models. However, they are not easy to interpret, and therefore to trust in practice. In this work, we introduce FaIRGP, a data-driven emulator based on physics. The emulator is flexible and can learn from data to improve its predictions, but is also grounded on physical energy balance relationships, which makes it robust and interpretable. The model performs well in predicting future global and local temperatures under realistic future scenarios, outperforming purely physics-driven or purely data-driven models. Further, the probabilistic nature of our model allows for mathematically tractable uncertainty quantification. By gaining trust in such a data-driven yet physically grounded model, we hope the climate science community can benefit more widely from their potential.

\newpage
\section{Introduction}

Earth system models (ESMs)~\cite{flato2011esm} are key tools to understand current climate dynamics and climate change responses to greenhouse gas emissions. They constitute an extensive physical simulation of Earth's atmosphere and ocean fluid dynamics, used for example in the Couple Model Intercomparison Project~\cite{meehl2007wcrp, taylor2012overview, eyring2016overview} to study past and future climate. As such, they offer the most comprehensive view of what future climate could look like. They are also used as an idealized fully controlled environment to study climate dynamics and understand its underlying drivers. In particular, they play a central role in the estimation of key properties of the climate system such as timescales and equilibrium responses to the change in carbon dioxide concentration in the atmosphere~\cite{allen2009warming, collins2014longterm} and the effect of aerosols~\cite{levy2013roles}.

Running simulations with an ESM requires an astute understanding of the climate science background, of the numerical schemes used to simulate climate dynamics, and access to an adequate computational infrastructure\footnote{As an order of magnitude, running the CESM2 model~\cite{danabasoglu2020community} for a single year ahead takes about 2000 core hours on a supercomputer.}. Therefore, only a limited number of research teams around the world can realistically afford to perform climate simulations. A direct consequence is that a variety of scientific applications relying on future climate projections --- such as agricultural studies~\cite{rosenzweig2014assessing}, energy infrastructure models~\cite{turner2017climate} or global socio-economic human models~\cite{akhtar2013integrated} --- must settle for publicly available precomputed climate projections that have been selected based on their policy-relevance, and may not be tailored to the application needs. The need for expensive computational resources also serves as an important barrier, making it inaccessible for independent researchers and less well-equipped research teams to run experiments with ESMs. This exacerbates the already unequal representation in high-impact climate science research, where the global north is disproportionately represented~\cite{tandon2021analysis}.

Furthermore, even when the resource needs are met to run experiments with an ESM, their computational cost remains a critical obstacle. Indeed, the uncertainty over the ESM parameterisation, the climate system internal variability and the emission pathway the world will choose, together span a high-dimensional uncertainty space. Therefore, obtaining a comprehensive coverage of this uncertainty requires running numerous climate simulations, which quickly meets computational cost limitations. As a result, much of the climate variability and potential socio-economic pathways remain in practice unexplored.

Together, these limitations have fostered the emergence of simpler surrogate ESMs which are inexpensive to compute and referred to as \emph{emulators}. Unlike ESMs, emulators do not explicitly model the fluid dynamics of the atmosphere and oceans and focus on a limited number of climate features, such as surface temperatures or precipitations. They can run thousands of years of simulation in less than a minute on an average personal computer~\cite{leach2021fairv2}, hence making accessible the emulation of climate projections under configurations unexplored by ESMs.

\emph{Simple climate models} (SCMs)~\cite{meinshausen2011emulating, meinshausen2020shared, millar2017modified, smith2018fair, leach2021fairv2} are an important class of models that can emulate climate projections. They propose a reduced-order representation of the climate system, modelling the carbon cycle, radiative forcing, and temperature response with simplified first order dynamics. A key component of SCMs is their \emph{energy balance model} (EBM) \cite{rypdal2012global, geoffroy2013transient, millar2017modified, fredriksen2017long, lovejoy2021fractional}, which describes the changes in global surface temperature through imbalances in Earth's energy budget~\cite{held2010probing}. It represents the atmosphere-ocean system as a set of connected boxes forced by radiative flux at the top of the atmosphere. Whilst this constitutes a drastic simplification of the climate system, EBMs are robust physically-motivated models, and therefore remain the main tool used to connect the IPCC working physical basis research~\cite{lee2021future} with the adaptation and mitigation efforts~\cite{portner2022climate, shukla2022climate}.

Another important class of emulators are \emph{data-driven emulators}. In contrast to SCMs, they do not explicitly model climate dynamics and rather rely on statistical modelling techniques to emulate key climate variables such as temperature or precipitation~\cite{del2019predicting, beusch2020emulating, link2019fldgen, goodwin2020computationally, castruccio2014statistical, holden2010dimensionally}. Statistical modelling enjoys greater flexibility and has produced powerful emulators, capable of successfully approximating ESMs' spatially resolved climate projections with visually indistinguishable outputs. More recently, statistically driven emulators drawing from advances in statistical machine learning have demonstrated a remarkable capacity at regressing global emission profiles onto spatially resolved temperature and precipitation maps~\cite{watsonparris2021climatebench}.

However, both SCMs and data-driven emulators display fundamental limitations. Whilst SCMs provide a robust physical framework to reason about the climate system, they remain a simplistic model which may display a poor fit to ESM's outputs in certain scenarios~\cite{meinshausen2011emulating, geoffroy2013transient, jackson2022errors} and can only operate at the global level. Reasoning only in terms of global mean temperatures fails to capture the difference in exposure of different world regions and limits use for scientific applications that require regional projections. On the other hand, data-driven emulators are also limited in their ability to provide a complete, reliable picture of the climate system. Indeed, their lack of robust physical grounding limits their capacity to make inferences about the climate system~\cite{beusch2020emulating, link2019fldgen}. Further, the outputs from these emulators are often subject to qualitative evaluation and may not be trustworthy\footnote{Whilst statistical explainability methods~\cite{shapley1953value, vstrumbelj2014explaining} may help understand the contributions to predictions, we argue that the lack of a physically grounded model would still harm trust in the predictions.}, in part because of limited understanding on how they might behave outside the training data regimes.

In this work, we address these limitations by formulating a hybrid physical/statistical emulator that will both enjoy the robust physical grounding of SCMs, and the flexibility of statistical machine learning methods, thereby combining the advantages of both classes of emulators. We propose a probabilistic emulator for the task of reproducing an ESM temperature response to greenhouse gas and aerosol emissions. Our model builds upon a simple EBM, but draws from the latent force model paradigm~\cite{alvarez2009latent, alvarez2013linear} by placing a Gaussian process (GP) prior over the radiative forcing, thereby inducing a stochastic temperature response model. We show that the resulting emulated temperature turns out to also be a GP, with a physically-informed covariance structure that reflects the dynamics of the EBM. In consequence, we can update the EBM with temperature data to learn a posterior distribution over global temperatures, but also over radiative forcing. Further, we demonstrate how our model can be easily extended to emulate spatially-resolved maps of annual surface temperatures.

Experiments demonstrate skilful prediction of global and spatial mean surface temperatures, with improvements over a simple EBM and a simple GP model. We obtain robust predictions even when only training with historical data, or when predicting over scenarios outside the range of emissions of greenhouse gases and aerosols observed during training. Additionally, we show that the model improves over baselines to emulate temperature changes induced by anthropogenic aerosol emissions, and can provide useful estimates of the top-of-atmosphere radiative forcing.

\section{Background}

\subsection{Energy Balance Models}

Energy balance models (EBMs)~\cite{sellers1969global, holden2010dimensionally} are a key component of SCMs that provide a lower-order representation of the climate system, where the changes in global temperature are explained by the imbalance in the Earth's energy budget. The most common and established class of EBM are box models. They represent the atmosphere and ocean as a set of vertically stacked boxes, where the uppermost box is exposed to top-of-atmosphere radiative flux.

In a box model EBM, each box has a different heat capacity $C_i$, heat transfer coefficients $\kappa_i$ with the adjacent boxes, and its own temperature $T^{(i)}(t)$, as depicted in Figure~\ref{fig:k-box-model}. The uppermost box represents the fast components of the climate system, generally limited to the atmosphere and the land surface, while the following boxes are used to represent slower components of the climate system such as shallow ocean and deep ocean.

\begin{figure}[h]
    \centering
    \includegraphics[width=0.9\linewidth]{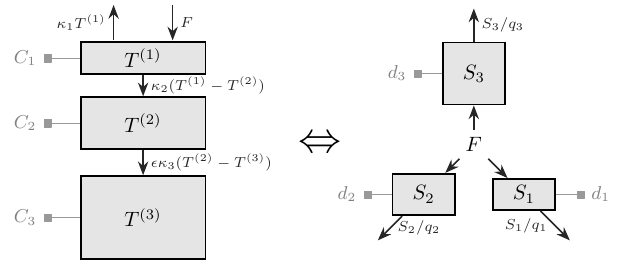}
    \caption{Schematic representation of $3$-box energy balance model, the arrows represent the flow of heat. \textbf{Left:} Temperature response model representation. Thicker boxes have a larger heat capacity (i.e.\ $C_1 < C_2 < C_3$) and the arrows represent the flow of heat. $\kappa_1, \kappa_2, \kappa_3$ are the heat transfer coefficients between the boxes and $\epsilon > 0$ denotes the deep ocean heat uptake efficacy. \textbf{Right:} Thermal response model representation. Thicker boxes have larger response timescale (i.e.\ $d_1 < d_2 < d_3$). Adapted from (Figure 1, Cummins et al., 2020).}
    \label{fig:k-box-model}
\end{figure}

Let $\bX(t)  = \begin{bmatrix}T^{(1)}(t) & \ldots & T^{(k)}(t)\end{bmatrix}^\top$
be the vector concatenation of the temperatures within the boxes in a model with $k$ boxes. The change in temperature of a $k$-box EBM is described by the following simple first order linear ordinary differential equation (ODE)
\begin{equation}\label{eq:temperature-response-ode}
    \frac{\d\bX(t)}{\d t} = \bA \bX(t) + \bb F(t),
\end{equation}
where $\bA$ is a tridiagonal temperature feedback matrix\footnote{explicit form of the matrix is provided in Appendix~\ref{appendix:derivation-of-fairgp}} that depends on heat capacities $C_i$, heat transfer coefficients $\kappa_i$, and deep ocean uptake efficacy $\varepsilon$, and $\bb$ is a radiative forcing feedback vector given by
\begin{equation}\label{eq:b-vector}
    \bb = \begin{bmatrix} 1/C_1 & 0 & \ldots & 0\end{bmatrix}^\top.
\end{equation}
$F(t)$ denotes the top of atmosphere effective radiative forcing --- radiative forcing for short --- i.e.\ the change in energy flux caused by natural or anthropogenic factors of climate change. It is only applied to the surface box, i.e.\ the uppermost box.

Whilst physically motivated, box-models remain an abstract representation of the climate system. Therefore their parameters, such as boxes heat capacities, are not realistic quantities that can be calculated, but rather need to be tuned against data using calibration methods~\cite{meinshausen2011emulating, millar2017modified} or maximum likelihood strategies~\cite{cummins2020optimal}.

The simplest box models only use $k = 2$ boxes, thereby splitting the climate system into 2 groups: fast and slow components. Whilst reductive, this split has proven to be a robust approximation of the climate system~\cite{held2010probing}. In fact, 2-box EBMs are today the primary tool used for linking the IPCC physical basis research~\cite{lee2021future} with the adaptation and mitigation efforts~\cite{portner2022climate, shukla2022climate}. It is however worth noting recent advocacy for 3-box models~\cite{leach2021fairv2}, highlighting the insufficiency of 2-box models to capture the full range of behaviour observed in CMIP6 models~\cite{tsutsui2020diagnosing, cummins2020optimal}.


\begin{table}[t]
\centering
    \vspace*{-3em}
    \caption{Notation, description and unit of named parameters.}
    \vspace*{-1em}
    \begin{subtable}[t]{\textwidth}
        \centering
        \caption{Deterministic parameters.}
        \vspace*{-0.5em}
        \begin{tabular}{lcc}\toprule
        Notation & Unit & Description \\ \midrule
            $F(t)$ & Wm\textsuperscript{-2} & Effective top of atmosphere radiative forcing \\
            $T^{(i)}(t)$ & K & Temperature of the $i$\textsuperscript{th} box \\
            $C_i$ & JK\textsuperscript{-1} & Heat capacity of the $i$\textsuperscript{th} box \\
            $\kappa_i$ & Wm\textsuperscript{-2}K\textsuperscript{-1} & Heat exchange coefficient of the $i$\textsuperscript{th} box \\
            $\varepsilon$ & 1 & Deep ocean heat uptake efficacy \\
            $T(t)$ & K & Global mean surface temperature ($T^{(1)}(t)$)\\
            $S_i(t)$ & K & Response of the $i$\textsuperscript{th} thermal box \\
            $d_i$ & years & Response timescale of $i$\textsuperscript{th} thermal box \\
            $q_i$ & KW\textsuperscript{-1}m\textsuperscript{2} & Equilibrium response of $i$\textsuperscript{th} thermal box \\ 
            $\chi$ & - & Atmospheric agent (e.g. CO\textsubscript{2}, SO\textsubscript{2}) \\
            $E(t)$ & - & Emission of agents in the atmosphere \\ \bottomrule
        \end{tabular}
        \label{table:deterministic-variables}
    \end{subtable}
    \begin{subtable}[t]{\textwidth}
        \centering
        \caption{Stochastic parameters.}
        \vspace*{-0.5em}
        \begin{tabular}{lcc}\toprule
        Notation & Unit & Description \\ \midrule
            $\sfF(t)$ & Wm\textsuperscript{-2} & Effective top of atmosphere radiative forcing \\
            $\sfT(t)$ & K & Global mean surface temperature \\
            $\sfS_i(t)$ & K & Response of the $i$\textsuperscript{th} thermal box \\ \bottomrule
        \end{tabular}
        \label{table:stochastic-variables}
    \end{subtable}
    \vspace*{-2em}
\end{table}

\subsection{Impulse response formulation}\label{subsection:impulse-response-formulation}

The main box of interest in a box-model is the uppermost box since it describes the global mean surface temperature response $T^{(1)}(t)$ to the radiative forcing $F(t)$. In general, computing $T^{(1)}(t)$ can be simplified by diagonalising the ODE (\ref{eq:temperature-response-ode}). The result is an equivalent impulse response formulation of the temperature response model, depicted on the right in Figure~\ref{fig:k-box-model}, and referred to as the \emph{thermal response ODE}~\cite{millar2017modified}.

Let us rename the temperature of the first box as $T(t)$ for simplicity. Then, for a $k$-box model, the thermal response ODE is formally given by
\begin{equation}\label{eq:impulse-response-formulation}
    \left\{
    \begin{aligned}
        \begin{split}
            & \frac{\d S_i(t)}{\d t} = \frac{1}{d_i} \left(q_i F(t) - S_i(t)\right), \quad 1\leq i\leq k \\
            & T(t) = \sum_{i=1}^k S_i(t),
        \end{split}
    \end{aligned}
    \right.
\end{equation}
where $S_i(t)$ is the $i$\textsuperscript{th} thermal response, $d_i$ is the $i$\textsuperscript{th} response timescale and $q_i$ is the $i$\textsuperscript{th} equilibrium response. Table~\ref{table:deterministic-variables} provides a brief description of named parameters and units.

Whilst it is harder to interpret the impulse response formulation in terms of physical processes, its primary benefit is that each thermal response ODE can be solved independently, thereby avoiding the intricacies of solving a coupled system of ODE. Further, the timescale and equilibrium parameters $d_i$ and $q_i$ can be expressed in terms of the original boxes heat capacities $C_i$ and heat transfer coefficients $\kappa_i$. Detailed derivations can be found in \cite{geoffroy2013transient}.

Throughout this work, we will use as a reference SCM FaIRv2.0.0~\cite{leach2021fairv2} (for Finite amplitude Impulse Response), a recent update of a well established SCM~\cite{millar2017modified, smith2018fair}, that offers a minimal level of structural complexity. FaIRv2.0.0 is effectively composed of 3 submodels: a gas cycle model, which converts emissions to concentrations, a radiative forcing model, which converts concentrations to radiative forcing, and a temperature response model, which converts radiative forcing into temperatures. The temperature response model of FaIRv2.0.0 exactly corresponds to the impulse response EBM described in (\ref{eq:impulse-response-formulation}). We refer to reader to the work of \citeA{leach2021fairv2} for a comprehensive presentation of FaIRv2.0.0. In the rest of the paper, we permit ourselves to drop v.2.0.0 and refer to the model as FaIR.

\subsection{Gaussian processes}

Gaussian processes (GPs)~\cite{rasmussen2005gaussian} are a ubiquitous class of Bayesian priors over real-valued functions. They enjoy convenient closed-form expressions, principled uncertainty quantification, and cover a rich class of complex functions. As a result, they have been widely used in various nonlinear and nonparametric regression problems in geosciences~\cite{camps2016survey}. 

We say that a real-valued stochastic process function $\sfff(t)$ is a GP if any finite collection of its evaluations has a joint multivariate normal distribution. The parameters of this distribution are fully determined by the GP mean function $m(t) := \EE[\sfff(t)]$ and its covariance function $k(t, t') := \operatorname{Cov}(\sfff(t), \sfff(t'))$. For example, for two inputs $t_1, t_2$ we have that
\begin{equation}
    \begin{bmatrix} \sfff(t_1) \\ \sfff(t_2) \end{bmatrix} \sim \cN\left(\begin{bmatrix} m(t_1) \\ m(t_2) \end{bmatrix}, \begin{bmatrix} k(t_1, t_1) & k(t_1, t_2) \\ k(t_2, t_1) & k(t_2, t_2)\end{bmatrix}\right).
\end{equation}
Because this extends to any finite collection of inputs  $t_1, \ldots, t_n$, we say that GPs induce a distribution over the entire function $\sfff(t)$. Figure~\ref{fig:gp-background-prior-samples}~(a) shows examples of functions sampled from this distribution, referred to as \emph{sample paths}. Formally, since this distribution is fully determined by $m(t)$ and $k(t, t')$, we write
\begin{equation}
    \sfff(t)\sim\GP(m, k).
\end{equation}
\begin{figure}[h]
    \centering
    \vspace*{-2em}
    \includegraphics[width=\linewidth]{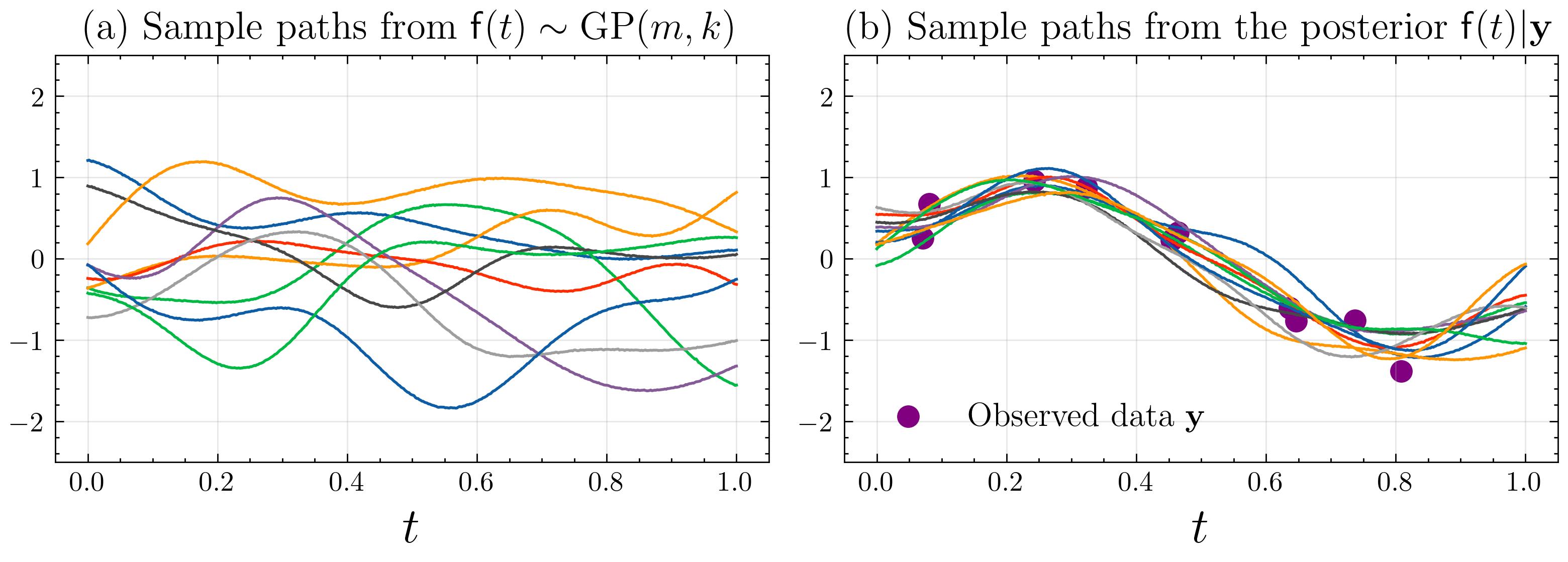}
    \vspace*{-2em}
    \caption{(a): 10 sample paths from $\sfff(t)\sim\GP(m, k)$ with $m(t) = 0$ and $k(t, t') = \sigma_\sfff^2 \exp(-|t-t|^2/\ell)$, $\sigma_\sfff^2 = 0.43$, $\ell = 0.17$. (b): 10 sample paths from the posterior $\sfff(t)|\by$ obtained by updating the prior. The data $\by$ is sampled from $\sfy = \sin(2\pi t) + \epsilon$ with $\epsilon\sim \cN(0, 0.04)$.}
    \vspace*{-3em}
    \label{fig:gp-background-prior-samples}
\end{figure}

When using GPs as a Bayesian prior over a function, the mean and covariance functions are typically user-specified. The mean $m(t)$ can be set as any function specifying the average behaviour of the function. The covariance function $k(t, t')$, commonly called \emph{kernel}, is chosen to specify the correlation between $\sfff(t)$ and $\sfff(t')$ by computing a notion of similarity between $t$ and $t'$. For example, a widely used family of kernels are the Matérn kernels~\cite{stein1999interpolation}, where the one-dimensional Matérn-1/2 kernel is given by
\begin{equation}
    k(t, t') = \sigma^2_\sfff \exp\left(-\frac{|t-t'|}{\ell}\right),
\end{equation}
with variance parameter $\sigma^2_\sfff$ and lengthscale parameter $\ell$. When $t$ and $t'$ are similar, i.e\ $|t - t'| \ll \ell$, this kernel specifies a strong correlation between $\sfff(t)$ and $\sfff(t')$, i.e.\ $\operatorname{Cov}(\sfff(t), \sfff(t')) / \sigma^2_\sfff$ $\approx 1$. Conversely, if $t$ and $t'$ are dissimilar, i.e.\ $|t - t'| \gg \ell$, the kernel tends to zero and specifies almost no correlation between $\sfff(t)$ and $\sfff(t')$.

Other popular choices of kernel include the squared exponential kernel $k(t, t') = \sigma^2_\sfff \exp(|t - t'|^2/\ell)$ which specifies a smooth correlation, or the cosine kernel $k(t, t') = \sigma^2_\sfff \cos(|t-t'|/\ell)$ which specifies a periodic correlation~\cite[Chapter 4]{rasmussen2005gaussian}. More broadly, kernels (and hence GPs) can also be defined over multidimensional inputs, for example by substituting $|t - t'|$ by the Euclidean distance between input vectors. When each dimension has a different lengthscale parameter, the kernel is referred to as an anisotropic, or automatic relevance determination, kernel~\cite[Chapter 5.1]{rasmussen2005gaussian}.

Let $\by = \begin{bmatrix} y_1 & \ldots & y_n\end{bmatrix}^\top$ denote independent data samples at inputs $\bt = \begin{bmatrix} t_1 & \ldots & t_n\end{bmatrix}^\top$ from the noisy data generating process $\sfy = \sfff(t) + \epsilon$, where $\epsilon\sim \cN(0, \sigma_\epsilon^2)$. It is possible to inform the GP with this data, thereby updating the prior into a posterior. In particular, the posterior distribution gives rise to a \emph{posterior GP}, with a posterior mean function $\bar m(t)$ and a posterior covariance function $\bar k(t, t')$,
\begin{equation}
    \sfff(t)\mid \by \sim \GP(\bar m,  \bar k).
\end{equation}
In other words, the posterior GP defines a new distribution over functions, parametrised by $\bar m(t)$ and $\bar k(t, t')$, which is informed by collected data. Figure~\ref{fig:gp-background-prior-samples}~(b) shows examples of posterior sample paths, which are now concentrated around the collected data points.

A great advantage of GPs is that the posterior mean and covariance can be analytically computed following
\begin{align}
    \bar m(t) & = m(t) + k(t, \bt)(\bK + \sigma_\epsilon^2\bI_n)^{-1} (\by - m(\bt)), \label{eq:gp-posterior-mean-update-rule} \\
    \bar k(t, t') & = k(t, t') - k(t, \bt)(\bK + \sigma_\epsilon^2\bI_n)^{-1}k(\bt, t), \label{eq:gp-posterior-covariance-update-rule}
\end{align}
where $\bI_n$ denotes the identity matrix of size $n$, $k(t, \bt) = \begin{bmatrix}
    k(t, t_1) & \ldots & k(t, t_n)\end{bmatrix}$, $k(\bt, t) = k(t, \bt)^\top$ and $\bK = k(\bt, \bt) = \begin{bmatrix}k(t_i, t_j)\end{bmatrix}_{1\leq i, j\leq n}$ is the covariance matrix.

We provide in Appendix~\ref{appendix:illustrations-gp} a detailed illustrated walk-through of GP regression, and refer the reader to \cite{rasmussen2005gaussian} for a comprehensive overview of GPs in machine learning.

\newpage
\section{FaIRGP}\label{section:fairgp}

In this section, we present FaIRGP, a GP emulator for global mean surface temperature that builds upon the thermal response model from FaIR. We begin by motivating a Bayesian treatment of the radiative forcing and formulate a GP prior over the forcing. We show this modification naturally results in a stochastic formulation of the thermal response model, which admits a GP with a physically-informed covariance structure for solution. In addition, we show how this framework can seamlessly account for the climate internal variability. Finally, we provide closed-form expressions for the posterior distributions over temperature and forcing, which can readily be used for emulation.

In what follows, we adopt the following notational conventions: deterministic variables are denoted with serif font ($Z$), stochastic variables are denoted with sans-serif font ($\sfZ$) and vector/matrix versions are denoted in bold ($\bZ$, $\sfbZ$).

\subsection{Epistemic uncertainty over the radiative forcing model}\vspace*{-0.5em}

The top-of-atmosphere radiative forcing is the fundamental quantity used to describe Earth's energy imbalance. The value of this forcing is primarily determined by the emissions of greenhouse gases and aerosols in the atmosphere. In FaIR, the radiative forcing term $F(t)$ drives the temperature response model (\ref{eq:temperature-response-ode}). Having access to a reliable estimate of the forcing, and in particular, to the forcing response to greenhouse gas and aerosol emissions, is therefore critical to produce trustworthy emulated temperatures.

Computing an accurate estimate of the radiative forcing requires modelling how greenhouse gases and aerosols interact with radiations in the atmosphere, accounting for atmospheric adjustments and uncertainty in the measured parameters which may further complicate the calculation. The intricacy of such tasks has led to the development of forcing estimation methods which rely on simplifying assumptions. For example, in FaIR, the forcing model is formulated as a combination of logarithmic, linear and square-root terms of greenhouse gas and aerosol concentrations. Namely, let $\chi$ denote a given atmospheric agent (e.g.\ CO\textsubscript{2}, CH\textsubscript{4}, SO\textsubscript{2}), the radiative forcing induced by $\chi$ is modeled as
\begin{equation}\label{eq:F-chi-deterministic}
    F^\chi(t) = \alpha_{\log}^\chi \log\left(\frac{C^\chi(t)}{C_0^\chi}\right) + \alpha_{\mathrm{lin}}^\chi \left(C^\chi(t) - C_0^\chi\right) + \alpha_{\mathrm{sqrt}}^\chi \left(\sqrt{C^\chi(t)} - \sqrt{C_0^\chi}\right),
\end{equation}\vspace*{-0.5em}
where $C^\chi(t)$ denotes the concentration in the atmosphere of agent $\chi$, $C_0^\chi$ the agent concentration at pre-industrial period and $\alpha_{\log}^\chi, \alpha_{\mathrm{lin}}^\chi, \alpha_{\mathrm{sqrt}}^\chi$ are scalar sensitivity coefficients. The total radiative forcing is then obtained by combining the contribution of each agent
\begin{equation}\label{eq:F-deterministic}
    F(t) = \sum_{\chi} F^\chi(t).
\end{equation}\vspace*{-1em}

This choice of forcing model has a physical motivation drawing from model calculations and empirical studies of temperature and historical trajectories of these concentrations~\cite{manabe1967thermal, forster2007changes, myhre2014ipcc}, which have provided substantial evidence that the concentration-to-forcing relationship can be reasonably approximated by the combination of terms in (\ref{eq:F-chi-deterministic}). However, it is important to recall that such a functional form is primarily based on empirical physics, in contrast with the EBM which has a greater theoretical motivation. In fact, the values of the sensitivity coefficients $\alpha_{\log}^\chi, \alpha_{\mathrm{lin}}^\chi, \alpha_{\mathrm{sqrt}}^\chi$ can vary substantially depending on the data used and the fitting procedure~\cite{leach2021fairv2}.

From a statistical machine learning perspective, the concentration-to-forcing relationship is arguably the component of FaIR which is most akin to statistical modelling. Indeed, $\Phi(t) = \begin{bmatrix} \log C^\chi(t) & C^\chi(t) & \sqrt{C^\chi(t)} \end{bmatrix}^\top$ can be interpreted as a feature engineering vector, where the features relevance have been identified by physics domain knowledge. The relationship (\ref{eq:F-chi-deterministic}) can then be rearranged as postulating a linear model $F^\chi(t) = \boldsymbol{\alpha}^\top \Phi(t) + \alpha_0$, where the coefficients $\boldsymbol{\alpha}$ are fitted against climate model data and the intercept $\alpha_0$ is determined by the pre-industrial concentration $C_0^\chi$.

Because the radiative forcing is a key driving quantity of the climate system, and therefore of the emulator, we argue it is critical to account for the uncertainty introduced by the FaIR model of radiative forcing. Specifically, this choice of a simplified model of forcing reflects uncertainty caused by the lack of knowledge about the phenomenon we wish to describe, referred to as \emph{epistemic uncertainty}~\cite{hullermeier2021aleatoric}. Drawing from the latent force modelling literature~\cite{alvarez2009latent}, we propose to acknowledge the epistemic uncertainty over the forcing model functional form by treating the radiative forcing as an underspecified latent force. Specifically, following the work of \citeA{alvarez2013linear}, we propose to account for the epistemic uncertainty over the forcing model through a Bayesian formalism and place a GP prior over the radiative forcing.

\subsection{Radiative forcing as a Gaussian process}

To specify a GP prior over the radiative forcing, we must first specify a choice of mean function. We propose to use the deterministic forcing function $F(t)$ as the mean function of our GP. This choice ensures that our prior will behave on average like the FaIR forcing model. Second, we must introduce a covariance function that specifies how different forcing levels should correlate as a function of some input variable.

The radiative forcing level is primarily determined by the emission rate of short-lived forcing agents and the cumulative emissions of long-lived forcing agents. Therefore, we anticipate a correlation of forcing levels for similar long-lived and short-lived forcing agent emissions trajectories. Conversely, for dissimilar emissions trajectories, we expect the resulting forcing levels to be fairly different, resulting in a lower correlation. For this reason, we propose to specify our covariance function as a function of long-lived forcing agents' cumulative emissions and short-lived forcing agents' emission rates.

Namely, let $E^\chi_r(t)$ and $E^\chi_c(t)$ be respectively the rate of emission and the cumulative emission of agent $\chi$ in the atmosphere at time $t$. Let $E^{\text{LL}}_c(t) = \begin{bmatrix} E^{\chi_1}_c(t) & \ldots & E^{\chi_d}_c(t)\end{bmatrix}^\top$ be the vector concatenation of long-lived agents $\{\chi_i\}_{i=1}^p$ cumulative emissions, and let $E^{\text{SL}}_r(t) = \begin{bmatrix} E^{\chi_{p+1}}_r(t) & \ldots & E^{\chi_d}_r(t)\end{bmatrix}^\top$ be the vector concatenation of short-lived agents $\{\chi_i\}_{i=p+1}^d$ rate of emission. The covariance function we introduce takes the form
\begin{equation}\label{eq:prior-rho-specification}
    K(t, t') = \rho\left(E(t) - E'(t')\right),
\end{equation}
where $E(t) = \begin{bmatrix}E^{\text{LL}}_c(t) \!&\! E^{\text{SL}}_r(t)\end{bmatrix}^{\!\!\top}$ and $E'(t') = \begin{bmatrix} E^{'\text{LL}}_c(t') \!&\! E^{'\text{SL}}_r(t')\end{bmatrix}^{\!\!\top}$ are two vectors encapsulating the cumulative emissions of long-lived forcing agents and emission rates of short-lived forcing agents for two different scenarios at times $t$ and $t'$, and $\rho$ is a user-specified kernel.

\begin{figure}[h]
    \centering
    \vspace*{-0.5em}
    \includegraphics[width=0.6\linewidth]{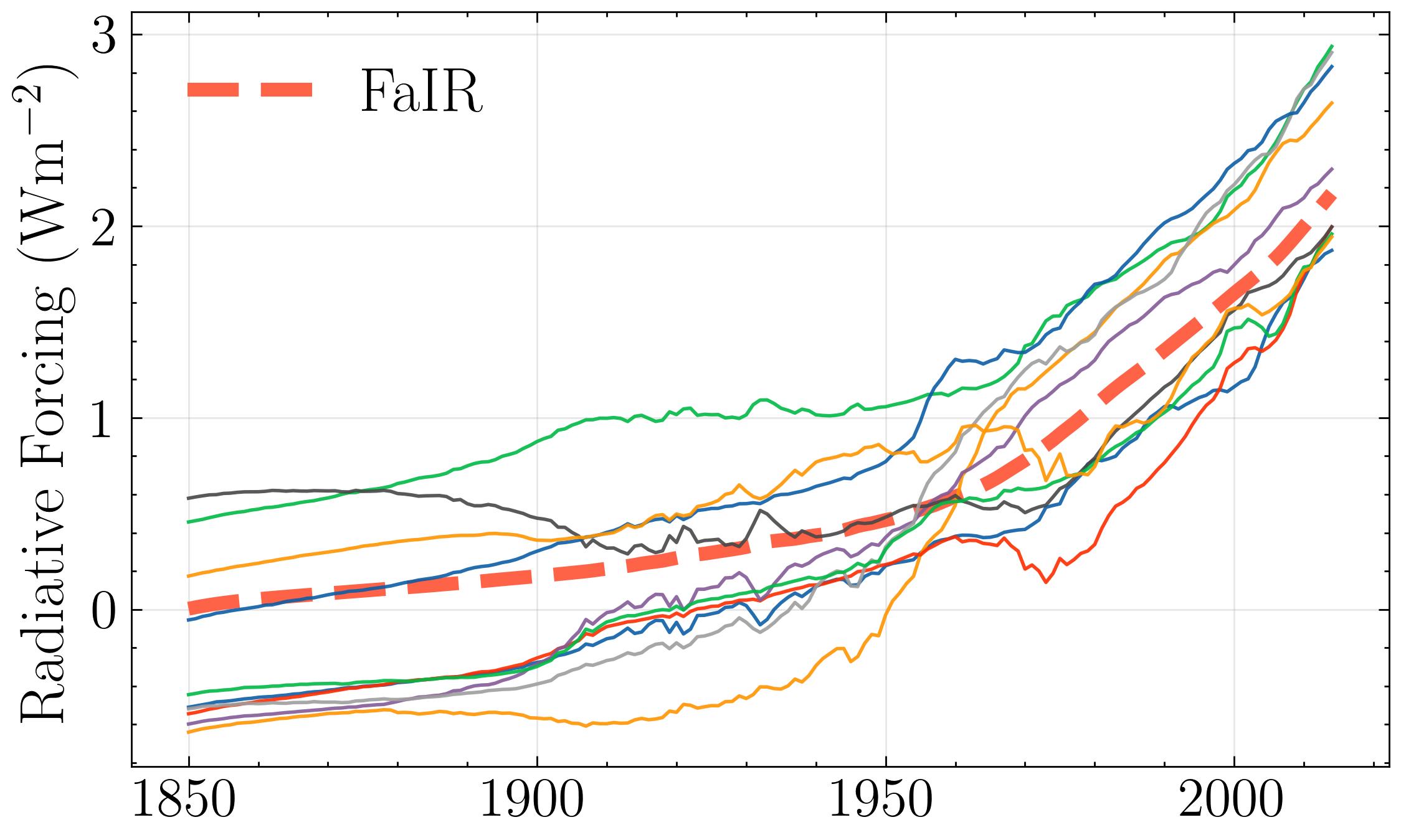}
    \vspace*{-0.5em}
    \caption{Example of sample paths from the radiative forcing prior $\sfF(t, E(t))\sim\GP(F, K)$ over the 1850-2014 period. The FaIR forcing response (dashed red) corresponds to the average sample path. The covariance function is taken as a Matérn-3/2 covariance.}
    \label{fig:sample-paths-from-sfF}
    \vspace*{-2em}
\end{figure}

With these notations, we can then formally define the prior we specify over the forcing model as
\begin{equation}
    \sfF(t, E(t)) \sim \GP(F, K),
\end{equation}
where we emphasise that $\sfF$ is now a stochastic process by using a sans-serif notation. This Bayesian formulation of the forcing effectively introduces a distribution over the functional forms the forcing response can take, and each sample path drawn from this distribution may describe a different function. Figure~\ref{fig:sample-paths-from-sfF} shows examples of sample paths from $\sfF(t, E(t))$.

Taking the mean function as the forcing model $F$ from (\ref{eq:F-chi-deterministic}) aligns sample paths with the FaIR forcing response on average. However, the probabilistic nature of this prior introduces a relaxation from exactly replicating the FaIR forcing response: sample paths can deviate from the mean, insofar as the forcing trajectories they describe maintain a correlation specified by $K(t, t')$ given their emission trajectories. This degree of freedom introduces a notion of variability over the forcing response, thereby allowing us to account for the epistemic uncertainty over the functional form of the FaIR forcing response.

To concretely understand how the covariance function affects $\sfF$, recall that for GPs,  $\operatorname{Cov}\big(\sfF(t, E(t)), \sfF(t', E'(t'))\big) = K(t, t')$. Therefore, if $E(t)$ and $E'(t')$ describe similar emission trajectories, the covariance $K(t, t')$ will take larger values, and specify a strong correlation between  $\sfF(t, E(t))$ and $\sfF(t', E'(t'))$. On the other hand, if $E(t)$ and $E'(t')$ describe very different emission trajectories, $K(t, t')$ will take values close to zero, and thus specify a low correlation between $\sfF(t, E(t))$ and $\sfF(t', E'(t'))$. Figure~\ref{fig:covariances-demo} shows an example of a covariance matrix between two scenarios with different CO\textsubscript{2} emissions trajectories.

The choice of kernel $\rho$ is an important choice as it specifies how correlation in forcings should be encoded as a function of emissions. For example, choosing a dirac kernel $\rho(E(t) - E'(t')) = \delta(E(t) - E'(t'))$ makes the forcing levels corresponding to $E(t)$ and $E'(t')$ independent, except if the emission trajectories are identical. Another extreme is to choose $\rho(E(t) - E'(t')) = 1$, which causes the forcing values taken to correlate regardless of emissions. In between these extremes, a widely used class of kernels are the Matérn family~\cite{stein1999interpolation}. They specify an exponentially decaying covariance as a function of the distance between emission trajectories, and most notably allow to control the smoothness of the GP sample path. Appendix~\ref{appendix:matern-covariance} provides detailed expressions of the Matérn covariance family and illustrations of their sample paths.

\begin{figure}[h]
    \centering
    \vspace*{-1.5em}
    \includegraphics[width=0.9\linewidth]{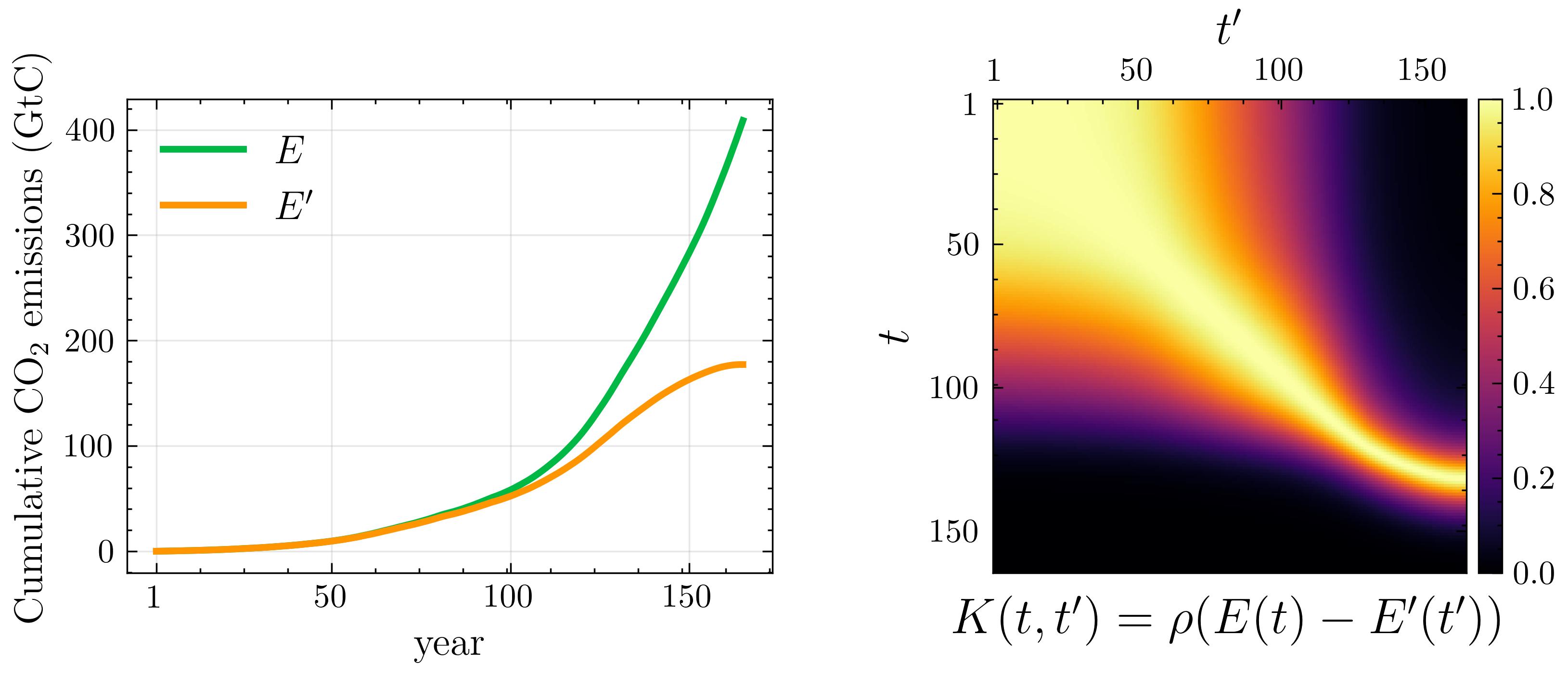}
    \vspace*{-0.5em}
    \caption{\textbf{Left} : Cumulative CO\textsubscript{2} emissions for a scenario with increasing emissions $E(t)$ and a scenario with reduced emissions $E'(t')$. \textbf{Right} : Cross-covariance matrix between emission scenarios for a Matérn-3/2 covariance, diagonals gives covariances for $t = t'$ and off-diagonals give covariances for $t\neq t'$. For $t, t' \leq 100$ the emission trajectories are similar, therefore $K(t, t')$ specifies a correlation between $\sfF(t, E(t))$ and $\sfF(t', E'(t'))$. For $t \geq 130$, $E(t)$ becomes dissimilar from the reduced emission scenario, therefore $K(t, t')$ specifies a low correlation. For $t'\geq 150$, $E'(t')$ stabilises to levels similar to $E(t\approx 130)$, therefore $K(t, t')$ specifies a correlation between $\sfF(t', E'(t'))$ and $\sfF(t\approx 130, E(t\approx 130))$.}
    \label{fig:covariances-demo}
    \vspace*{-2em}
\end{figure}

A popular choice within the Matérn family is the squared exponential kernel\footnote{to be exact, the squared exponential kernel corresponds to the limit of the Matérn family for an infinitely smooth kernel.}, which produces infinitely differentiable GP samples paths. We argue that this may not be faithful to real-world physical processes, which can exhibit abrupt changes. For this reason, we propose in this work to specify $\rho$ as an anisotropic Matérn-3/2 kernel given by
\begin{equation}
C_{3/2}\left(E(t) - E'(t')\right)  = \sigma^2_{\sfF} \left(1 + \sqrt{3\sum_{i=1}^d \left(\frac{\delta_i}{\ell_i}\right)^2}\right)\exp\left(-\sqrt{3\sum_{i=1}^d \left(\frac{\delta_i}{\ell_i}\right)^2}\right),
\end{equation}
where $\delta_i$ is the difference between the $i$\textsuperscript{th} component of $E(t)$ and $E'(t')$, $\sigma^2_{\sfF}$ is a variance parameter, $\ell_1, \ldots, \ell_d$ are independent lengthscale parameters associated to each atmospheric agent $\chi_1, \ldots, \chi_d$. The Matérn-3/2 covariance yields GP sample paths that are differentiable, thus allowing for smoothness, and yet permitting more abrupt changes compared to a squared exponential kernel.

In the remainder of the paper, we permit ourselves to drop notations and we denote the prior $\sfF(t)$ as a function of time only for the sake of conciseness. However, we emphasise that $\sfF(t)$, along with everything derived from it, fundamentally depends on emissions.

\subsection{Thermal response model with GP forcing}

We will now show that by combining the GP prior over the radiative forcing and the FaIR thermal response model, we obtain a GP prior over temperatures, which we name FaIRGP. Recall the thermal impulse response model presented in Section~\ref{subsection:impulse-response-formulation}, described by a system of $k$ independent linear first order ODEs
\begin{equation}\label{eq:jth-thermal-ode}
    \frac{\d S_i(t)}{\d t} = \frac{1}{d_i} (q_i F(t) - S_i(t)),
\end{equation}
where the thermal responses $S_i(t)$ are such that $T(t) = \sum_{i=1}^k S_i(t)$ is the global mean surface temperature.

We propose to substitute the deterministic forcing function $F(t)$ in (\ref{eq:jth-thermal-ode}) by its Bayesian counterpart $\sfF(t)$. In doing so, we naturally induce a stochastic version of the thermal impulse response model. Namely, the resulting $i$\textsuperscript{th} stochastic thermal response, which we denote $\sfS_i(t)$, will satisfy a differential equation forced by $\sfF(t)\sim\GP(m, K)$, and given by
\begin{equation}\label{eq:jth-thermal-sdee}
    \d \sfS_i(t)  = \frac{1}{d_i} (q_i \sfF(t) - \sfS_i(t))\d t.
\end{equation}

The resolution of this stochastic thermal response differential equation is similar to the resolution of the thermal response ODE. Namely, assuming $\sfS_i(0) = 0$ at pre-industrial time, the solution to (\ref{eq:jth-thermal-sdee}) takes the canonical form
\begin{equation}\label{eq:solution-to-sde}
    \sfS_i(t) = \frac{q_i}{d_i}\int_0^t \sfF(s) e^{-(t-s)/d_i}\d s.
\end{equation}

However, because $\sfF(t)$ is random, the solution is now a stochastic process. In particular, because $\sfF(t)$ is a GP to which we apply a linear convolution with an exponential function $e^{-(t-s)/d_i}$, we can show that $\sfS_i(t)$ must also be a GP, with its mean and covariance functions shaped by the form of this linear convolution.

Namely, the $i$\textsuperscript{th} thermal response can be characterised as a GP with the following mean and covariance functions
\begin{equation}\label{eq:sj-stochastic}
    \left\{
    \begin{aligned}
        \begin{split}
            & \sfS_i(t)  \sim \GP(m_i, k_{ii}) \\
            & m_i(t) = \frac{q_i}{d_i}\int_0^t F(s) e^{-(t-s)/d_i}\d s \\
            & k_{ii}(t, t') = \left(\frac{q_i}{d_i}\right)^2\int_0^t\int_0^{t'} K(s, s') e^{-(t-s)/d_i}e^{-(t'-s')/d_i}\d s \d s'.
        \end{split}
    \end{aligned}
    \right .
\end{equation}

We observe that the mean function $m_i(t)$ exactly corresponds to the solution of the deterministic thermal response ODE (\ref{eq:jth-thermal-ode}). This is consistent with our expectation: because the forcing sample paths from $\sfF(t)$ are specified to behave on average like the FaIR forcing response, the GP thermal response sample paths $\sfS_i(t)$ should behave on average like the FaIR thermal response model. Further, the covariance $k_{ii}(t, t')$ is expressed as a function of the forcing prior covariance $K$, but also of the parameters of the EBM, $d_i$ and $q_i$. As such, $k_{ii}(t, t')$ defines a physically-informed covariance structure that propagates the uncertainty over the forcing $\sfF(t)$ --- specified by our Bayesian prior --- into uncertainty over the thermal response $\sfS_i(t)$.


Under this model, the correlation $k_{ii}(t, t')$ between two stochastic thermal responses $\sfS_i(t)$ and $\sfS_i(t')$ is influenced by how fast the response of the $i$\textsuperscript{th} thermal box is, and by how much the radiative forcings trajectories leading to $\sfS_i(t)$ and $\sfS_i(t')$ did correlate. Indeed, the covariance in (\ref{eq:sj-stochastic}), expressed in integral form, specifies a stronger correlation if the past correlation $K(s, s')$ between stochastic forcing trajectories $(\sfF(s))_{s\leq t}$ and $(\sfF(s'))_{s'\leq t'}$ is greater. This aligns with the idea that similar forcing processes $\sfF(s)$ and $\sfF(s')$ should lead to similar thermal responses. Further, a fast thermal box (small $d_i$) primarily influences the integral with local correlations in forcings just before $t$ and $t'$, reflecting the short-term impact of fast components of the climate system. Conversely, a slow thermal box (large $d_i$) accumulates forcing covariance over an extended period before $t$ and $t'$, reflecting the long-term influence of slow components of the climate system. Figure~\ref{fig:demo-kjs} shows an example of cross-covariance matrices for a stochastic two-box thermal response model.

\begin{figure}[h]
    \centering
    \vspace*{-1.5em}
    \includegraphics[width=0.9\linewidth]{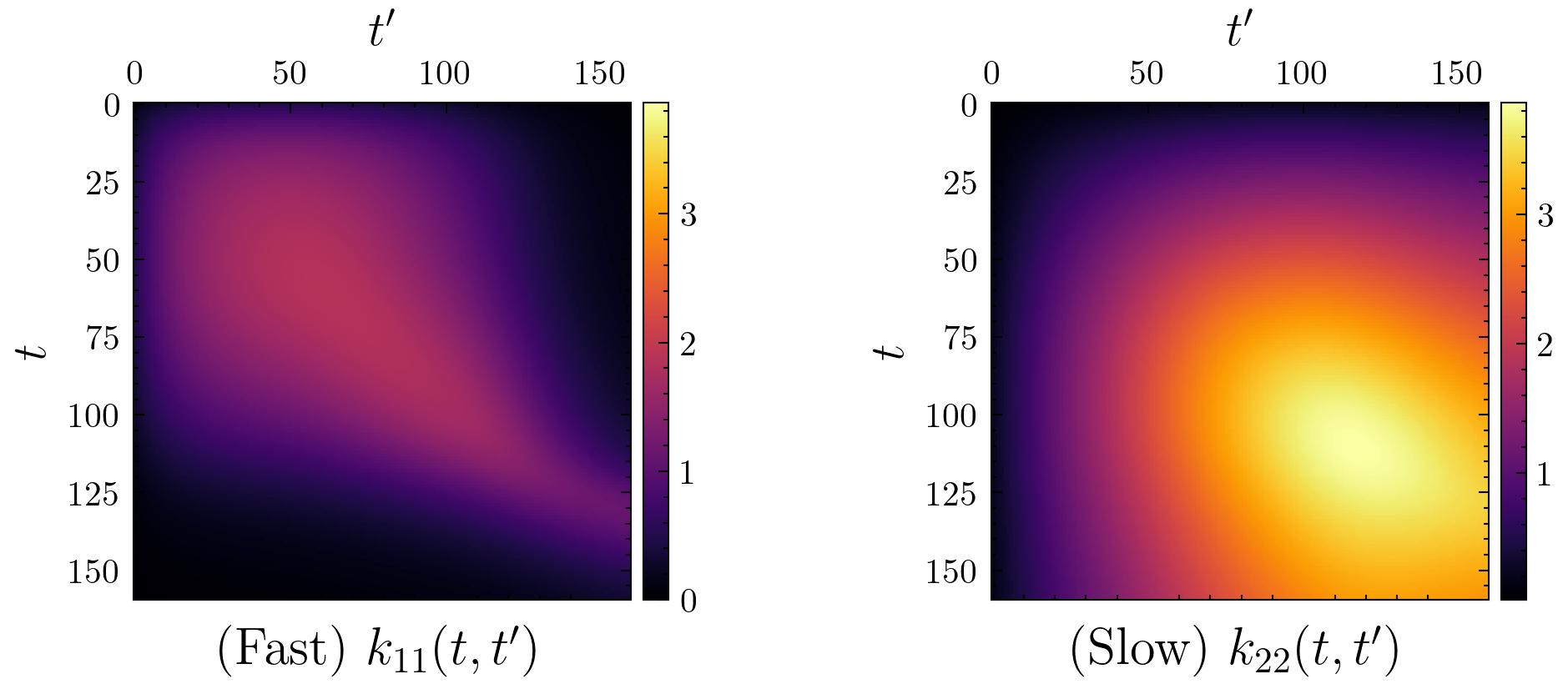}
    \caption{Cross-covariance matrices for the stochastic thermal responses $\sfS_1(t)$ and $\sfS_2(t)$ from a 2-box model under the two emission scenarios from Figure~\ref{fig:covariances-demo}. \textbf{Left} : Cross-covariance matrix of a fast thermal box with $d_1 = 8$ years and $q_1 = 2$ KW\textsuperscript{-1}m\textsuperscript{-2}. Only the forcings correlations at times right before $(t, t')$ influence $k_{11}(t, t')$, hence the specified correlation is moderate and the form of the thermal box covariance matrix is similar to the form of the forcing covariance matrix $K(t, t')$ in Figure~\ref{fig:covariances-demo}. \textbf{Right} : Cross-covariance matrix of a slow thermal box with $d_2 = 80$ years and $q_2 = 20$ KW\textsuperscript{-1}m\textsuperscript{-2}. The covariance $k_{22}(t, t')$ integrates forcing correlations over an extended period, which specifies a stronger correlation, and additionally models the memory effect of the climate system: even after the two emission scenarios start to diverge for $t, t'\geq 130$, the cross-covariance remains important.}
    \vspace*{-2em}
    \label{fig:demo-kjs}
\end{figure}

If we now define the global mean surface temperature as the sum of thermal response GPs $\sfT(t) = \sum_{i=1}^k \sfS_i(t)$, then we can show that $\sfT(t)$ must also be a GP. Namely, let us define the cross-covariance function between thermal responses $\sfS_i(t)$ and $\sfS_j(t')$ as
\begin{equation}
    k_{ij}(t, t') := \operatorname{Cov}(\sfS_i(t), \sfS_j(t')) = \frac{q_i q_j}{d_i d_j}\int_0^t\int_0^{t'} K(s, s') e^{-(t-s)/d_i}e^{-(t'-s')/d_j}\d s \d s',
\end{equation}
which describes how the stochastic thermal responses from different boxes should correlate within our model. Then, the global mean surface temperature is a GP specified by the following mean and covariance functions

\begin{equation}\label{eq:t-stochastic}
    \left\{
    \begin{aligned}
        \begin{split}
            & \sfT(t)  \sim \GP(m_\sfT, k_\sfT), \\
            & m_\sfT(t) = \sum_{i=1}^k m_i(t), \\
            & k_\sfT(t, t')  = \sum_{i=1}^k\sum_{j=1}^k k_{ij}(t,t').
        \end{split}
    \end{aligned}
    \right.
\end{equation}

The mean function $m_\sfT(t)$ exactly corresponds to the FaIR temperature response model from (\ref{eq:impulse-response-formulation}). This is again consistent with the expectation that since $\sfF(t)$ behaves on average like the FaIR forcing model, our GP temperature response $\sfT(t)$ should behave on average like the FaIR temperature response. The covariance function $k_\sfT(t, t')$, which specifies the correlation between $\sfT(t)$ and $\sfT(t')$, includes the covariances of individual thermal boxes $k_{ii}(t, t')$, aligning with the intuition that if thermal responses correlate, so should global mean surface temperatures. However, we note $k_\sfT(t, t')$ also considers correlations between different thermal boxes $k_{ij}(t, t')$ for $i\neq j$. Therefore, if certain components of the climate system exhibit correlated responses across different timescales, this contributes positively to the correlation between $\sfT(t)$ and $\sfT(t')$ within our model.

By specifying a GP prior over the radiative forcing, we have obtained a GP prior over the temperature that uses the FaIR thermal response model as its backbone, which we name FaIRGP. Using FaIRGP serves as a principled measure of epistemic uncertainty over the emulator design. In particular, the integral form of its covariance function allows us to account for past trajectories, thereby capturing the climate system memory effect, and producing robust uncertainty estimates. In addition, as we will see in Section~\ref{subsection:posterior-distribution}, FaIRGP can go beyond a standard impulse response model by learning from data using standard GP regression techniques. We emphasise that whilst we abuse notations for conciseness, the forcing prior $\sfF(t)$ is effectively a function of emissions (or cumulative emissions) through its covariance function $\rho\left(E(t) - E'(t')\right)$. Therefore, $\sfT(t)$ is also a function of emissions, and its covariance function can be understood as \enquote{$k_\sfT\left(E(t), E'(t')\right)$}.

\subsection{Accounting for climate internal variability}\label{subsection:fairgp-internal-variability}

An important component of the climate system is its internal variability, which integrates the effects of weather phenomena --- typically operating on the scale of days --- into elements of the climate system, such as the ocean, cryosphere and land vegetation --- which rather operate on the scale of months, years or decades~\cite{hasselman1976stochastic}.

The climate internal variability can classically be modelled in a $k$-box model by introducing a white noise forcing disturbance over the uppermost box, i.e.\ the atmosphere and land surface box~\cite{cummins2020optimal}. Formally, let $\sfB(t)$ be the standard one-dimensional Brownian motion and let $\sfbX(t)$ denote a stochastic version of the $k$-box model temperatures from (\ref{eq:temperature-response-ode}). The temperature response model with internal variability is given by
\begin{equation}
    \d \sfbX(t) = \bA \sfbX(t)\d t + \bb F(t)\d t + \sigma\bb \d\sfB(t),
\end{equation}
where $\sigma > 0$ is a variance term that controls the amplitude of the white noise, and we recall that $\bb$ has zero everywhere but its first entry, therefore the white noise disturbance is only applied to the uppermost box. When the radiative forcing $F(t)$ is deterministic, this is equivalent to adding red noise onto the global mean surface temperature, or equivalently, modelling it as an Ornstein-Uhlenbeck process\footnote{In the literature, it is common to encounter its discrete time analogue, the autoregressive process of order 1 AR(1).}. The corresponding diagonalised impulse response system is given by
\begin{equation}
    \d\sfS_i(t) = \frac{1}{d_i}\left(q_i F(t) - \sfS_i(t)\right)\d t + \sigma \frac{q_i}{d_i}\d\sfB(t),
\end{equation}
where derivations are detailed in Appendix~\ref{appendix:internal-variability-proof}. If we now again substitute the deterministic forcing $F(t)$ with the GP forcing $\sfF(t)$, it can be shown that the long time regime solution to this stochastic impulse response system can be expressed as
\begin{equation}\label{eq:sj-stochastic-with-internal-variability}
    \sfS_i(t) \sim \GP(m_i, k_{ii} + \sigma^2\gamma_i),
\end{equation}
where $\gamma_i(t, t')$ is an exponential, or Matérn-1/2, kernel function given by
\begin{equation}\label{eq:gamma-internal-variability}
    \gamma_i(t, t') = \frac{q_i^2}{2d_i} \exp\left(-\frac{|t - t'|}{d_i}\right).
\end{equation}

Therefore, FaIRGP can easily model the internal variability of the climate system simply by modifying its covariance structure. We observe that accounting for internal variability essentially corresponds to adding an independent autocorrelated noise process $\upxi_i(t) \sim \GP(0, \sigma^2\gamma_i)$ to the thermal response GP obtained in (\ref{eq:sj-stochastic}).

\begin{figure}[t]
    \centering
    \vspace*{-0.5em}
    \includegraphics[width=\linewidth]{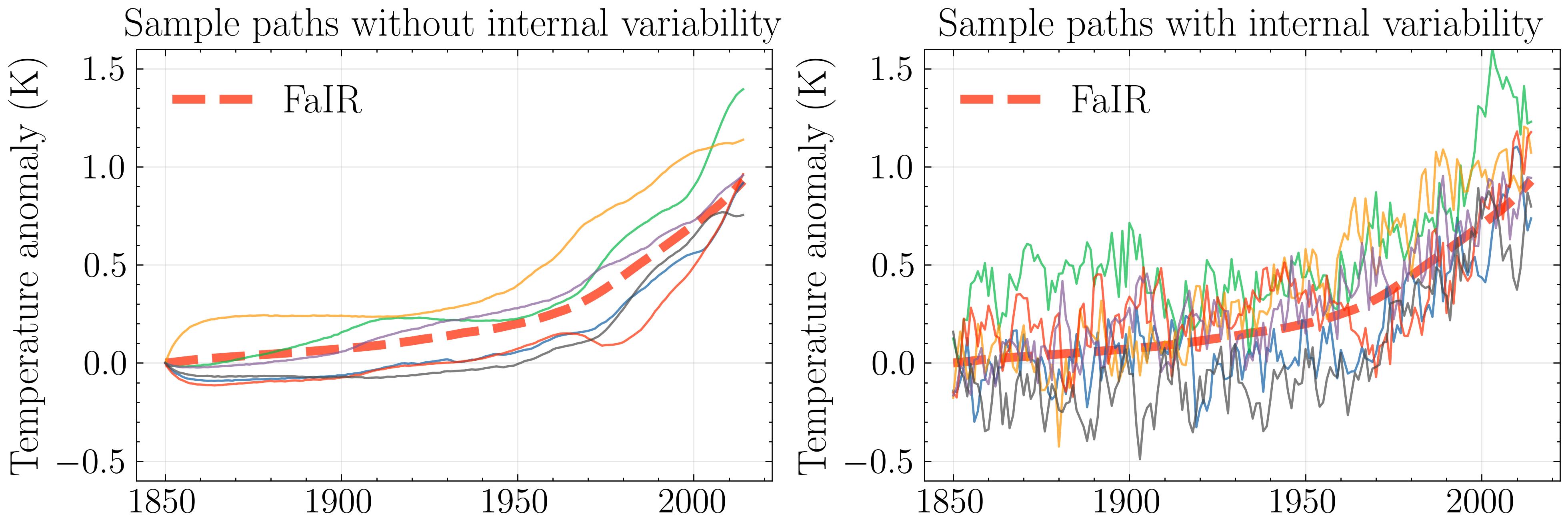}
    \vspace*{-2em}
    \caption{\textbf{Left}: Example of sample paths from the FaIRGP prior $\sfT(t)\sim\GP(m_\sfT, k_\sfT)$ without internal variability over the 1850-2014 period. \textbf{Right}: Example of sample paths from the FaIRGP prior with internal variability $\sfT(t)\sim\GP(m_\sfT, k_\sfT + \sigma^2\gamma_\sfT)$ over the same period. The FaIR temperature response (dashed red) corresponds in both cases to the average sample path.}
    \vspace*{-2em}
    \label{fig:fairgp-sample-paths}
\end{figure}

Ultimately, we sum the thermal response together to obtain a surface temperature GP $\sfT(t) = \sum_{i=1}^k \sfS_i(t)$, which in the long time regime takes the form
\begin{equation}
    \left\{
    \begin{aligned}
        \begin{split}
            & \sfT(t)  \sim \GP\left(m_\sfT, k_\sfT + \sigma^2 \gamma_\sfT\right) \\
            & \gamma_\sfT(t, t') = \sum_{i=1}^k \nu_i \gamma_i(t, t'),
        \end{split}
    \end{aligned}
    \right.
\end{equation}
where $\nu_i$ is a dimensionless weight determined by $q_i, d_i$ which accounts for cross-covariances between thermal responses internal variabilities. Its detailed expression can be found in Appendix~\ref{appendix:internal-variability-proof}.

Therefore, we can account for the climate internal variability simply by adding an independent noise process $\upxi_\sfT(t) = \sum_{i=1}^k \upxi_i(t) \sim \GP(0, \sigma^2\gamma_\sfT)$ to the temperature response GP obtained in (\ref{eq:t-stochastic}). Figure~\ref{fig:fairgp-sample-paths} shows examples of sample paths from the FaIRGP prior with and without accounting for internal variability.


\subsection{Posterior distribution over temperature and radiative forcing}\label{subsection:posterior-distribution}

In FaIRGP, the surface temperature response and the radiative forcing model can both be informed by global temperature data. Using standard GP regression techniques~\cite{rasmussen2005gaussian}, FaIRGP can learn from data how to deviate from its backbone impulse response model to best account for actual temperature simulated by a climate model. The resulting model can then be used to emulate future temperatures, benefiting from both the robustness of its prior and the flexibility of GP regression. As we will demonstrate in Section~\ref{section:global-experiments}, we can inform FaIRGP with historical observations and climate model data.

Assume we are under a fixed emission scenario e.g.\ working with historical emissions only. For timesteps $t_1 < \ldots < t_n$, suppose we simulate global mean surface temperatures $T_1, \ldots, T_n$, and also have access to greenhouse gas and aerosol emissions data $E_1, \ldots, E_n \in \RR^d$, where $d$ corresponds to the number of atmospheric agents $\chi_1, \ldots, \chi_d$. We concatenate this data into $\bt = \begin{bmatrix} t_1 & \ldots & t_n\end{bmatrix}^\top$, $\bE = \begin{bmatrix} E_1 & \ldots & E_n\end{bmatrix}^\top$ and $\bT = \begin{bmatrix}T_1 & \ldots & T_n\end{bmatrix}^\top$.


With notation abuse, the kernel $k_\sfT$ allows to specify how the FaIRGP prior correlates between times $t_i$ and $t_j$ by computing the similarity between emission data, i.e.\ $k_\sfT(t_i, t_j) = k_\sfT(E_i, E_j)$, because the forcing prior kernel $\rho$ is a function of emissions. The internal variability kernel $\gamma_\sfT$ allows us to specify how the additive internal variability noise process correlates between times $t_i$ and $t_j$ by computing $\gamma_\sfT(t_i, t_j)$. By extending this to the entire dataset, we can define the covariance matrices induced by $k_\sfT$ and $\gamma_\sfT$ over $\bE$ and $\bt$ following
\begin{align}
    \bK & = k_\sfT(\bE, \bE) = \begin{bmatrix}k_\sfT(E_i, E_j)\end{bmatrix}_{1\leq i,j\leq n} \\
    \bGamma & = \gamma_\sfT(\bt, \bt) = \begin{bmatrix}\gamma_{\sfT}(t_i, t_j)\end{bmatrix}_{1\leq i, j\leq n}.
\end{align}\vspace*{-2em}

Intuitively, $\bK$ captures how much the values taken by the stochastic global mean temperature response from the FaIRGP prior will correlate between different times. Similarly, $\bGamma$ captures how much the values taken by the internal variability process $\upxi_\sfT$ correlate between different time steps. Combining them into $\bK + \sigma^2\bGamma$, we obtain a covariance matrix that captures the total correlation we have between different times for the values taken by the FaIRGP prior with internal variability.

\begin{figure}[h]
    \centering
    \vspace*{-0.8em}
    \includegraphics[width=\linewidth]{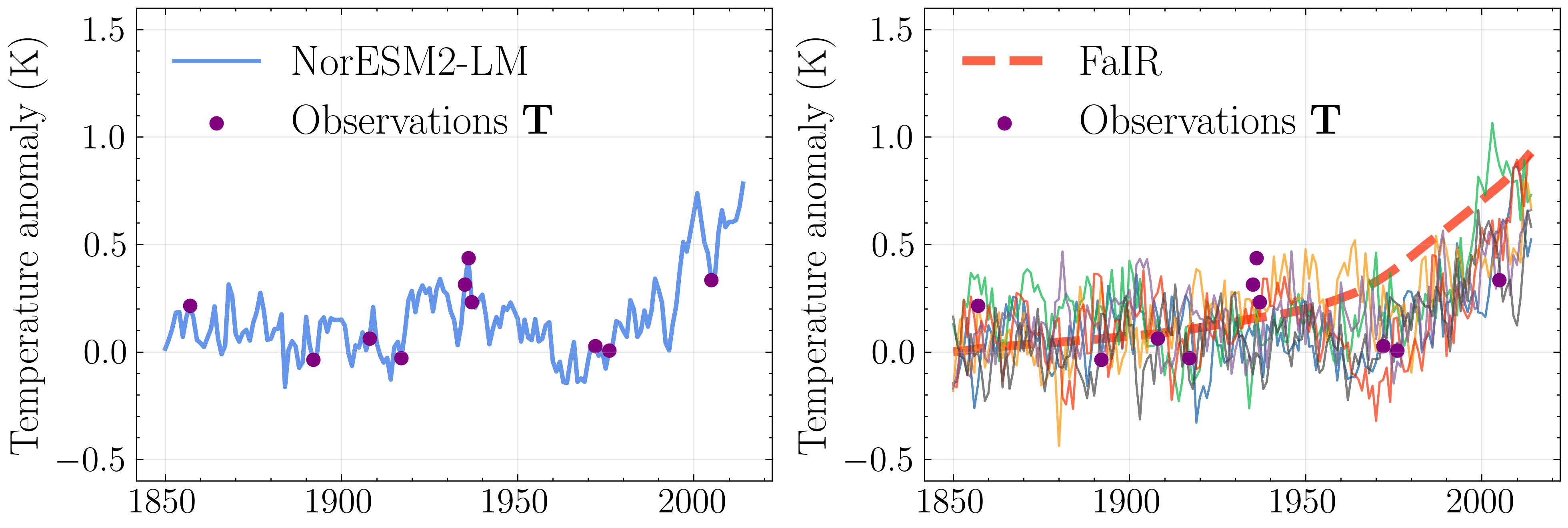}
    \vspace*{-1.5em}
    \caption{\textbf{Left}: Historical global mean surface temperature anomaly simulated with NorESM2-LM from which we randomly extract 10 data points (purple). \textbf{Right}: Example of sample paths from the FaIRGP posterior $\sfT(t)|\bT$ conditioned on the data. In contrast with the FaIRGP prior sample paths (see Figure~\ref{fig:fairgp-sample-paths}), posterior sample paths are more concentrated around data points and deviate from the FaIR temperature response (dashed red).}
    \label{fig:fairgp-posterior-sample-paths}
    \vspace*{-2.2em}
\end{figure}

Because it captures both the effect of internal variability and of the global mean temperature response to changes in emissions, the covariance matrix $\bK + \sigma^2\bGamma$ is well suited to capture the correlation structure we get from global mean surface temperature data $\bT$ simulated with a climate model. Using this new covariance matrix, we can then apply the update rules from (\ref{eq:gp-posterior-mean-update-rule}) and (\ref{eq:gp-posterior-covariance-update-rule}) to update our prior over $\sfT(t)$ with climate model data $\bT$, and obtain a posterior GP over global mean surface temperatures given by
\begin{equation}\label{eq:posterior-temperature}
    \sfT(t)\mid \bT \sim \GP(\bar m_\sfT, \bar k_\sfT),
\end{equation}
where the posterior mean $\bar m_\sfT(t)$ and the posterior covariance $\bar k_\sfT(t, t')$ have the following closed-form expressions
\begin{align}\label{eq:posterior-temperature-mean}
    \underbrace{\bar m_\sfT(t)}_{\text{posterior mean}} & = \underbrace{m_\sfT(t)}_{\text{prior mean}} \!\! + \,\,\, \underbrace{k_\sfT(t, \bt)(\bK + \sigma^2 \bGamma)^{-1}(\bT - m_\sfT(\bt))}_{\text{posterior correction}} \\
    \label{eq:posterior-temperature-covar}
    \underbrace{\bar k_\sfT(t, t')}_{\text{posterior covariance}} & = \underbrace{k_\sfT(t, t')}_{\text{prior covariance}} \!\!\! - \,\,\,\underbrace{k_\sfT(t, \bt)(\bK + \sigma^2 \bGamma)^{-1}k_\sfT(\bt, t')}_{\text{posterior correction}}.
\end{align}

The posterior mean $\bar m_\sfT(t)$ explicitly corrects the prior mean $m_\sfT(t)$, which we recall corresponds to the FaIR deterministic temperature response model. The correction introduced by the posterior mean in (\ref{eq:posterior-temperature-mean}) is a data-driven correction term, informed by the temperature data $\bT$, times $\bt$, and the emissions $\bE$ (which are used to compute the covariance matrix $\bK$). Intuitively, the correction term regresses the discrepancy between the simulated temperature data and the FaIR temperature response onto the emission data. Thus, it learns how much to deviate from the FaIR temperature response model to better fit climate model data, given emissions data. Similarly, the posterior covariance $\bar k_\sfT(t, t')$ explicitly adjusts the prior covariance $k_\sfT(t, t')$ with a data-driven correction term. Intuitively, this correction term acknowledges the new information we have updated our model with, such that when predicting over emission trajectories similar to the ones from the training data, the variability of the FaIRGP global mean surface temperature response is reduced. Figure~\ref{fig:fairgp-posterior-sample-paths} shows examples of sample paths from the FaIRGP posterior when updated with simulated data of global mean surface temperature.


This posterior can then be used to emulate temperatures by making predictions at future time steps. A major interest of this formulation is that whilst the emulator hinges on the robustness of the energy balance model for its prior, it can also benefit from the flexibility of GP regression by informing it with data in its posterior, thereby going beyond a simple impulse response model. In addition, the GP approach enjoys full probabilistic tractability, with the analytical expression of the posterior probability distribution provided in Appendix~\ref{appendix:analytical-expressions}. This means we can analytically compute the mean, the 95\% credible range and quantiles of the predicted surface temperature anomaly as shown in Figure~\ref{fig:mean-and-quantiles}.

\begin{figure}[h]
    \centering
    \vspace*{-1em}
    \includegraphics[width=\linewidth]{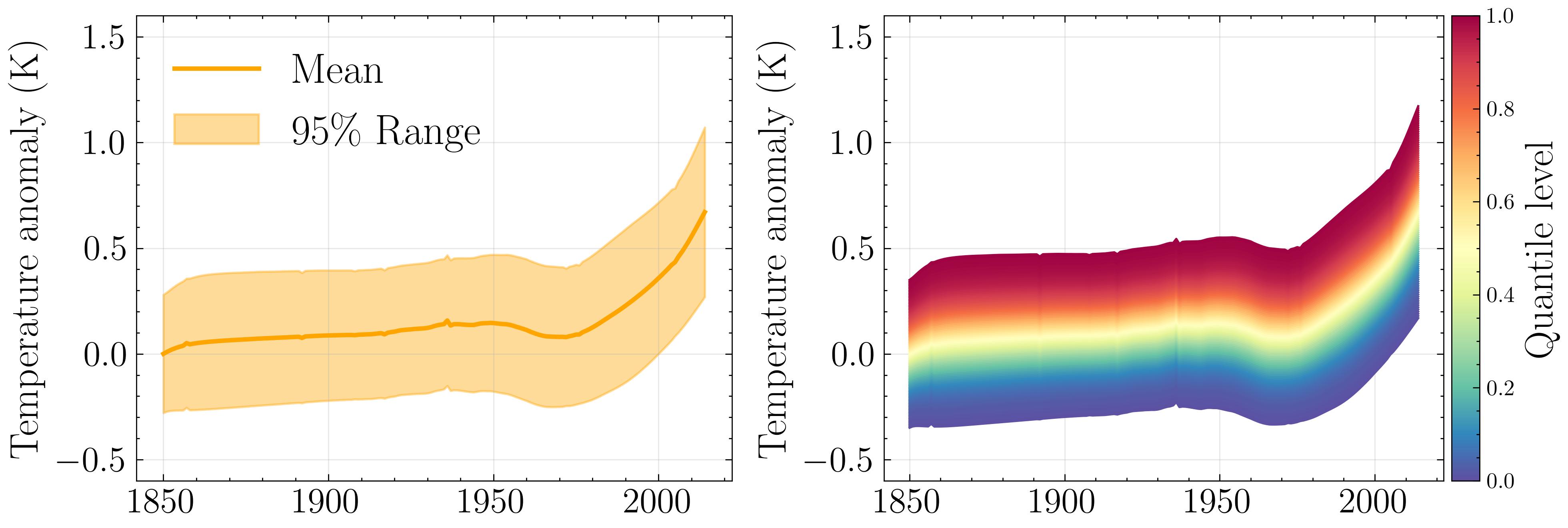}
    \vspace*{-2em}
    \caption{\textbf{Left:} FaIRGP posterior mean and 95\% credible range. \textbf{Right:} FaIRGP posterior quantiles. The posterior distribution $\sfT(t)|\bT$ is conditioned on the same 10 random data points used in Figure~\ref{fig:fairgp-posterior-sample-paths}.}
    \label{fig:mean-and-quantiles}
    \vspace*{-2em}
\end{figure}

Using the same data, we can also update the GP prior over the radiative forcing with temperature data to formulate a posterior distribution over the forcing. Indeed, because the radiative forcing $\sfF(t)$ is also GP, it is jointly Gaussian with the temperature $\sfT(t)$, and thus admits a cross-covariance function that we can derive following
\begin{equation}\label{eq:cross-covariance-F-T}
    k_{\sfF\sfT}(t, t') := \operatorname{Cov}(\sfF(t), \sfT(t')) = \sum_{i=1}^k\frac{q_i}{d_i}\int_0^{t'}K(t, s')e^{-(t'-s')/d_i}\d s'. 
\end{equation}



Intuitively, this cross-covariance function plays a mediation role between the temperature and radiative forcing domains: it translates how variations in surface temperatures correlate with variations in radiative forcing under our model. Using this cross-covariance, it is therefore possible to update our prior over $\sfF(t)$ with temperature data $\bT$, resulting in a posterior distribution over the radiative forcing given by
\begin{equation}\label{eq:posterior-forcing}
    \left\{
    \begin{aligned}
        \begin{split}
            & \sfF(t)\mid \bT \sim \GP(\bar m_\sfF, \bar k_\sfF) \\
            & \bar m_\sfF(t) = F(t) + k_{\sfF\sfT}(t, \bt)(\bK + \sigma^2 \bGamma)^{-1}(\bT - m_\sfT(\bt)) \\
            & \bar k_\sfF(t, t') = K(t, t') - k_{\sfF\sfT}(t, \bt)(\bK + \sigma^2 \bGamma)^{-1}k_{\sfF\sfT}(t, \bt)^\top.
        \end{split}
    \end{aligned}
    \right.
\end{equation}

As for temperatures, the posterior mean $\bar m_\sfF(t)$ corrects the FaIR forcing model $F(t)$ with a data-informed correction term. The mean correction term again regresses out the discrepancy between $\bT$ and the FaIR temperature response onto emissions, but converts this temperature discrepancy signal into a forcing discrepancy signal through the cross-covariance function $k_{\sfF\sfT}$. The uncertainty quantification given by the posterior covariance $\bar k_F(t, t')$ also adjusts the forcing prior covariance $K(t, t')$ to reduce the variability of the stochastic forcing response when predicting over emission trajectories similar to the ones from the training data $\bE$. Figure~\ref{fig:sample-path-F-posterior} shows examples of sample paths from the forcing posterior when updated with global mean surface temperature data.

Finally, we can verify that solving the thermal response stochastic differential equation for the posterior forcing $\sfF(t)\mid \bT$ yields the posterior temperature $\sfT(t)\mid\bT$ as the solution. Therefore, the forcing posterior and the temperature posterior are consistent with each other. It is important to note that the posterior forcing correction proposed in (\ref{eq:posterior-forcing}) may correct the forcing prior with sources of temperature discrepancies that are not imputable to the forcing. Indeed, beyond differences between the FaIR and the climate model forcing response, other aspects of the climate system, such as the climate feedback parameter or the ocean inertia, can also contribute to discrepancies between $\bT$ and the FaIR temperature response. Because the posterior correction reports all of the temperature response discrepancies onto the forcing, there is a risk of simply compensating errors, which needs cautious examination.

\begin{figure}[h]
    \centering
    \vspace*{-1em}
    \includegraphics[width=0.6\linewidth]{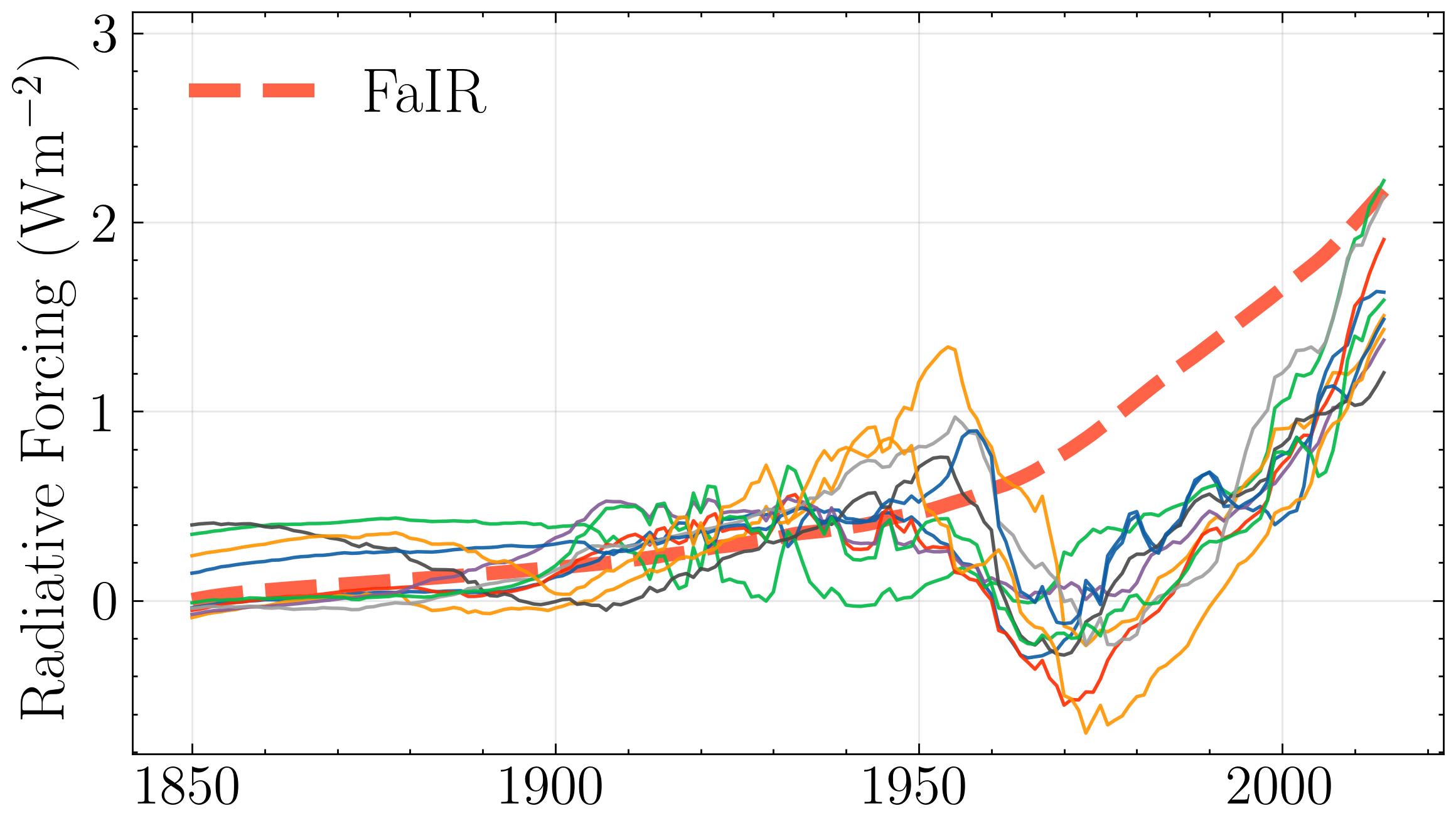}
    \vspace*{-1em}
    \caption{Example of sample paths from the radiative forcing posterior $\sfF(t)|\bT$ conditioned on the same 10 random data points used in Figure~\ref{fig:fairgp-posterior-sample-paths}. In contrast with the forcing prior sample paths (see Figure~\ref{fig:sample-paths-from-sfF}), the posterior sample paths are informed by temperature data and deviate from the FaIR forcing response (dashed red).}
    \vspace*{-2em}
    \label{fig:sample-path-F-posterior}
\end{figure}

\subsection{Parameters tuning}\label{section:parameters-tuning}

Access to a closed-form probability density for the prior allows us to tune the model parameters using a maximum likelihood strategy. To illustrate this, consider again the data $\bt, \bE, \bT$ from the previous section. We can derive a closed-form expression for the \emph{marginal log-likelihood} $\log p(\bT|\bE, \bt)$ of the data under our FaIRGP prior, where the exact formula can be found in Appendix~\ref{appendix:analytical-expressions}. The marginal log-likelihood represents (the logarithm of) how likely is the simulated data $\bt, \bE, \bT$ to occur under the FaIRGP prior we have specified. Hence, a well-specified prior yields a greater marginal log-likelihood. This motivates a tuning strategy for the model parameters that seeks parameter values that maximise $\log p(\bT|\bE, \bt)$.

This can be achieved by using a gradient-based optimisation algorithm to maximise the marginal log-likelihood $\log p(\bT|\bE, \bt)$ with respect to the model parameters. Namely, let $\theta = \{\sigma, \sigma_\sfF, \ell_1, \ldots, \ell_d\}$ be the internal variability standard deviation and the kernel $K(t, t')$ parameters we wish to tune. We can compute the direction in which a small shift in $\theta$ can be taken to increase $\log p(\bT|\bE, \bt)$. This direction corresponds to the gradient of the marginal log-likelihood with respect to the model parameters. By iteratively performing this operation, each \emph{gradient step} brings us closer to an optimal choice of model parameters $\theta$ that maximises $\log p(\bT|\bE, \bt)$. This procedure is known in machine learning as \emph{gradient descent} and illustrated in Algorithm 1. The parameter $\eta > 0$ is typically a small positive number called the learning rate which controls the size of the steps taken in the direction of the gradient at each step. A detailed discussion of its role and of the parameter tuning procedure with gradient descent is provided in Appendix~\ref{appendix:analytical-expressions}.
\begin{algorithm}[h]
    \caption{Maximum likelihood tuning of $\theta$ with gradient descent}
\begin{algorithmic}[1]
    \STATE {\bfseries Input:} $\bt, \bE, \bT, \theta_0, \eta > 0, \texttt{n\_steps}$
    \FOR{$t\in\{1, \ldots, \texttt{n\_steps}\}$}
        \STATE Compute $\log p(\bT|\bE, \bt)$
        \STATE Take $\theta_t \leftarrow \theta_{t-1} + \eta \nabla_\theta\log p(\bT|\bE, \bt)$
    \ENDFOR
    \STATE {\bfseries Return:} $\theta_{\texttt{n\_steps}}$
\label{alg:test-again}
\end{algorithmic}
\end{algorithm}

Finally, note that whilst we limit ourselves to tuning the GP and internal variability parameters in this work, we may also want to tune FaIR parameters such as $d_i$, $q_i$, the forcing model coefficients $\alpha_{\log}^\chi, \alpha_{\mathrm{lin}}^\chi, \alpha_{\mathrm{sqrt}}^\chi$ or even the temperature response parameters $\bA, \bb$. This can be achieved using the same procedure by including these parameters in $\theta$.


\subsection{Spatial FaIRGP}
Whilst informative, the global mean surface temperature fails to capture the difference in exposure of world regions to a changing climate. It is therefore a necessity to invest efforts in obtaining spatially-resolved climate projections. We propose an extension of FaIRGP to emulate spatially-resolved surface temperature maps that grounds itself on a \emph{pattern scaling} prior.

Pattern scaling is a well-established technique to model changes in local surface temperature as a function of changes in global mean surface temperature~\cite{santer1990developing, mitchell1999towards, huntingford2000analogue, tebaldi2014pattern}. It consists of a simple scaling of a fixed spatial pattern by global mean temperature changes. Whilst very simple, such approach is supported by findings that regional changes in temperature scale robustly with global temperature~\cite{seneviratne2016allowable}, and has been successfully used in existing spatial temperature emulation models~\cite{beusch2020emulating, beusch2020crossbreeding, beusch2022emission, link2019fldgen, goodwin2020computationally}.


Formally, pattern scaling assumes that for a given spatial location $x$, the local temperature response $T(x,t)$ is given by
\begin{equation}
    T(x,t) = \beta^{(x)} T(t) + \beta^{(x)}_0,
\end{equation}
with regression coefficient $\beta^{(x)}$ and intercept $\beta^{(x)}_0$. These coefficients are typically obtained by fitting independent local linear regression models of global temperature $T(t)$ onto local temperature $T(x,t)$. If we substitute the deterministic global temperature response $T(t)$ with its GP version $\sfT(t)$, we obtain a local FaIRGP prior temperature response given by
\begin{equation}\label{eq:spatial-fairgp-prior}
    \sfT(x,t) \sim \GP\left(\beta^{(x)} m_\sfT + \beta^{(x)}_0, (\beta^{(x)})^2 k_\sfT\right),
\end{equation}
which admits the local pattern scaling temperature response for its mean, and a locally scaled covariance. 


As for the global response, this local prior can be updated with local temperature data to obtain a posterior temperature response at spatial location $x$. This gives FaIRGP the flexibility to learn from data how to deviate from a fixed spatial pattern, in order to better account for climate model simulated data.


\section{Experimental Setup}

In this section, we introduce the dataset used in our emulation experiments, the baseline models we benchmark FaIRGP against, and the evaluation metrics. Code and data to reproduce experiments are publicly available\footnote{\url{https://github.com/shahineb/FaIRGP}}.

\subsection{Dataset description}

The data is obtained from the ClimateBench v1.0~\cite{watsonparris2021climatebench} climate emulation benchmark dataset. ClimateBench v1.0 proposes a curated dataset of local annual mean surface temperature ($\sim$2$^\circ$ horizontal resolution), paired with annual emissions for four of the main anthropogenic forcing agents: carbon dioxide (CO\textsubscript{2}), methane (CH\textsubscript{4}), sulfur dioxide (SO\textsubscript{2}) and black carbon (BC). The temperature data is generated from the latest version of the Norwegian Earth System Model (NorESM2)~\cite{seland2020overview} as part of the sixth coupled model intercomparison project (CMIP6)~\cite{eyring2016overview}. The emission data is constructed from the Community Emissions Data System~\cite{hoesly2018historical}, and corresponds to the input data used to drive the original NorESM2-LM simulations.

The dataset we use includes paired time series of global annual emissions and local annual mean surface temperature for five experiments. We include the CMIP6 \textit{historical} experiment, and four experiments corresponding to different possible shared socio-economic pathways (SSPs)~\cite{riahi2017shared} from the ScenarioMIP protocol~\cite{o2016scenario}: \textit{SSP126}, \textit{SSP245}, \textit{SSP370} and \textit{SSP585}. These scenarios are designed to span a range of emissions trajectories corresponding to plausible mitigation scenarios and end-of-century forcing possibilities. Table~\ref{table:dataset} gives a description of the experiments and the periods covered.

In all experiments, we use as inputs for FaIRGP the global annual cumulative emissions of CO\textsubscript{2} and global annual emissions of CH\textsubscript{4}, SO\textsubscript{2} and BC. In the experiments of Section~\ref{section:global-experiments}, we focus on emulating global annual mean surface temperature. Therefore, we preprocess the temperature data by computing the global annual mean surface temperatures using an area-weighted mean of the local annual mean surface temperatures data. In Section~\ref{section:spatial-experiments} we focus on emulating local annual mean surface temperatures, thus no preprocessing is required.

Our experiments are effectively conducted on data from a single climate model, and therefore do not reflect the diversity of response of different climate models to a common scenario. However, we argue they still provide a valid assessment of an emulator's ability to reproduce climate models outputs. Indeed, whilst different climate models exhibit variations in their response to a similar scenario, it has been shown that SCMs are still capable of capturing the diversity of forcing responses spanned~\cite{smith2018fair}. Therefore, there are reasons to believe that insights drawn from experiments where emulators are tuned on NorESM2 data should carry over to experiments where emulators are tuned on data from different climate models.

\begin{table}[t]
\centering
    \caption{Details of the experiments from ClimateBench v1.0~\protect\cite{watsonparris2021climatebench} used in the dataset.}
    \vspace*{-0.5em}
    \begin{tabularx}{\textwidth}{lcX}\toprule
     Experiment & Period & Description \\ \midrule
        \textit{historical} & 1850-2014 & A simulation using historical emissions  of all forcing agents designed to recreate  the historically observed climate \\
        \textit{SSP126} & 2015-2100 & A high ambition scenario designed to produce significantly less than 2° warming by 2100 \\
        \textit{SSP245} & 2015-2100 & A medium forcing future scenario aligned with current trajectories. \\
        \textit{SSP370} & 2015-2100 & A medium-high forcing scenario with high emissions of near-term climate forcers such as methane and aerosols \\
        \textit{SSP585} & 2015-2100 & A high forcing scenario leading to a large 8.5Wm\textsuperscript{-2} forcing in 2100 \\
    \bottomrule
    \end{tabularx}
    \label{table:dataset}
    \vspace*{-4em}
\end{table}

\subsection{Baseline emulators}

To develop intuition on the benefit of combining FaIR with a GP, with propose to compare our model to the temperature anomaly projections obtained with (\textit{i}) FaIR only and (\textit{ii}) with a purely data-driven plain GP regression model only. By comparing to the emulation with FaIR, we hope to highlight the benefit that there is to combine the flexibility of a data-driven approach to an impulse response model, and answer the question \textit{\enquote{Does introducing a GP in FaIR add value?}}. On the other hand, by comparing FaIRGP with a plain GP regression model, we hope to demonstrate the importance of having a robust physical prior underlying a data-driven approach, and answer the question \textit{\enquote{Does a FaIR prior in a GP add value?}}. Both baseline models take as inputs greenhouse gas and aerosol emissions data, and predict annual mean surface temperature anomalies with respect to the preindustrial period.

For FaIR, we use parameter values that have been tuned against NorESM2 simulations~\cite{leach2021fairv2}. Therefore, the FaIR model we use in experiments is fully deterministic with fixed parameter values. 

The plain GP model is entirely physics-free, and is specified with a zero mean and an anisotropic Matérn-3/2 kernel. This is exactly the same kernel $\rho$ we use in our prior over the forcing in FaIRGP, which takes as inputs the global annual cumulative emissions of CO\textsubscript{2} and global annual emissions of CH\textsubscript{4}, SO\textsubscript{2} and BC. This kernel is to contrast with the physics-informed kernel $k_\sfT$ used for temperatures in FaIRGP. As it is standard practice in GP regression, we assume Gaussian noise with constant variance over the data. The noise variance and the kernel parameters are tuned using marginal likelihood maximisation by gradient descent.

\subsection{Evaluation metrics}

We use two kinds of metrics to evaluate the predicted probabilistic surface temperatures: deterministic metrics, which compare only the posterior mean prediction to the ClimateBench temperature data, and probabilistic metrics, which evaluate the entire posterior probability distribution against ClimateBench temperature data. Table~\ref{table:metrics} provides a brief description of the metrics used.

\begin{table}[h]
\centering
    \caption{Details on evaluation metrics used in experiments.}
    \begin{tabular}{llcc}\toprule
    & Notation & Description & Best when \\ \midrule
    \multirow{3}{*}{\textit{Deterministic}}
        & RMSE & Root mean square error & close to 0 \\
        & MAE & Mean absolute error & close to 0 \\
        & Bias & Mean of (prediction $-$ groundtruth) & close to 0 \\ \midrule
    \multirow{3}{*}{\textit{Probabilistic}}
        & LL & Log-likelihood & higher \\
        & Calib95 & 95\% calibration score & close to 95\% \\
        & CRPS & Continous ranked probability score & close to 0 \\
    \bottomrule
    \end{tabular}
    \label{table:metrics}
\end{table}

The log-likelihood (LL) score evaluates the log probability of the groundtruth temperature data under the predictive posterior probability distribution predicted by our model. Therefore, greater LL means that the predicted distribution is a good fit for the test data. The 95\% calibration score (Calib95) computes the percentage of groundtruth temperature data that fall within the 95\% credible interval of the predicted posterior distribution. Therefore, if the predicted probability distribution is well calibrated, this score should be close to 95\%. 

The Continuous Ranked Probability Score (CRPS) is an extension of the RMSE for probability distributions, which measures the discrepancy between the cumulative distribution functions (cdfs) of the predicted posterior distribution and the empirical cdf of the test data. Namely, if $\Phi$ denotes the cdf of the predicted posterior distribution and $y$ is a scalar test data point, the CRPS between $\Phi$ and $y$ is given by
\begin{equation}
    \operatorname{CRPS}(\Phi, y) = \int_{-\infty}^{+\infty} \big(\Phi(x) - \mathbbm{1}_{[y, +\infty)}(x)\big)^2\d x, 
\end{equation}
where $\mathbbm{1}_{[y, +\infty)}(x)$ is the indicator function which takes value $0$ if $x < y$ and $1$ if $x \geq y$.

When evaluating predictions of local annual mean surface temperatures, we compute global mean metrics using an area-weighted mean that accounts for the decreasing grid-cell area toward the poles. Namely, we take 
\begin{equation}\label{eq:spatially-weighted-score}
    \langle \operatorname{Score}\rangle  = \frac{1}{N_\text{lon}\sum_{i=1}^{N_\text{lat}}w_i}\sum_{i=1}^{N_\text{lat}}\sum_{j=1}^{N_\text{lon}}w_i\operatorname{Score}_{i,j},
\end{equation}
where $w_i = \cos(\operatorname{lat}(i))$.

\newpage
\section{Application: global surface temperatures emulation}\label{section:global-experiments}

In this section, we benchmark FaIRGP against baseline models for the task of emulating mean global surface temperatures over SSP scenarios. We first briefly illustrate how the model concretely applies to an example. Then, we evaluate the emulated global temperatures trajectories when the model is trained on historical and SSP data, and when trained on historical data only. Finally, we probe the potential of FaIRGP to infer the top-of-atmosphere radiative forcing.

\subsection{Example of global temperature emulation with FaIRGP}

We begin with an example, and propose in Figure~\ref{fig:global-illustrative-example} a concrete illustration of global mean surface temperature emulation for \textit{SSP370} with FaIRGP, that parallels the posterior mean and covariance expressions from (\ref{eq:posterior-temperature-mean}) and (\ref{eq:posterior-temperature-covar}).

The prior temperature response over \textit{SSP370} admits as its mean the FaIR response, but with an additional layer of uncertainty quantification that arises from the GP. We then use temperature data from scenarios $\cD_\text{train} = \{$\textit{historical, SSP126, SSP245, SSP585}$\}$ as training data to learn a posterior correction. The posterior correction allows deviation from the prior --- both in mean and variance --- to provide a better fit to the training data. Finally, by linearly adding the posterior correction to the prior, we obtain a posterior temperature response over \textit{SSP370}.

\begin{figure}[h]
    \centering
    \hspace*{-1.5em}
    \includegraphics[width=1.05\linewidth]{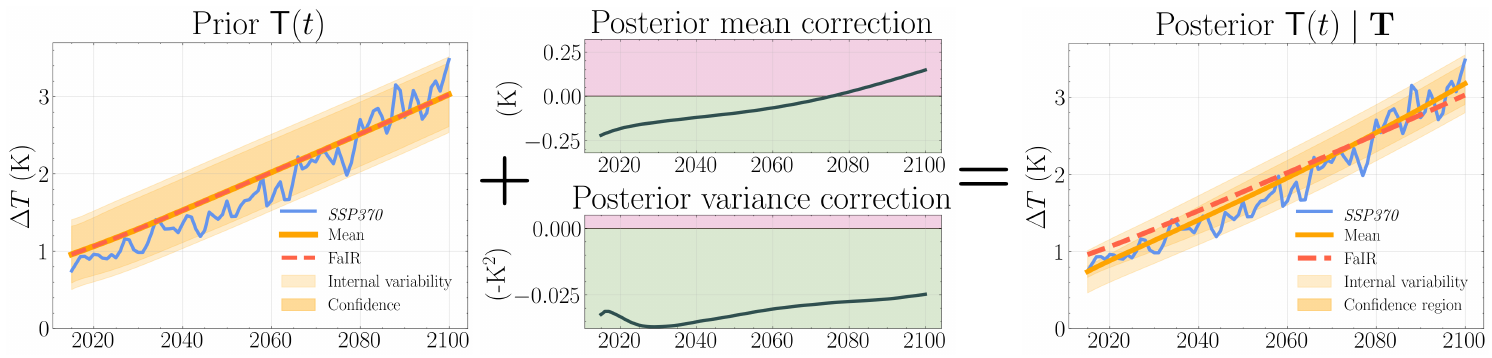}
    \caption{Emulation of NorESM2-LM \textit{SSP370} global mean surface temperature anomaly. \textbf{Left:} FaIRGP prior temperature response over \textit{SSP370}. The prior mean response exactly corresponds to the FaIR temperature response, but with prior uncertainty quantification. \textbf{Middle:} Posterior correction for the mean and variance that learns from data how to deviate from the prior. The pink area corresponds to deviations toward greater values and the green area deviations toward smaller values. \textbf{Right:} FaIRGP posterior temperature response when updated with the other scenarios. The posterior response provides a sounder fit to the climate model data, with better calibrated uncertainty. Shading indicates a 5\%-95\% range.}
    \vspace*{-2em}
    \label{fig:global-illustrative-example}
\end{figure}

\subsection{Shared socio-economic pathways emulation}\label{subsection:global-ssp-experiment}

In this experiment, we consider the experiment dataset given by $\cD = \{$\textit{historical, SSP126, SSP245, SSP370, SSP585}$\}$. We iteratively remove one SSP experiment from the dataset to construct a training set (e.g.\ retain \textit{SSP245} to obtain $\cD_\text{train} = \{$\textit{historical, SSP126, SSP370, SSP585}$\}$) and use it to train FaIRGP and the baseline GP model. We run predictions over the retained SSP experiment and evaluate the emulated global temperatures against NorESM2-LM data. We find that on average, FaIRGP outperforms the baseline models in every evaluation metric, as reported in Table~\ref{table:global-ssp-scores}. These results suggest that FaIRGP provides a better emulation of the global temperature response to anthropogenic forcing compared to both FaIR and the plain GP model. We include in Appendix~\ref{appendix:comparison-with-climatebench} a comparison with the GP emulator from ClimateBench~\cite{watsonparris2021climatebench}.

\begin{table}[t]
    \centering
    \caption{Scores of baseline emulators and FaIRGP for the task of emulating SSPs global temperatures over 2015-2100 period; the test experiment is retained from the training data and only used for evaluation; the best emulator scores for each test experiment are highlighted in bold; $\uparrow\!/\!\downarrow$ indicates higher/lower is better; we report 1 standard deviation over the mean scores; $\dagger$ indicates our proposed method.}
    \vspace*{-0.5em}
        \resizebox{\linewidth}{!}{
        \begin{tabular}{llcccccc}
        \toprule
         \shortstack{Test \\ experiment}   &  Emulator    &   RMSE\small{$\;\downarrow$} &  MAE\small{$\;\downarrow$} &  Bias &  LL\small{$\;\uparrow$} & Calib95 & CRPS\small{$\;\downarrow$} \\
        \midrule
        \multirow{3}{*}{\textit{SSP126}} & FaIR &  0.17 &  0.13 &  0.06 &       - &       -  & - \\
            & Plain GP &  0.18 &  0.14 &  0.08 &  0.23 &     1.0  & 0.10\\
            & FaIRGP$^{\dagger}$ &  \textbf{0.15} &  \textbf{0.12} &  \textbf{0.02} &  \textbf{0.48} &   \textbf{0.97}  & \textbf{0.08}\\\thinrule
        \multirow{3}{*}{\textit{SSP245}} & FaIR &  0.14 &  0.12 &  0.06 &     - &       - & - \\
            & Plain GP &  \textbf{0.09} &  \textbf{0.07} &  0.03 &  0.67 &     \textbf{1.0} & \textbf{0.06} \\
            & FaIRGP$^{\dagger}$ &  \textbf{0.09} &  0.08 & \textbf{-0.01} &  \textbf{0.71} &     \textbf{1.0} & \textbf{0.06}\\\thinrule
        \multirow{3}{*}{\textit{SSP370}} & FaIR &  0.22 &  0.19 &  0.10 &       - &       - & - \\
            & Plain GP &  0.22 &  0.17 & -0.16 &  0.12 &     1.0  & 0.12 \\
            & FaIRGP$^{\dagger}$ &  \textbf{0.17} &  \textbf{0.14} &  \textbf{0.06} &  \textbf{0.39} &    \textbf{0.94} & \textbf{0.10} \\\thinrule
        \multirow{3}{*}{\textit{SSP585}} & FaIR &  0.28 &  0.22 & \textbf{-0.07} &      - &  - & - \\
            & Plain GP &  0.30 &  0.21 & -0.12 &   0.16 &     \textbf{1.0} & 0.14 \\
            & FaIRGP$^{\dagger}$ &  \textbf{0.21} &  \textbf{0.17} & -0.08 &  \textbf{0.18} &   0.88  & \textbf{0.12} \\\midrule
        \multirow{3}{*}{\textit{Mean}} & FaIR &  0.20\scriptsize{$\pm$0.06} &  0.16\scriptsize{$\pm$0.05} &  0.04\scriptsize{$\pm$0.08} &     - &       -  & -\\
            & Plain GP &  0.20\scriptsize{$\pm$0.09} &  0.15\scriptsize{$\pm$0.06} & -0.04\scriptsize{$\pm$0.11} &  0.30\scriptsize{$\pm$0.25} &     1.0\scriptsize{$\pm$0.0}  & 0.11\scriptsize{$\pm$0.03} \\
            & FaIRGP$^{\dagger}$ &  \textbf{0.16\scriptsize{$\pm$0.05}} &  \textbf{0.12\scriptsize{$\pm$0.04}} & \textbf{-0.01\scriptsize{$\pm$0.06}} &  \textbf{0.44\scriptsize{$\pm$0.22}} &   \textbf{0.95\scriptsize{$\pm$0.05}} & \textbf{0.09\scriptsize{$\pm$0.02}}\\
        \bottomrule
        \end{tabular}
        }
    \label{table:global-ssp-scores}
    \vspace*{-1em}
\end{table}
\begin{figure}[H]
    \centering
    \hspace*{-3em}
    \includegraphics[width=1.1\linewidth]{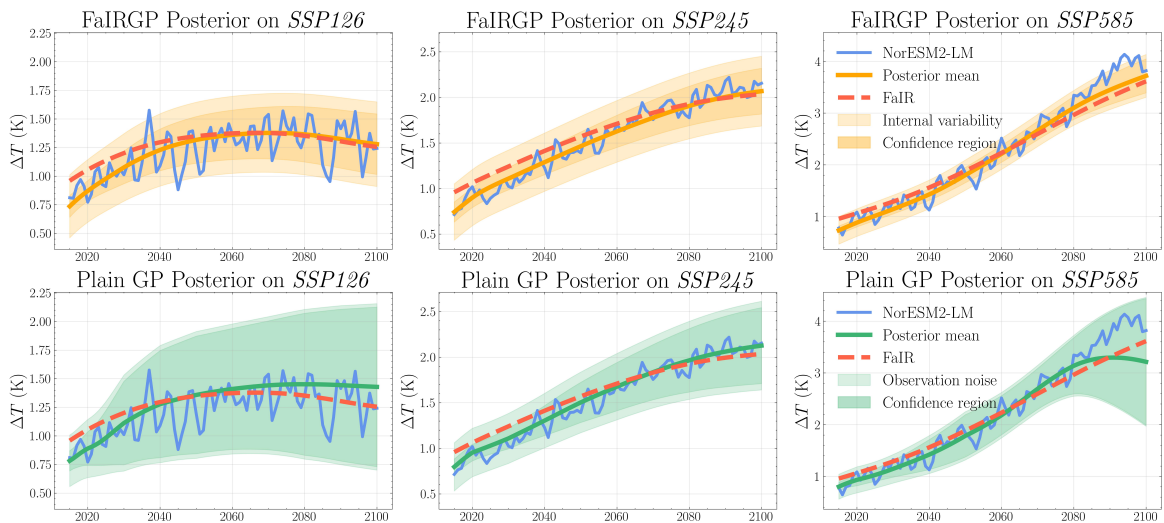}
    \caption{Emulated global mean surface temperature anomaly over 2015-2100 period with our FaIRGP emulator (top) and with a baseline GP emulator (bottom) for a low-forcing future scenario (\textit{SSP126}), a medium forcing future scenario (\textit{SSP245}) and a high forcing future scenario (\textit{SSP585}).}
    \label{fig:ssp-global-plots}
    \vspace*{-2em}
\end{figure}

Figure~\ref{fig:ssp-global-plots} shows that FaIRGP provides a better temperature projection than FaIR in the near future --- over the 2015-2050 period --- for all scenarios. This is because being informed by data grants FaIRGP the flexibility to deviate from the impulse response prior on which it hinges, and provide predictions that are better aligned with the historical and near future data from the training set. However, when further away from the training data, over the 2080-2100 period, FaIRGP reverts back to the prior behaviour of FaIR. This suggests that FaIRGP is better suited for emulating global temperature response in the near future, and is at least as good as FaIR for longer term emulation.

The plain GP model provides excellent predictions on \textit{SSP245}, and even outperforms FaIRGP by a slight margin on deterministic metrics for this scenario. This is because GPs are notorious for excelling at interpolation tasks, and the \textit{SSP245} scenario is a medium forcing scenario that perfectly lies within the range of the other low and high forcing scenarios used in the training data. However, when the evaluation scenario lies outside of the range of the training data, the plain GP model predictions will simply revert to the prior, and are therefore much less reliable. This is an important drawback of purely data-driven emulation, which by essence are interpolation models, and struggle to extrapolate outside the range of the training set. In Figure~\ref{fig:ssp-global-plots}, this is particularly evident in the plain GP prediction on \textit{SSP585}, where the model struggles at the end of the century due to unprecedented forcing levels that have not been seen in the training data. FaIRGP addresses this shortcoming by using FaIR as its prior, and therefore naturally reverts to an impulse response model when the prediction scenario is too distant from the training data.

Finally, FaIRGP provides a tighter and more robust uncertainty quantification over emulated temperatures compared to the plain GP model, which tends to overestimate the variance. This is reflected in Table~\ref{table:global-ssp-scores} by a substantial improvement in LL of FaIRGP against the Plain GP, and suggests that the physics-informed covariance structure of FaIRGP is a sound choice to quantify emulator uncertainty. Further, the model formulated is able to quantify the uncertainty introduced by the internal variability separately from the uncertainty introduced by the emulator design. As shown in Figure~\ref{fig:ssp-global-plots}, the uncertainty due to internal variability dominates in the near future, but the emulator uncertainty becomes predominant toward the end of the century, in accordance with the well-established assessment of \citeA{hawkins2009potential}.

\begin{figure}[h]
    \centering
    \vspace*{-0.5em}
    \includegraphics[width=\linewidth]{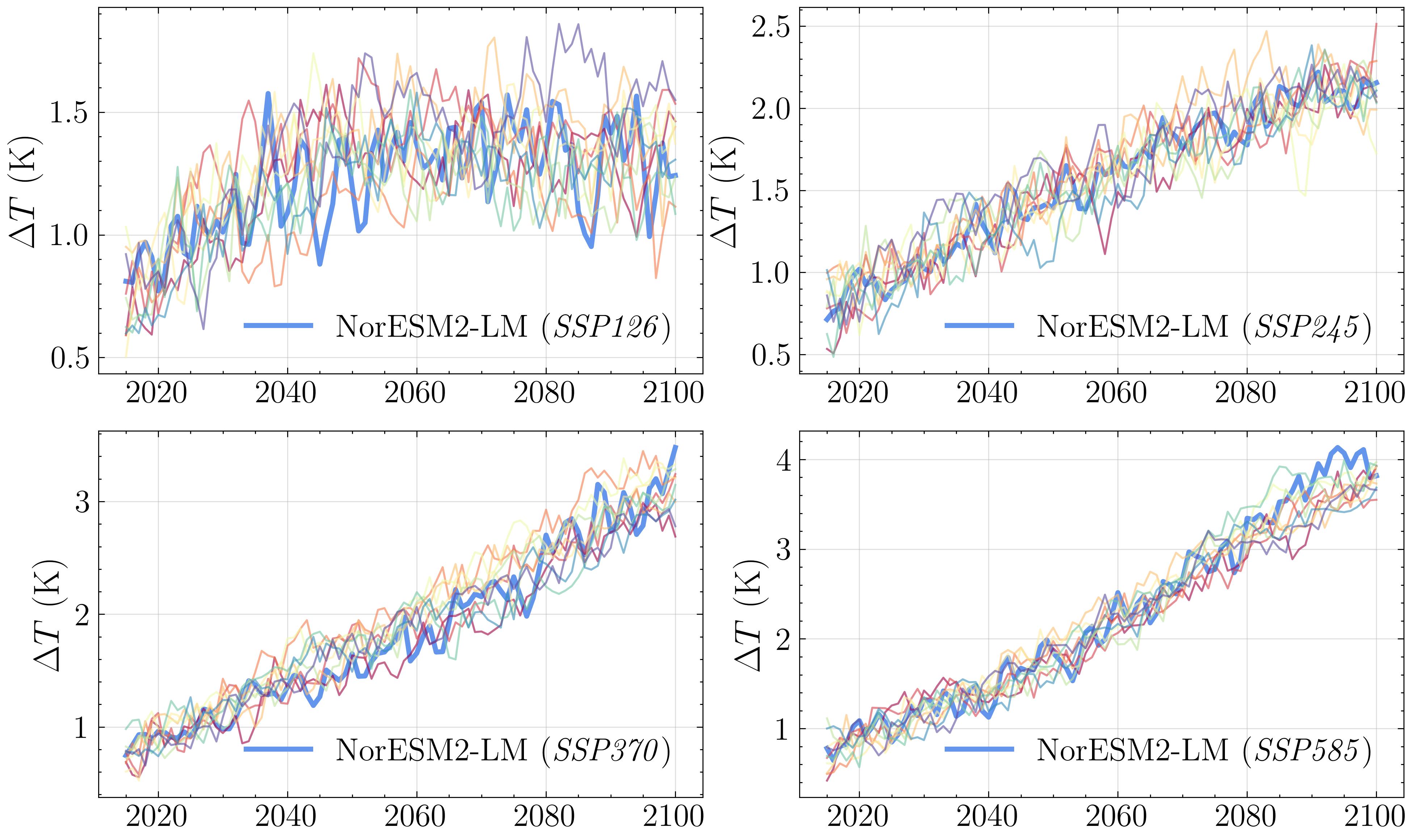}
    \vspace*{-1.5em}
    \caption{Example of sample paths from the FaIRGP posterior for the emulation of \textit{SSP126} (top-left), \textit{SSP245} (top-right), \textit{SSP370} (bottom-left) and \textit{SSP585} (bottom-right) global mean surface temperature anomaly over the 2015-2100 period.}
    \label{fig:sample-paths-SSP-global}
    \vspace*{-2em}
\end{figure}

In Figure~\ref{fig:ssp-global-plots}, the posterior mean lines and the shaded area depict the mean function and the 95\% credible region of the predicted posterior distribution over temperatures using FaIRGP. Since this represents a probability distribution, we can generate sample paths from the posterior GP (which will fall within the shaded area with probability $0.95$). Figure~\ref{fig:sample-paths-SSP-global} displays examples of sample paths from the FaIRGP posterior for each scenario. It shows that the smooth shaded areas hide the specific form of the functions described by the FaIRGP posterior. The sample paths are in fact capable of reproducing the variability of abrupt annual changes in global mean surface temperature anomalies. Further, the sample paths are generally well-concentrated around the NorESM2 simulated temperature anomaly, with the exception of the end of century temperature anomaly under \textit{SSP585} which tends to be underestimated --- consistently with the shaded area from Figure~\ref{fig:ssp-global-plots}. Appendix~\ref{appendix:sample-paths-quizz} provides additional plots of individual sample paths from the FaIRGP posterior over SSPs.

\subsection{Emulation from historical data}

In this experiment, we propose to investigate how well can the emulators generalise when they are solely trained on historical data, and no simulated future data\footnote{with the exception of the FaIR parameters which have been calibrated using NorESM2 outputs for future scenarios.}. We train the plain GP and FaIRGP using the \textit{historical} experiment only, and evaluate predictions on all SSP scenarios. The results are reported in Table~\ref{table:results-historical-only}. We find that FaIRGP outperforms other models in almost every metric, and displays performance similar to FaIR only in mean bias.

\begin{table}[h]
\vspace*{-2em}
\centering
    \caption{Scores of baseline emulators and FaIRGP when models are only trained using the \textit{historical} experiment; scores are evaluated and averaged for the task of emulating \textit{SSP126, SSP245, SSP370, SSP585} global surface temperatures over 2015-2100 period; the best emulator for each metric is highlighted in bold; $\uparrow\!/\!\downarrow$ indicates higher/lower is better; $\dagger$ indicates our proposed method.}
        \begin{tabular}{lcccccc}
        \toprule
        Emulator &   RMSE\small{$\;\downarrow$} &   MAE\small{$\;\downarrow$} &   Bias &   LL\small{$\;\uparrow$} &  Calib95 & CRPS\small{$\;\downarrow$} \\
        \midrule
        FaIR & 0.208 &  0.163 &  0.038 &      - &       - & - \\
        Plain GP &  1.875 &  1.690 & -1.690 & -30.979 &    0.009  & 1.558\\
        FaIRGP$^{\dagger}$   &  \textbf{0.191} &  \textbf{0.140} & \textbf{-0.037} &   \textbf{0.314} &    \textbf{0.948}  & \textbf{0.102} \\
        \bottomrule
        \end{tabular}
    \label{table:results-historical-only}
    \vspace*{-1em}
\end{table}
\begin{figure}[h]
    \centering
    \vspace*{-0.5em}
    \includegraphics[width=0.9\linewidth]{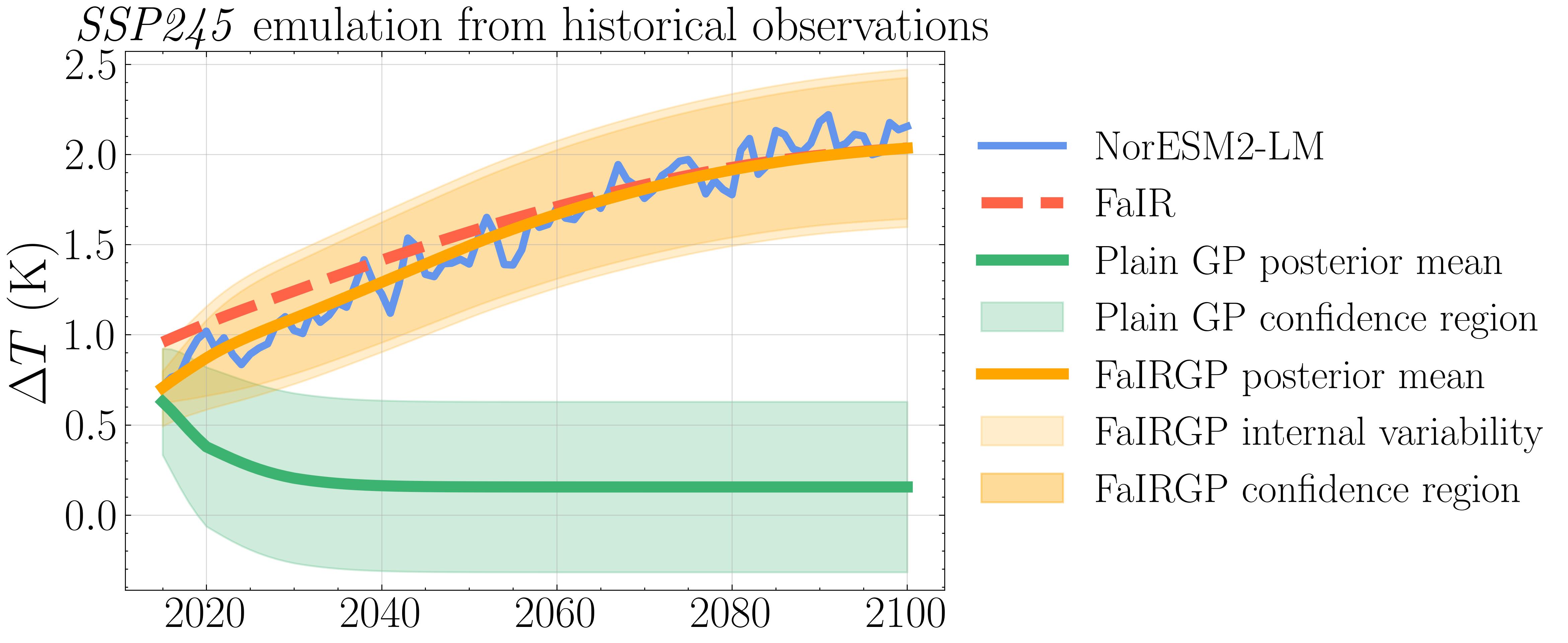}
    \vspace*{-0.5em}
    \caption{Emulated global mean surface temperature anomaly over 2015-2100 for the \textit{SSP245} scenario with emulators trained over historical temperatures only.}
    \vspace*{-2em}
    \label{fig:historical-only-experiment}
\end{figure}

In this context, a purely data-driven method struggles to extrapolate because it cannot marry the new previously unseen emission values with the underlying physical model. Figure~\ref{fig:historical-only-experiment} shows that the plain GP posterior immediately reverts to its prior after 2015, therefore providing uninformative temperature projections. On the other hand, FaIRGP manages to learn from historical data to provide a better fit to the SSP over the 2015-2050 period, and slowly reverts back to the prior behaviour of FaIR toward the end of the century.

This demonstrates the importance of having a robust underlying physical model for an emulator, and further suggests that FaIRGP is a well suited candidate for this task, and can provide meaningful temperature projections based only on historical data.

\subsection{Inferring radiative forcing from temperatures}\label{subsection:global-forcing-posterior}

Whilst FaIRGP is intended to emulate the surface temperature response to changes in emission, its GP nature also allows to make inference of radiative forcing given temperature data as demonstrated in Section~\ref{subsection:posterior-distribution}. In this experiment, we propose to assess how global surface temperature anomaly data can be used to inform an estimate of the effective radiative forcing using FaIRGP. 

We use the complete dataset of global temperature time series from the \textit{historical} and \textit{SSP126, SSP245, SSP370, SSP585} experiments, which are depicted in Figure~\ref{fig:global-forcing-posteriors}. We update the GP prior placed over the radiative forcing $\sfF(t)$ with temperature data to obtain a posterior radiative forcing response which incorporates information from temperature data. We emphasise that the posterior over $\sfF(t)$ does not use any forcing data, only simulated temperature data. As for temperature emulation, the posterior radiative forcing response is given in closed-form by  (\ref{eq:posterior-forcing}). We evaluate the posterior historical radiative forcing response against NorESM2-LM historical top-of-atmosphere radiative forcing.

Figure~\ref{fig:global-forcing-posteriors} illustrates that the posterior forcing response learns from temperature time series how to deviate from the FaIR forcing trend to better account for the NorESM2 simulated data. This is particularly evident between 1960 and 1980 where the posterior reproduces a decrease in forcing to better account for the global cooling trend in that time period, which FaIR struggles to capture. The results reported in Table~\ref{table:results-posterior-forcing} show that FaIRGP helps improve over the FaIR forcing in RMSE and Bias. Whilst FaIR performs better at inferring the stable forcing before 1950, informing the inferred forcing with temperature data with FaIRGP helps better account for the forcing variations after 1950.

\begin{figure}[h]
    \centering
    \includegraphics[width=\linewidth]{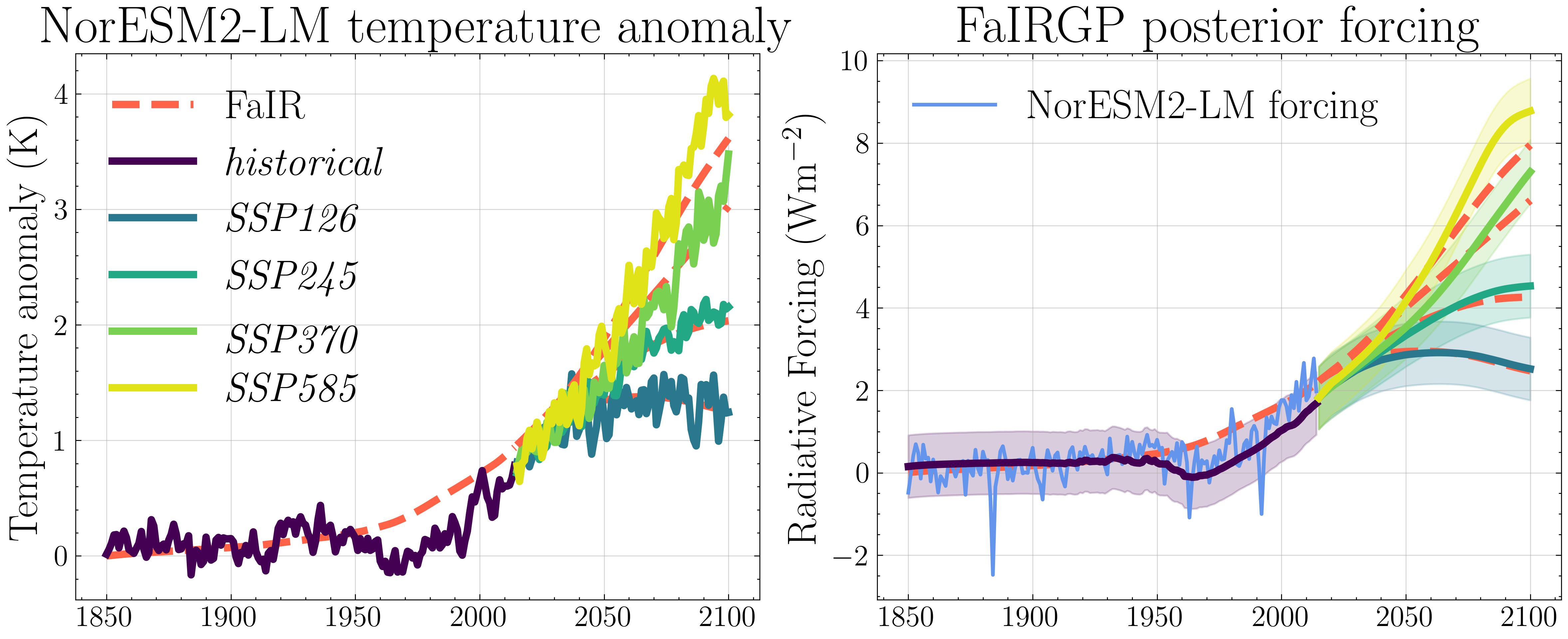}
    \caption{\textbf{Left:} Global temperature anomaly responses used to inform the posterior forcing response; \textbf{Right:} FaIRGP posterior radiative forcing predicted over experiments; the temperature response and radiative forcing obtained with FaIR are plotted in dashed red lines.}
    \label{fig:global-forcing-posteriors}
    \vspace*{-2em}
\end{figure}
\begin{table}[h]
    \centering
    \caption{Scores of FaIR and FaIRGP at predicting the global radiative forcing over historical periods; scores are evaluated against NorESM2-LM top-of-atmosphere radiative forcing; the best emulator scores for each period are highlighted in bold; $\uparrow\!/\!\downarrow$ indicates higher/lower is better; $\dagger$ indicates our proposed method.}
        \vspace*{-0.5em}
        \begin{tabular}{llcccccc}
        \toprule
         Period & Emulator &  RMSE\small{$\;\downarrow$}  &    MAE\small{$\;\downarrow$}  &   Bias  &     LL\small{$\;\uparrow$}  & Calib95 &   CRPS\small{$\;\downarrow$}  \\
        \midrule
        \multirow{2}{*}{1850-1950} & FaIR &  \textbf{0.396} &  \textbf{0.261} &  \textbf{-0.008} &      - &       - &      - \\
                  & FaIRGP$^{\dagger}$ &  0.421 &  0.285 &  0.041 & -0.797 &    0.99 &  0.25 \\ \thinrule
        \multirow{2}{*}{1950-2014} & FaIR &  0.622 &  \textbf{0.461} &  0.350 &      - &       - &      - \\
                  & FaIRGP$^{\dagger}$ &  \textbf{0.568} &  0.466 & \textbf{-0.315} & -0.925 &   0.984 &  0.328 \\ \thinrule
        \multirow{2}{*}{1850-2014} & FaIR &  0.496 &  \textbf{0.339} &  0.129 &      - &       - &      - \\
                  & FaIRGP$^{\dagger}$ &  \textbf{0.484} &  0.357 & \textbf{-0.101} &   -0.847 &   0.988 &  0.281 \\
        \bottomrule
        \end{tabular}
    \label{table:results-posterior-forcing}
    \vspace*{-1.5em}
\end{table}

FaIRGP therefore proves to be useful to infer global radiative forcing trends informed by surface temperatures. This is possible because our prior is specified within an energy balance model which explicitly connects the changes in surface temperatures to the changes in radiative forcing. On the contrary, this would not be possible with the plain GP model because it ignores forcing dynamics.

\subsection{FaIR multi-model uncertainty in FaIRGP}

In this section, we explore how the uncertainty bounds introduced by the FaIR model compare with those of the FaIRGP uncertainty. FaIR has been widely employed to assess the uncertainty surrounding global temperature projections. The methodology employed by FaIR for generating probabilistic projections is rooted in a constrained ensembling process. This involves specifying prior distributions over model parameters (such as carbon cycle sensitivity, forcing model sensitivity and equilibrium climate sensitivity). These prior distributions are determined from the inter-model variability. Then, a large ensemble of climate projections is drawn using these distributions, and is subsequently constrained based on the likelihood of projections aligning with the simulated present-day level and rate of anthropogenic warming~\cite{leach2021fairv2}. The resulting \emph{constrained ensemble} of projections allows to quantify the effect of \emph{FaIR multi-model uncertainty} on temperature projections.


\begin{figure}[h]
    \centering
    \vspace*{-0.5em}
    \includegraphics[width=0.96\linewidth]{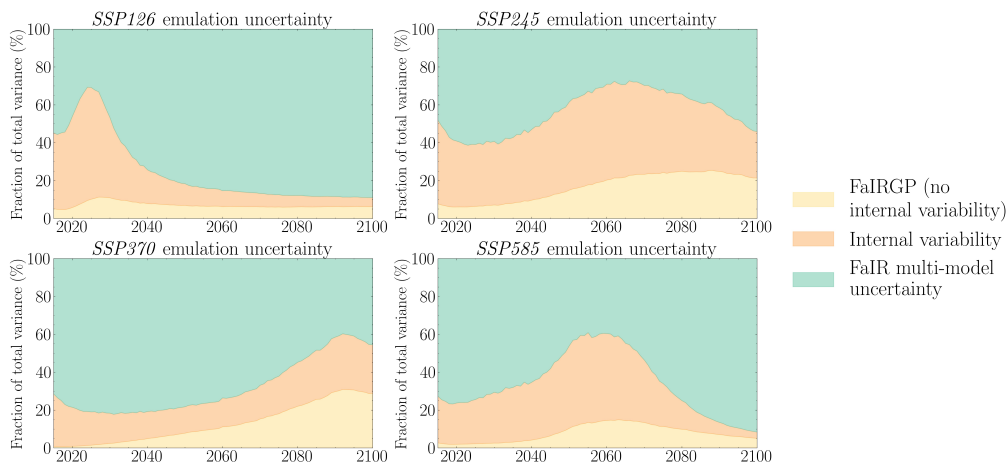}
    \vspace*{-1em}
    \caption{Proportional stacked area charts of the contribution to total variance of the FaIRGP uncertainty without internal variability (yellow), the internal variability modelled in FaIRGP as in Section~\ref{subsection:fairgp-internal-variability} (orange) and the FaIR multi-model uncertainty (green), for the emulation of the global annual mean surface temperature anomaly under \textit{SSP126} (top-left), \textit{SSP245} (top-right), \textit{SSP370} (bottom-left) and \textit{SSP585} (bottom-right).}
    \vspace*{-2em}
    \label{fig:total-variance-fraction}
\end{figure}

Following the protocol of \citeA{leach2021fairv2}, we generate a constrained ensemble of 1000 FaIR surface temperature anomaly projections. For each ensemble member, we then generate a FaIRGP projection. The total variability of this projection can be decomposed into two components: \textit{(i)} the average variability introduced by FaIRGP across ensemble members and \textit{(ii)} the variability in the mean FaIRGP response across ensemble members. The first component corresponds to the variability that FaIRGP's stochasticity introduces on the deterministic FaIR temperature response, averaged across ensemble member, and can be decomposed into the variability introduced by FaIRGP without internal variability, and the contribution of internal variability alone. For this reason component \textit{(i)} is viewed as the uncertainty imputable to FaIRGP. The second component corresponds to the variability that sampling across the constrained ensemble introduces in the FaIRGP mean response. For this reason, component \textit{(ii)} is viewed as the uncertainty imputable to the FaIR multi-model uncertainty. Figure~\ref{fig:total-variance-fraction} shows how the different sources of uncertainty contribute to the total variance when emulating surface temperature anomalies for SSPs.

\begin{figure}[t]
    \centering
    \vspace*{-1em}
    \includegraphics[width=0.96\linewidth]{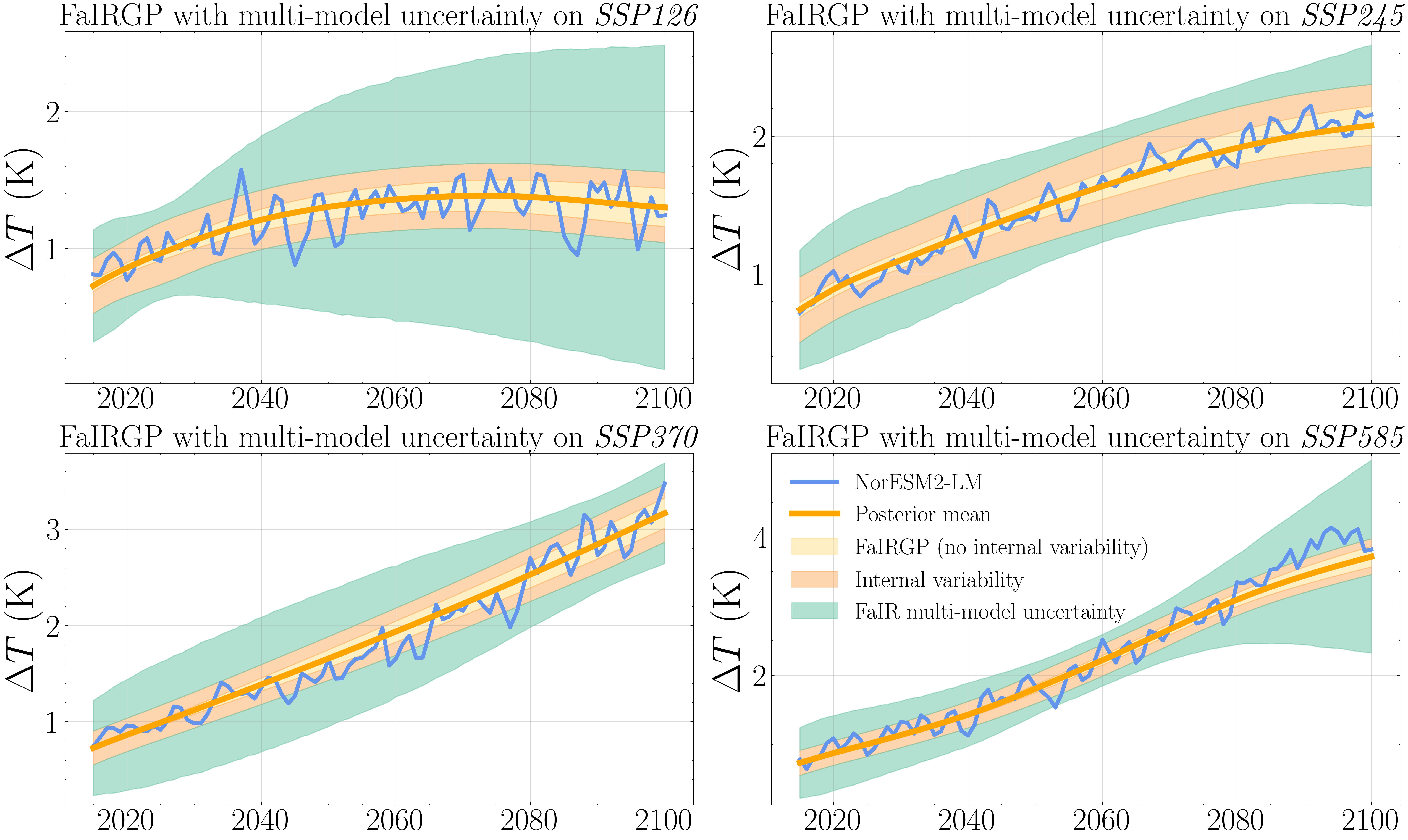}
    \vspace*{-1em}
    \caption{Global annual mean surface temperature anomaly projections over SSPs with contributions to uncertainty from FaIRGP without internal variability (yellow), the internal variability modelled in FaIRGP as in Section~\ref{subsection:fairgp-internal-variability} (orange) and the FaIR multi-model uncertainty (green). Shading indicates 2 standard deviations and is computed by taking the average contribution to the standard deviation of each source of uncertainty}
    \label{fig:ssp-prediction-uncertainty-stacked}
    \vspace*{-4em}
\end{figure}

A higher FaIR multi-model uncertainty contribution to the total variance (green) indicates significant sensitivity of the FaIRGP mean response to parameter choices within the constrained ensemble. Conversely, a greater FaIRGP uncertainty contribution (orange and yellow) suggests limited sensitivity to parameter choices, with most variability arising from our stochastic treatment of the temperature response. We find that in general, the FaIR multi-model uncertainty contributes the largest fraction of the total variance in projections of future warming. This is particularly evident in 2100 for low/high end forcing scenarios \textit{SSP126} and \textit{SSP585}, since they lie outside of the training data regime used in FaIRGP and the model's need for extrapolation makes it more sensitive to FaIR's structural uncertainty.

We find that for intermediate forcing scenarios \textit{SSP245} and \textit{SSP370}, the FaIRGP uncertainty contribution to the total variance remains significant. Because these scenarios lie in between other scenarios from the training data, FaIRGP's reliance on interpolation between its training data results in lower sensitivity to the multi-model uncertainty, making the idea of an uncertain future warming better constrained. This is supported by Figure~\ref{fig:ssp-prediction-uncertainty-stacked}, depicting how the average contribution of each source of uncertainty to the standard deviation stacks up in predictions. It shows a consistently stable and comparable magnitude of uncertainty introduced by FaIRGP with internal variability across scenarios.

It is important to caveat that the FaIR constrained ensemble samples from inter-model variability, whereas FaIRGP is here only calibrated on NorESM2-LM data. Therefore, one could reasonably expect that a multi-model version of FaIRGP, allowing calibration across multiple climate models, would likely contribute more uncertainty to projections of end-of-century warming.

\section{Application: spatial surface temperatures emulation}\label{section:spatial-experiments}

In this section, we pursue the comparison of FaIRGP to baseline models, but for the task of emulating spatially-resolved temperatures. As for global emulation, we start by briefly illustrating how the model concretely applies. Then, we benchmark it against baseline models for SSP emulation and evaluate the emulation of surface temperatures forced by anthropogenic aerosols only. We conclude by investigating how FaIRGP can help infer spatial top-of-atmosphere radiative forcing maps.

\vspace*{-0.5em}
\subsection{FaIRGP for spatial temperatures emulation}
\vspace*{-0.5em}

Figure~\ref{fig:ssp245-prediction-FaIRGP} illustrates the spatial surface temperature anomaly emulation with FaIRGP for test scenario \textit{SSP245}. As in the global case, a prior response based on FaIR is first specified, and then shifted by a data-informed posterior correction map. The posterior correction is learned from training scenarios $\cD_\text{train} = \{$\textit{historical, SSP126, SSP370, SSP585}$\}$.

The prior spatial response is constructed using a pattern scaling model trained on $\cD_\text{train}$, and therefore introduces a fixed spatial pattern that is only rescaled by changes in global mean temperature. The posterior correction at location $x$ is obtained by updating the prior over $\sfT(x,t)$ from (\ref{eq:spatial-fairgp-prior}) with local surface temperature data from $\cD_\text{train}$. The posterior correction maps effectively provide a data-driven way to deviate from this fixed spatial pattern, and better account for the possibly varying spatial temperature patterns of the emulated scenario. Finally, by linearly adding the posterior correction map to the prior, we obtain a posterior spatial temperature response over \textit{SSP245}.

\begin{figure}[h]
    \centering
    \hspace*{-1.7em}
    \vspace*{-1em}
    \includegraphics[width=1.1\linewidth]{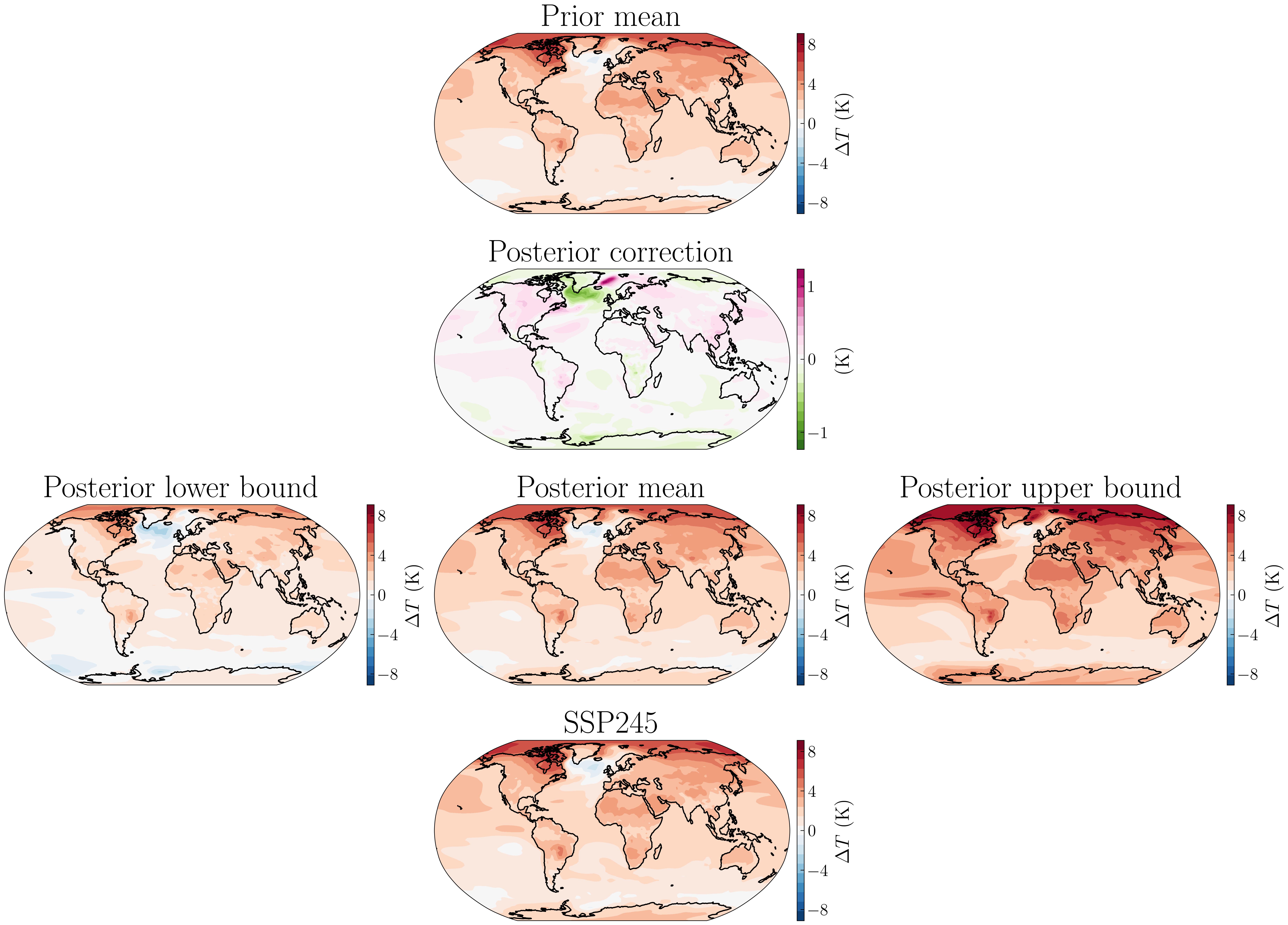}
    \caption{Spatial emulation of NorESM2-LM \textit{SSP245} temperatures. The prior mean (top) corresponds to a pattern scaling model used on the FaIR global temperature response. The posterior correction (middle-top) learns from data how to deviate from the pattern scaling prior to produce a posterior response (middle-bottom) that better emulates the groundtruth surface temperatures (bottom). The posterior lower/upper bounds capture the 95\% credible interval. Maps are averaged over the 2080-2100 period.}
    \label{fig:ssp245-prediction-FaIRGP}
    \vspace*{-3em}
\end{figure}


\begin{table}[t]
\centering
    \caption{Scores of baseline emulators and FaIRGP for the task of emulating SSPs spatial temperatures over 2080-2100 period; the best emulator for each metric is highlighted in bold; $\uparrow\!/\!\downarrow$ indicates higher/lower is better; we report 1 standard deviation; $\dagger$ indicates our proposed method.}
        \resizebox{\linewidth}{!}{
        \begin{tabular}{lccccccc}
        \toprule
        Emulator &   RMSE\small{$\;\downarrow$} &   MAE\small{$\;\downarrow$} &   Bias &  LL\small{$\;\uparrow$} &  Calib95 & CRPS\small{$\;\downarrow$} \\
        \midrule
        Pattern scaling & 0.637\tiny{$\pm$0.189} &  0.467\tiny{$\pm$0.139} & -0.162\tiny{$\pm$0.194} &       - &       - &      - \\
        Plain GP &  1.069\tiny{$\pm$0.616} &  0.761\tiny{$\pm$0.441} & -0.362\tiny{$\pm$0.350} & -1.625\tiny{$\pm$1.192} &   0.850\tiny{$\pm$0.141} &  0.560\tiny{$\pm$0.341} \\
        FaIRGP$^{\dagger}$   &  \textbf{0.619\tiny{$\pm$0.188}} &  \textbf{0.451\tiny{$\pm$0.138}} & \textbf{-0.113\tiny{$\pm$0.192}} & \textbf{-0.853\tiny{$\pm$0.398}} &   \textbf{0.873\tiny{$\pm$0.093}} &   \textbf{0.341\tiny{$\pm$0.099}} \\
        \bottomrule
        \end{tabular}
        }
    \label{table:spatial-ssp-scores}
\end{table}
\begin{figure}[h]
    \centering
    \hspace*{-1.7em}
    \includegraphics[width=1.1\linewidth]{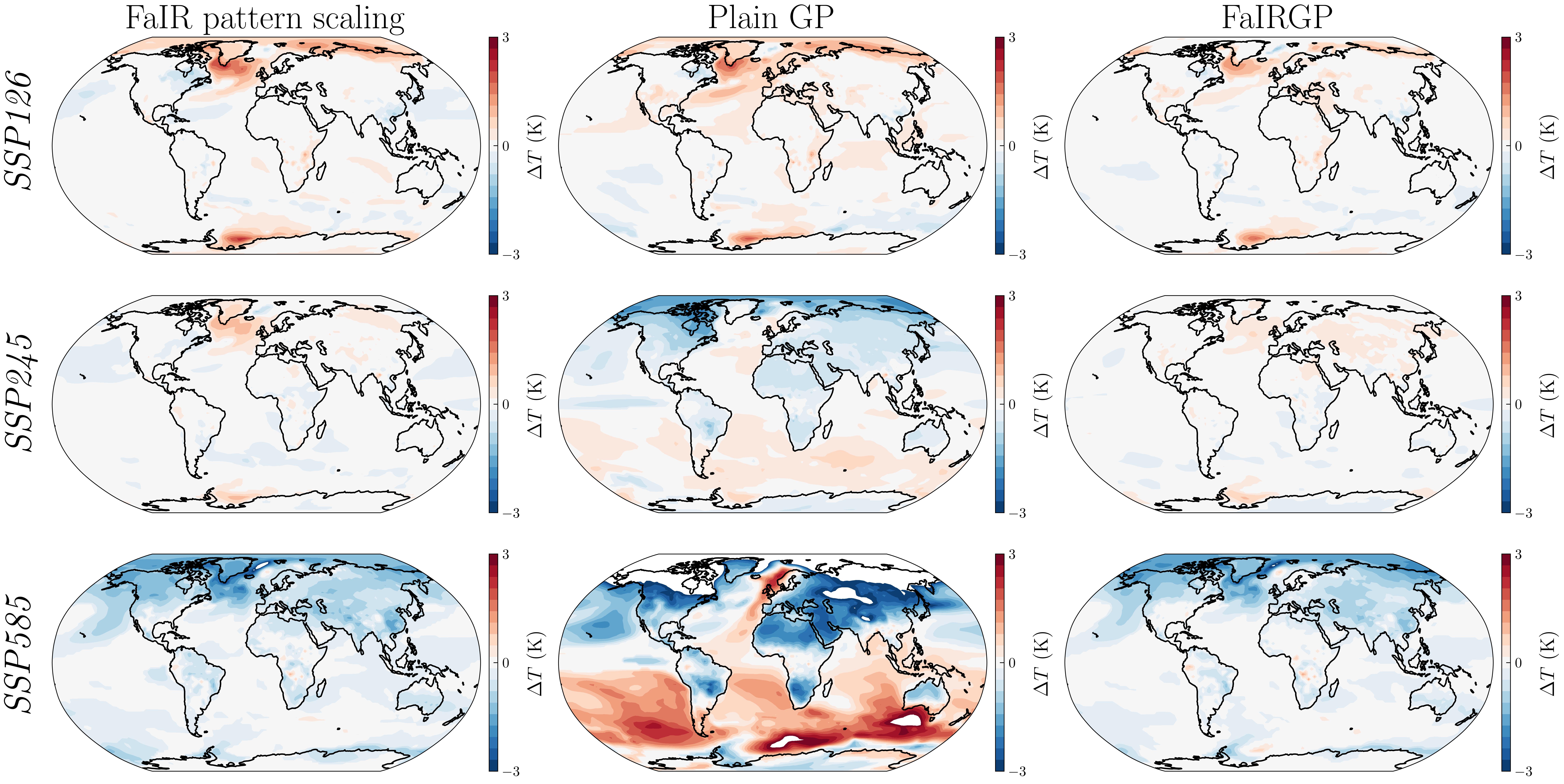}
    \caption{Maps of the mean difference between emulated spatial temperatures and target NorESM2 values for each emulator under a low forcing future scenario (\textit{SSP126}), a medium forcing future scenario (\textit{SSP245}) and a high forcing future scenario (\textit{SSP585}). Red/Blue means the emulator overshoots/undershoots NorESM2 temperatures. Maps are averaged over the 2080-2100 period. Differences insignificant at the $p < 0.05$ level are masked.}
    \label{fig:ssp-spatial-bias}
    \vspace*{-2em}
\end{figure}

\subsection{Shared socio-economic pathways emulation}

In this experiment, we follow the same procedure as in the global emulation experiment and iteratively train models on the dataset deprived of one SSP scenario, then use the retained scenario as a test scenario for evaluation. We benchmark FaIRGP against a FaIR pattern scaling model and a purely data-driven plain GP emulator. The pattern scaling model is obtained by fitting a linear regression model using the same training data as the other models. The plain GP model is analogous to the baseline GP emulator in \citeA{watsonparris2021climatebench}, but differs in two aspects: (\textit{i}) we adopt a simpler construction for the covariance with an anisotropic Matérn-3/2 kernel, and (\textit{ii}) our model takes as input global aerosols emissions whereas \citeA{watsonparris2021climatebench} use spatially-resolved aerosols emission maps. Scores are computed over the 2080-2100 period since the start of all SSPs is quite similar. Mean scores are reported in Table~\ref{table:spatial-ssp-scores}.

We find that FaIRGP has on average lower error than baseline models. Figure~\ref{fig:ssp-spatial-bias} shows that the spatial bias patterns of FaIRGP are similar to the bias patterns obtained with the FaIR pattern scaling model. This indicates that the prior has a strong influence on the predicted posterior. Nonetheless, we observe that the posterior correction in FaIRGP helps mitigate the spatial inaccuracies of its prior. This is particularly evident for \textit{SSP126} and \textit{SSP245} where the magnitude of the spatial bias in FaIRGP is overall smaller than the one for the FaIR pattern scaling model. Regarding the high forcing scenario \textit{SSP585}, the spatial correction is more subtle. This is due to it becoming an extrapolation task, and as a result, FaIRGP exhibits behaviour closer to its pattern scaling prior.

Figure~\ref{fig:ssp-spatial-bias} also shows that, as for global temperature emulation, the plain GP model predicts sound surface temperature maps for low and medium forcing scenarios, but struggles at extrapolating over high forcing scenarios.

\subsection{Emulating anthropogenic aerosols forcing}

In this experiment, we want to evaluate the emulation of temperature changes induced by anthropogenic aerosol emissions. We use the \textit{hist-aer} experiment from the ClimateBench v1.0 dataset. The \textit{hist-aer} experiment is generated using NorESM2-LM, using only historical anthropogenic aerosol emissions, and setting long-lived greenhouse gas emissions to zero. We emulate surface temperature anomalies over this scenario using emulators trained on all available \textit{historical} and SSPs experiments. The FaIR forcing model is modified to include the aerosol-cloud interaction parametrisation from \cite[Section 2.2.1]{leach2021fairv2}.

Figure~\ref{fig:global-hist-aer-emulation} shows the emulated global mean surface temperature anomaly with FaIRGP. Whilst FaIR faces challenges in reproducing the magnitude of the global temperature response to anthropogenic emissions, the posterior correction introduced by FaIRGP allows for a deviation from FaIR, resulting in predicted temperature anomalies that more accurately reflect the NorESM2-LM response. This suggests that FaIRGP can leverage information from the \textit{historical} and SSPs experiments to enhance its representation of the temperature response to anthropogenic aerosol emissions.

\begin{figure}[h]
    \centering
    \includegraphics[width=0.65\linewidth]{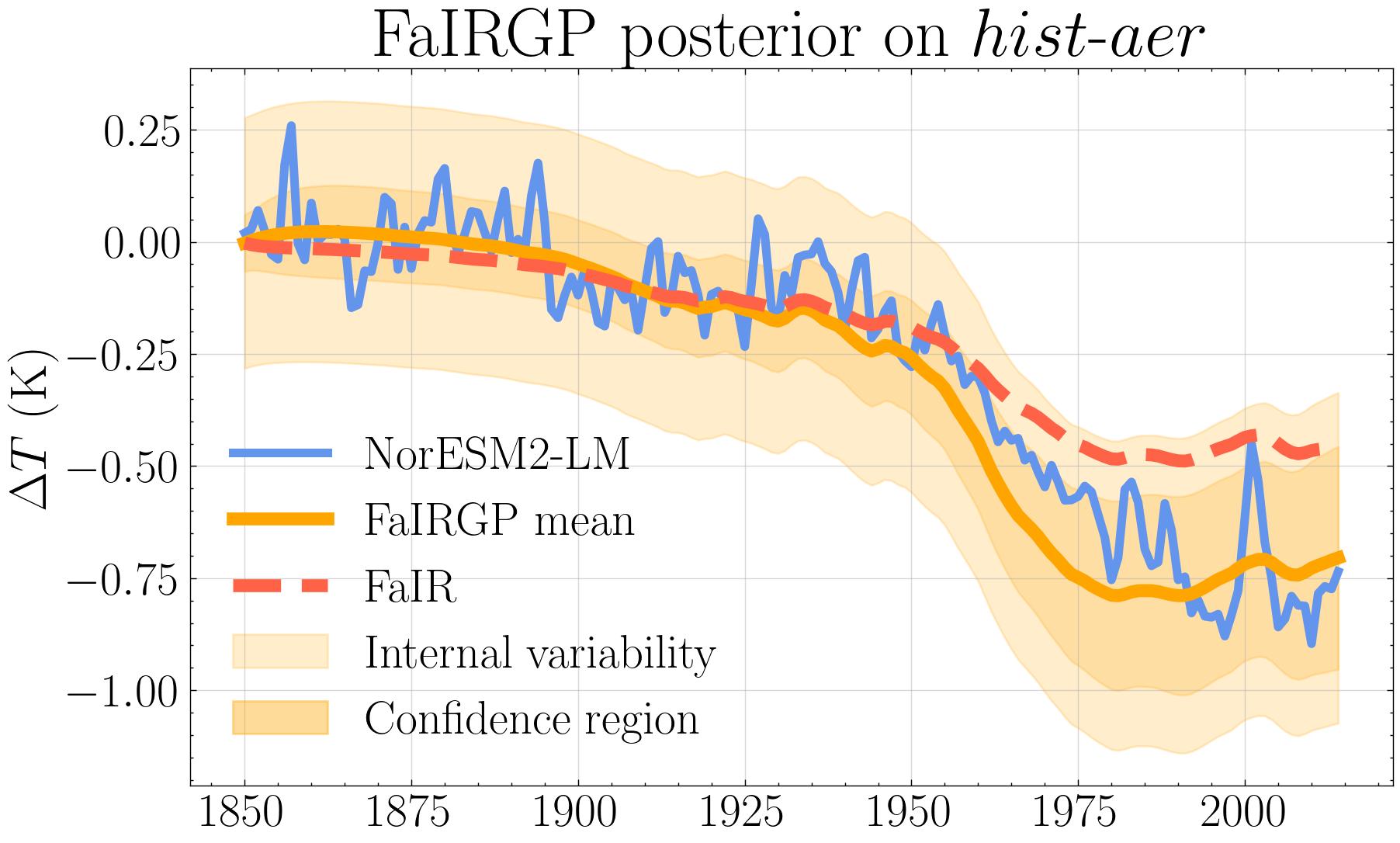}
    \caption{Emulated global mean surface temperature anomaly over 1850-2014 for the \textit{hist-aer} experiment}
    \label{fig:global-hist-aer-emulation}
\end{figure}


We now turn to the emulation of local surface temperature anomaly maps for the \textit{hist-aer} experiment. Figure~\ref{fig:hist-aer-plots} presents a comparison of predictions with the FaIR pattern scaling baseline and FaIRGP over three time periods. Consistently with the global trend from Figure~\ref{fig:global-hist-aer-emulation}, the predicted temperature anomaly magnitudes for 1900-1950 are similar between the pattern scaling model and FaIRGP, and align with the NorESM2-LM response whereas, over periods 1950-1980 and 1980-2014, the magnitude of the FaIRGP spatial response more accurately reflects the magnitude of the NorESM2-LM temperature response.

Regarding predicted spatial patterns, we observe that FaIRGP is able to deviate from its fixed pattern prior to partially reproduce the North Atlantic warming observed in NorESM2-LM after 1950. However, the predicted spatial temperature anomaly pattern with FaIRGP remains strongly influenced by its pattern scaling prior. This is particularly evident in the South Atlantic where FaIRGP reproduces the warming predicted by the pattern scaling model. This task is notoriously difficult for pattern scaling models, which excel at emulating responses to greenhouse gas emissions but struggle under strong aerosol forcing scenarios~\cite{may2012assessing, levy2013roles}. Considering the significant impact of the prior choice on FaIRGP, this fosters advocacy for a prior local response model that goes beyond pattern scaling models.  Additionally, using spatially-resolved emissions maps for aerosols, instead of global aerosols emissions, should enhance emulated temperatures given the strong influence spatial patterns of aerosols have over radiative forcing~\cite{williams2022strong}.

\begin{figure}[t]
    \centering
    \hspace*{-1.7em}
    \includegraphics[width=1.1\linewidth]{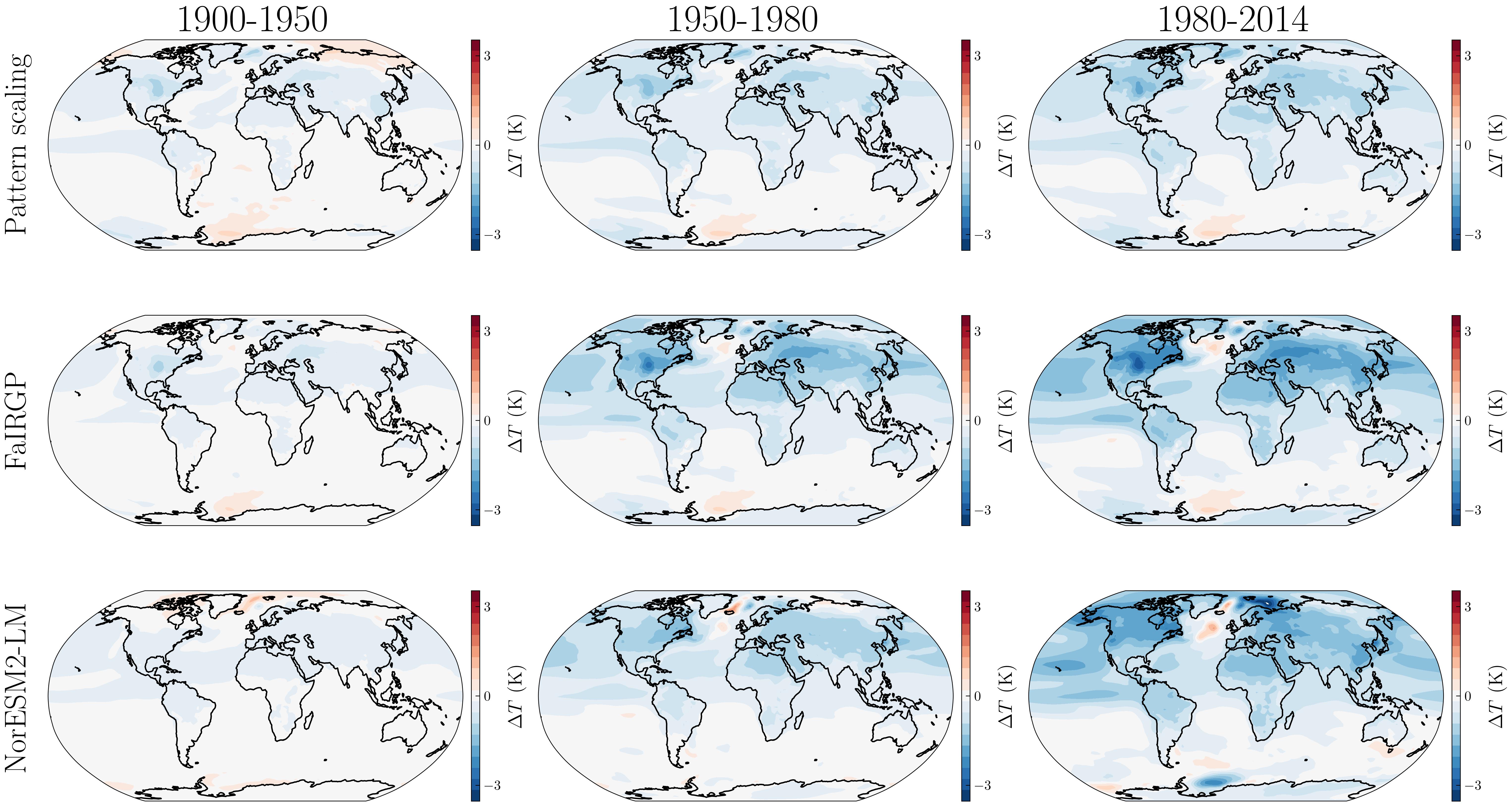}
    \caption{Maps of local surface temperature anomalies for the \textit{hist-aer} experiment averaged over the 1900-1950 period (left), 1950-1980 (left-right center) and 1980-2014 (right). \textbf{Top:} Prediction with the FaIR pattern scaling baseline. \textbf{Top-Bottom Center:} Prediction with the FaIRGP posterior mean. \textbf{Bottom:} Groundtruth NorESM2-LM surface temperature anomaly map.}
    \label{fig:hist-aer-plots}
    \vspace*{-2em}
\end{figure}

We conclude this experiment with a quantitative evaluation of the predicted local temperature anomaly maps over the 1950-2014 period. This is the period where the temperature response to aerosol emissions is most discernible. Because the anthropogenic aerosols forcing pattern is spatially heterogeneous and mostly localised around industrialised regions~\cite{CARSLAW2022101}, we assess the emulators' performance over four major industrialised regions: North America, East Asia, South Asia and Europe. The regions' extents are depicted in Figure~\ref{fig:regions-used-for-evaluation} and are constructed from world region delineation proposed by \cite{seneviratne2012changes}. Table~\ref{table:hist-aer-scores} shows that, consistently with the prediction maps, the FaIRGP posterior mean spatial response outperforms the FaIR pattern scaling baseline for every region considered. The plain GP model, as anticipated, faces difficulties in emulating local surface temperature anomalies solely based on aerosol emissions inputs. We believe this is primarily because every scenario from its training data includes greenhouse gas emissions. Appendix~\ref{appendix:aerosol-emulation-results} provides maps emulated with the plain GP model.

\begin{figure}[h]
    \centering
    \includegraphics[width=0.6\linewidth]{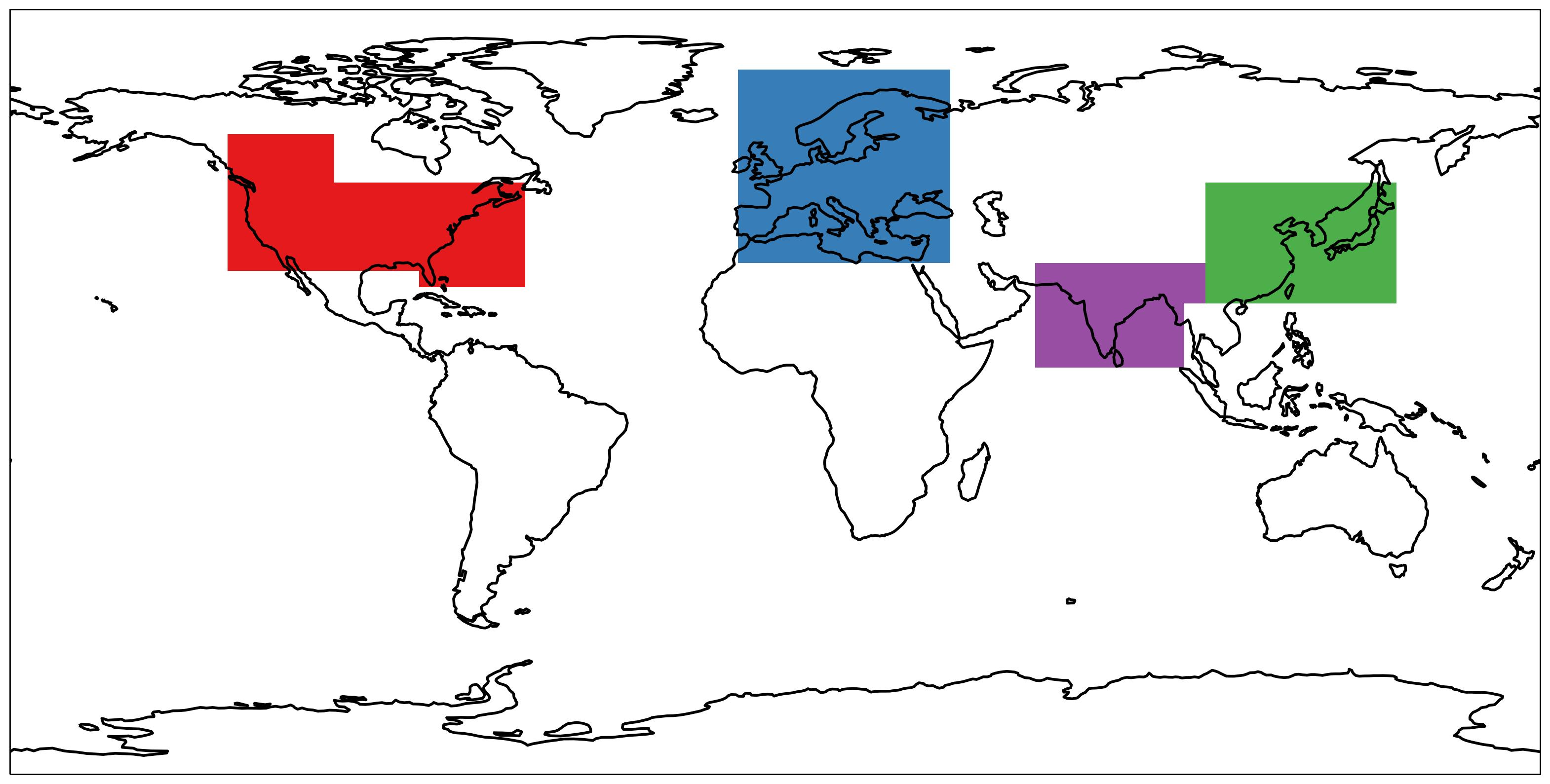}
    \caption{The polygons show the industrialised regions used for the evaluation of the predicted local annual mean surface temperature anomaly under the \textit{hist-aer} experiment:  North America (red), Europe (blue), South Asia (purple), East Asia (green).}
    \label{fig:regions-used-for-evaluation}
\end{figure}

\begin{table}[t]
\centering
    \caption{Scores of baseline emulators and FaIRGP for the task of emulating \textit{hist-aer} local annual mean surface temperature anomaly for four industrialised regions over the 1950-2014 period; the best emulator for each metric is highlighted in bold; $\dagger$ indicates our proposed method.}
        \begin{tabular}{llccc}
        \toprule
         Region   &   Emulator      &  RMSE (10\textsuperscript{-1})&  MAE (10\textsuperscript{-2}) &      Bias  (10\textsuperscript{-2}) \\
        \midrule
        \multirow{3}{*}{\textit{Europe}} & Pattern scaling &  1.29 &  1.71 &  1.29 \\
            & Plain GP &  2.26 &  3.37 &  3.07  \\
            & FaIRGP$^\dagger$ &  \textbf{1.05} &  \textbf{1.35} & \textbf{-0.313} \\ \midrule
        \multirow{3}{*}{\textit{North America}} & Pattern scaling &  1.10 &  1.63 &  1.21 \\
            & Plain GP &  2.40 &  3.70 &  3.65 \\
            & FaIRGP$^\dagger$ & \textbf{ 0.998} &  \textbf{1.35} & \textbf{-0.795} \\ \midrule
        \multirow{3}{*}{\textit{South Asia}} & Pattern scaling &  0.523 &  0.594 &  0.389 \\
            & Plain GP &  1.26 &  1.53 &  1.52 \\
            & FaIRGP$^\dagger$ &  \textbf{0.382} &  \textbf{0.404} & \textbf{-0.0100} \\ \midrule
        \multirow{3}{*}{\textit{East Asia}} & Pattern scaling &  1.01 &  1.32 &  1.20 \\
            & Plain GP &  2.22 &  3.08 &  3.08 \\
            & FaIRGP$^\dagger$ &  \textbf{0.631} &  \textbf{0.772} & \textbf{-0.173} \\
        \bottomrule
        \end{tabular}
    \label{table:hist-aer-scores}
\end{table}

\subsection{Inferring radiative forcing from temperatures}

As in Section~\ref{subsection:global-forcing-posterior}, we propose to probe whether FaIRGP can be used to infer spatial forcing maps. Since we do not have a spatial forcing model, we simply use a spatially constant forcing prior for $\sfF(x,t)$. We update it with spatially-resolved temperature data from \textit{historical} and SSPs scenarios. Figure~\ref{fig:spatial-radiative-forcing} compares the obtained posterior mean forcing with historical forcing maps simulated with NorESM2-LM.

The magnitude of the inferred posterior mean forcing with FaIRGP is very conservative, and does not align with the magnitude of the forcing levels observed in the groundtruth NorESM2-LM maps. The sign of the inferred forcing with FaIRGP also disagrees with the groundtruth maps in several regions. This is most evident around 90$^\circ$ latitudes, where FaIRGP infers a positive forcing whilst NorESM2-LM simulations suggest a negative forcing.

We attribute this discrepancy to FaIRGP relying only on local temperature anomalies to infer local radiative forcing. Whilst high latitudes tend to warm faster than around the equator (due among things to surface albedo feedback and atmospheric circulation for equator to poles), the top-of-atmospheric effective radiative forcing at high latitudes can remain negative (due to higher albedo, larger zenith-angle). This underlines an important point: FaIRGP performs probabilistic inference of local radiative forcing given local temperature anomalies, which alone is not sufficient to account for the complex spatiotemporal dynamics of the climate system. However, there are reason to believe that informing this probabilistic inference with additional covariates, such as surface albedo or net top-of-atmosphere data, could enhance the accuracy of the inferred spatial forcing maps.

\begin{figure}[h]
    \centering
    \vspace*{-0.5em}
    \hspace*{-1.7em}
    \includegraphics[width=1.1\linewidth]{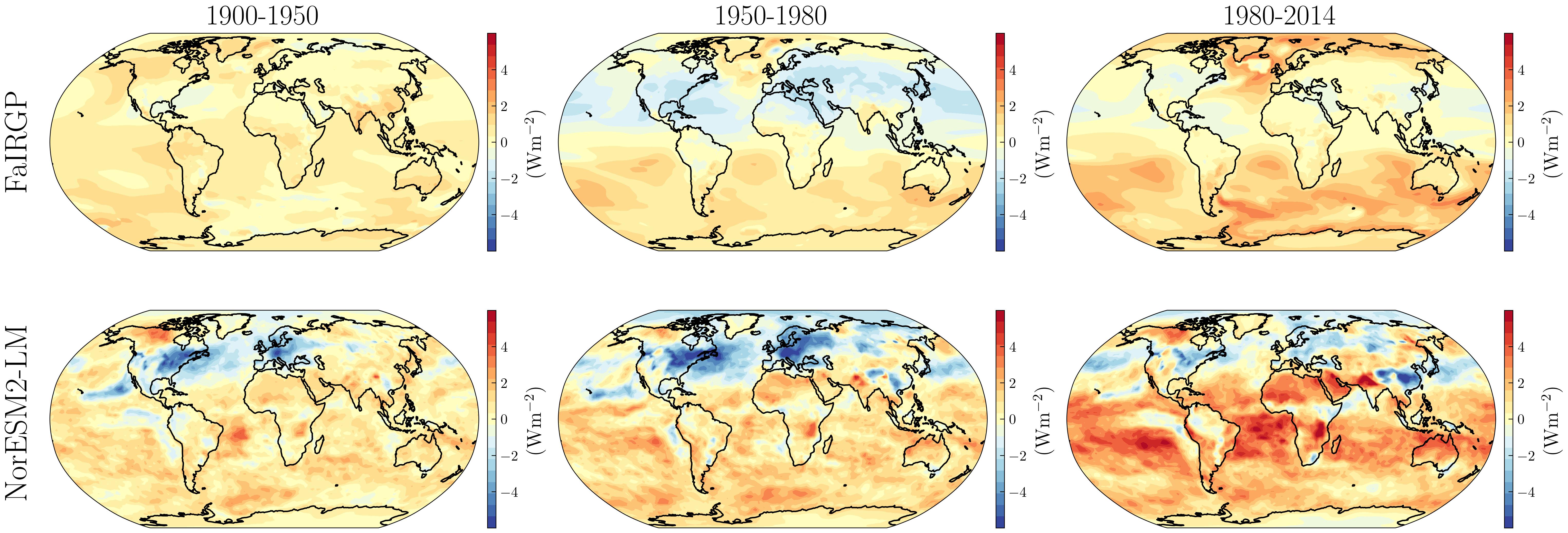}
    \vspace*{-0.5em}
    \caption{Maps of top-of-atmosphere radiative forcing averaged over 1900-1950, 1950-1980 and 1980-2014 periods of the \textit{historical} scenario. \textbf{Top:} Mean posterior radiative forcing inferred with FaIRGP. \textbf{Bottom:} NorESM2-LM radiative forcing maps.}
    \vspace*{-2em}
    \label{fig:spatial-radiative-forcing}
\end{figure}

Finally, we note that whilst the spatial forcing maps inferred with FaIRGP inadequately capture the groundtruth forcing magnitude and pattern, they still manage to partially reproduce two important features of the climate system: the overall temporal increase in forcing, as well as the mid-1900s hemispheric contrast. In particular, Figure~\ref{fig:spatial-radiative-forcing-latitudinal} shows that whilst predictions largely underestimate the forcing magnitude, they also display some latitudinal correlation with the NorESM2-LM simulated forcing over the 1940-1980 period. This is encouraging considering that the predicted spatial patterns are solely inferred from temperature patterns.

\begin{figure}[h]
    \centering
    \vspace*{-0.5em}
    \includegraphics[width=0.6\linewidth]{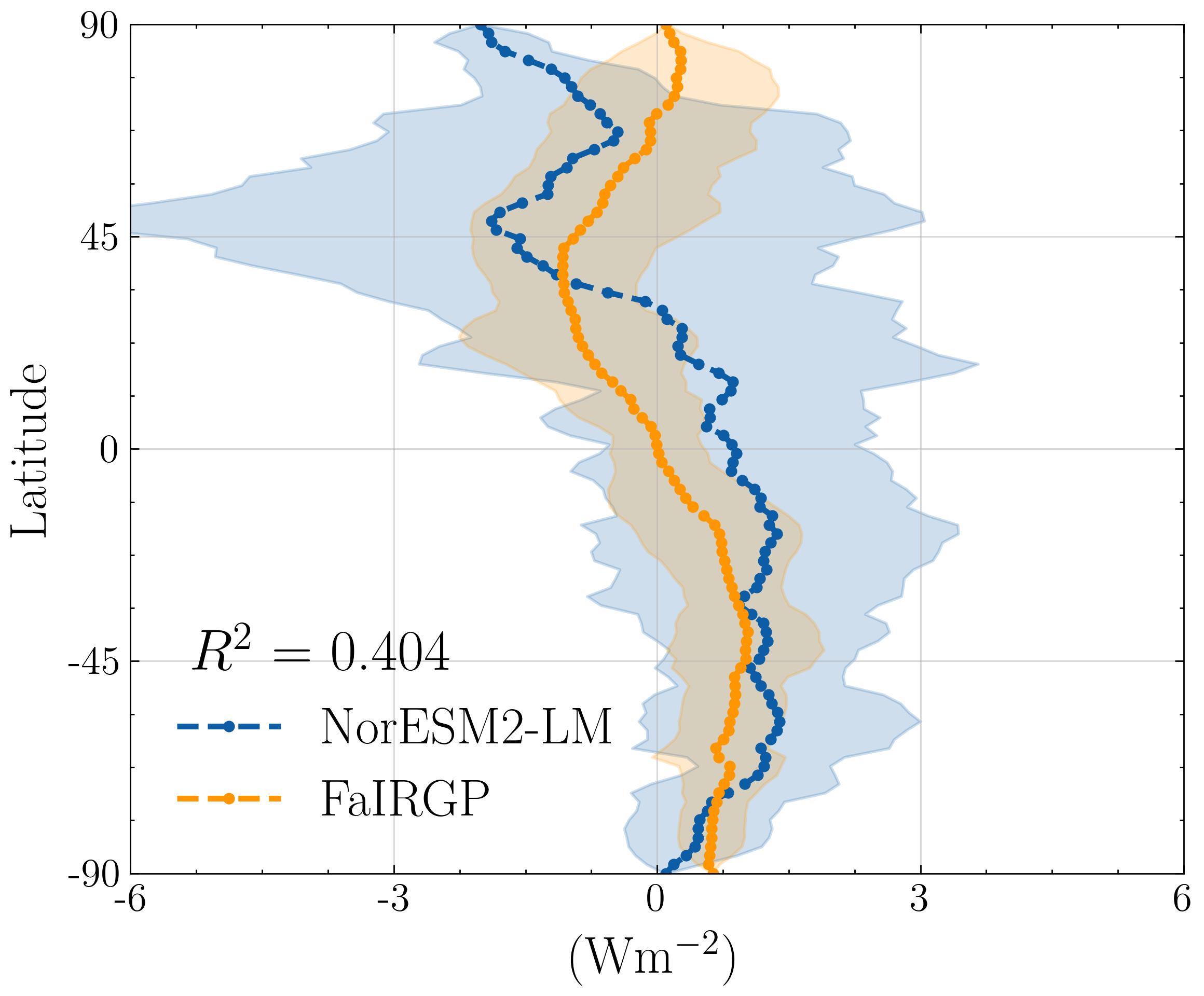}
    \vspace*{-0.5em}
    \caption{Latitudinal average of FaIRGP inferred radiative forcing and NorESM2-LM radiative forcing over 1940-1980 period. The shaded area reports two standard deviations of longitudinal variability. $R^2 = 0.404$ is computed between latitudinal averages.}
    \vspace*{-2em}
    \label{fig:spatial-radiative-forcing-latitudinal}
\end{figure}

\section{Discussion}\label{section:discussion}

\subsection{About the GP approach}

\subsubsection{Comparison to placing Bayesian priors over the SCM parameters}

A simple way to introduce model variability in simple climate models is to place Bayesian priors over model parameters, such as the carbon cycle feedback terms or the forcing model coefficients~\cite{leach2021fairv2}. Such priors can then be updated with global temperature data to formulate posterior distributions. This specifies a probabilistic climate model emulator calibrated against climate model data. Whilst aligned in spirit with the work proposed in this paper, it is important to highlight the distinctions in objectives and technical scope with our GP-based approach.

Placing Bayesian priors over parameters allows to sample from diverse sources of uncertainty such as uncertainty on the temperature response parameters, forcing sensitivities or carbon cycle sensitivities. Importantly, this allows to explore climate sensitivity and the structural uncertainty of models. In contrast, the framework we propose only addresses the forcing uncertainty and considers fixed values for the SCM parameters. In that sense, FaIRGP, as it is formulated in our work, provides a more restrictive exploration of the uncertainty than methods placing priors over an SCM parameters.

The main advantage of our GP-based approach rather lies in its technical convenience. Chiefly, FaIRGP formulates analytical expressions for the posterior distribution over forcing and temperatures, which have an intuitive interpretation as the sum of a physics-driven prior and a data-driven correction. In contrast, when placing a Bayesian prior over parameters, we do not in general have access to a probability distribution for temperatures. Therefore, probabilistic emulation may need to be sampling-based, and one must resort to more complex Markov-Chain Monte-Carlo (MCMC) techniques to sample from the posterior distribution.

In addition, beyond particular modelling objectives, sampling-based approaches have two main shortcoming from a purely technical perspective: \textit{(i)} a thorough uncertainty quantification requires storing a tremendous amount of scenarios, which can be limited by memory capacity --- \citeA{beusch2022emission} require an ensemble of 9 million emulations; \textit{(ii)} it reduces the statistical representation of uncertainty to summary statistics such as the mean, standard deviation or quantiles. In contrast, having a closed-form expression for the posterior distribution with GPs allows to analytically conduct probabilistic studies using the full probability density, and additionally draw samples, if needed. Finally, having access to the probability density expression allows to evaluate the likelihood of observations. This can critically be used as a maximisation objective to tune the model parameters against observations, but also in the context of Bayesian optimisation routines to find optimal emission trajectories to meet climate goals.

Because of the technical convenience that GP-based approaches display, and the richer structural uncertainty of sampling-based approaches, we argue that rather than opposing them, the two paradigms should be viewed as complementary. Namely, endeavours to formulate a hierarchical Bayesian model that both places priors over the SCM parameters and uses a FaIRGP backbone are promising directions to simultaneously address many sources of uncertainty while retaining some technical advantages of working with GPs.





\subsubsection{Connection with stochastic energy balance models}

Our work is related to the work of \citeA{cummins2020optimal}, which formulates a stochastic energy balance model by introducing a white noise variability term in the temperature response, but also in the forcing model. This allows them to account for climate internal variability, and formulate a Kalman filtering strategy to obtain maximum likelihood estimators of the energy balance model parameters. Our work similarly introduces a white noise in the temperature response model to account for climate internal variability, which results in an additional temporal Matérn-1/2 covariance term in the prior over temperatures. Whilst an extended discussion goes beyond the scope of this work, it can be shown that in the long term regime, these two approaches are equivalent, and that more broadly, Kalman filtering models are in fact equivalent to temporal GPs with Matérn covariance functions~\cite{hartikainen2010kalman, sarkka2012infinite, sarkka2013spatiotemporal}.

Our work differs from \citeA{cummins2020optimal} in that beyond the stochasticity arising from internal variability, we also introduce a GP prior over the radiative forcing. This GP prior introduces stochasticity over the SCM design, which is not only a function of time, but also of emission levels. Because our modelling is not purely temporal (the GP is also a function of emissions), we cannot employ Kalman filtering strategies and instead choose to use GP regression techniques.

\subsubsection{Choice of kernel $\rho$}
The choice of covariance function $\rho$ in (\ref{eq:prior-rho-specification}) is an important choice that allows the user to incorporate their domain knowledge into the prior over the radiative forcing. Let $u$ and $u'$ be generic notations for input data (in our work, greenhouse gas and aerosol emissions), the kernel $\rho(u, u')$ specifies how will the prior covary between these two inputs. For example, choosing $\rho(u, u') =  \delta(u-u')$ makes the GP independent at any two inputs. On the other hand, choosing $\rho(u, u') = 1$ causes the GP to covary equally between any two inputs.

The Matérn family are a common family of kernel parameterised by a degree $\nu$. They allow to control for the functional regularity of the GP. For $\nu = 1/2$, draws from the GP are continuous functions, for $\nu=3/2$ they are once differentiable, and in the limit $\nu=\infty$ they become infinitely differentiable\footnote{The Matérn-$\infty$ kernel actually corresponds to the squared exponential kernel.}. A detailed presentation of the Matérn kernels is provided in Appendix~\ref{appendix:matern-covariance}.

Additions or multiplications of kernels can be used to construct more elaborate covariance functions that reflect an additive or multiplicative structure in the forcing. Periodic kernels can also be introduced to model seasonality. Going further, more complex choices of kernels include the spectral mixture kernel~\cite{wilson2013gaussian} which attempts to learn the spectral density of the data, or even kernels parametrised as neural networks~\cite{wilson2016deep, law2019hyperparameter}.

Whilst kernel selection plays a key role in GP regression, our focus in this work is on the development of the FaIRGP framework rather than refining the kernel itself. The kernel is treated throughout as a modular component, with potential for refinement. As a result, we choose to work with a simple anisotropic Matérn-3/2 kernel throughout. Preliminary findings indicate that the choice of kernel does not significantly degrade the results. Hence, dedicating efforts to constructing more elaborate kernels is likely to yield comparable or better results than our current approach.

\subsubsection{Computational efficiency and scalability}

We report that to emulate 100 years of surface temperature anomaly with FaIRGP, it takes less than a second for global and spatial emulation on an \enquote{average} personal laptop\footnote{16Go memory}, without requiring any parallelization methods.

Scalability issues are commonly associated with Gaussian Processes (GPs) when the training set grows in size. This is because computing their posterior distribution involves a matrix inversion, which has a cubic computational cost in the number of training samples. Fortunately, unlike neural networks which require large amounts of data~\cite{watsonparris2021climatebench}, GPs excel in scenarios with limited data~\cite{rasmussen2005gaussian}. Consequently, it is possible to develop skilful GP emulators with limited training data.

In cases where using a larger training dataset becomes a necessity, one can still employ linear conjugate gradient methods and parallelisation schemes~\cite{wang2019exact} to scale exact GPs to millions of data points. Alternatively, sparse approximation techniques can be used to obtain a scalable estimate of the posterior distribution~\cite{rahimi2007random, tsitsias2009variational, matthews2016sparse}.

\subsection{Climate modelling considerations}

\subsubsection{Beyond FaIR: broader applicability of the method}

While we have chosen to use FaIR as the backbone climate model for our work, the rationale behind the development of FaIRGP is easily transferable to other commonly used simple climate models such as MAGICC~\cite{meinshausen2011emulating} or OSCAR~\cite{gasser2017compact}. These models share linear time invariant dynamics, allowing us to incorporate a GP prior into the forcing term of these dynamics. By doing so, we can obtain a GP-based solution that is informed by the model parameters, exactly like in FaIRGP. 

The dynamical systems of interest can naturally describe the temperature response to radiative forcing, as is the case in our work. However, we could also imagine extending this to carbon cycle models, where emission levels prescribe the forcing function, and the output of the dynamical system are atmospheric concentrations. This GP framework over dynamical systems is in fact highly general, and is known as latent force modelling. This paradigm was initially introduced by \citeA{alvarez2009latent, alvarez2013linear} in the context of linear dynamical systems where the forcing function is unknown. 

For more complex dynamical systems involving non-linearities, more sophisticated latent force modelling techniques may be applied, such as methodologies based on Volterra series approximation of the dynamical system~\cite{alvarez2019non, ross2021learning}, variational approximations of the posterior response~\cite{ward2020black, moss2021approximate}, with the option of using deep probabilistic modelling of the dynamical system~\cite{mcdonald2021compositional, baldwin2023deep}.

\subsubsection{Pixel independence assumption}

The pattern scaling model used in our prior effectively uses independent linear regressions at each location to map changes in global mean temperature onto changes in local temperature. However, this modelling approach challenges our intuition as it overlooks the spatial dependence of temperature fields, despite our expectation that temperatures at nearby locations should covary.

To address this modelling concern, a common solution is to incorporate spatially correlated innovations into the pattern scaling response, which represents the spatial expression of climate internal variability~\cite{beusch2020emulating, beusch2020crossbreeding, beusch2022emission, castruccio2013global, goodwin2020computationally}. Alternatively, \citeA{link2019fldgen} design a procedure based on the Wiener-Khinchin theorem~\cite{champeney1973fourier} to emulate a climate variability field with the same variance and spatiotemporal correlation structure as the one in ESMs outputs.

From a statistical modelling standpoint, these approaches introduce such spatially correlated innovations as a multivariate Gaussian variable, their main differences lying in the procedure used to determine its spatial covariance structure. Because of the Gaussian nature of the modelling we propose in this work, a Gaussian spatial correlated innovation can easily be incorporated in the FaIRGP spatial emulation model we propose. However, for the sake of clarity, we chose not to delve into these additional considerations and instead focus on the exposition of the Bayesian energy balance model.

Further, whilst it has been pointed out that, in general, regional changes in temperature scale robustly with global temperature~\cite{seneviratne2016allowable}, this may not be true under strong mitigation scenarios or under strong aerosol forcing~\cite{may2012assessing, levy2013roles, tebaldi2018evaluating, tebaldi2020emulating}. This fosters advocacy for the development of \enquote{spatialised} simple climate models going beyond pattern scaling. 


\subsubsection{Opportunities for precipitation emulation}

Going beyond surface temperature emulation, we can explore how FaIRGP can be used to emulate precipitations. One approach is to combine the Gaussian process GP emulator for precipitation described in \citeA{watsonparris2021climatebench} with FaIRGP. By leveraging Gaussian conjugacy relationships, a natural cross-covariance between the two emulators is induced. This enables the exchange of information between precipitation and temperature fields. This is in line with recent advocacy for joint emulation of temperatures and precipitations~\cite{snyder2019joint, schongart2022spatially, schongart2023extending}.

Another option is to use the emulated temperatures from FaIRGP as input for a statistical emulator of precipitation, such as a Gamma regression model~\cite{martinez2019why}. Indeed, many climate impacts are routinely assumed to be a function of temperatures. Additionally, the full probability distribution predicted by FaIRGP can be incorporated into the model, allowing for the propagation of epistemic uncertainty associated with FaIRGP.

\subsubsection{Application to detection and attribution}

With FaIRGP, we have access to the analytical expression of the probability density distribution of emulated temperatures. Therefore, we can emulate surface temperatures under historical scenarios, both with and without anthropogenic forcing, and analytically compute the probability of temperature occurrences in each scenario. By comparing these probabilities, we can assess the extent to which human activity has made a certain temperature range more likely, enabling us to conduct attribution studies.

Conducting detection and attribution studies with emulators is not exclusive to FaIRGP, and could in principle be conducted with any emulator as discussed in \citeA{watsonparris2021climatebench}. However, the strength of FaIRGP lies in its ability to input temperature ranges directly into a known probability density function, providing a precise probability between 0 and 1 of such temperatures to occur under a given emission scenario.

\section{Conclusion and outlooks}

Simple climate models (SCMs) are robust physically-motivated emulators of changes in global mean surface temperatures. Gaussian processes (GPs) are powerful Bayesian machine learning models capable of learning complex relationships, and emulate from data how changes in emissions affect changes in surface temperatures. By combining them together, we reconcile these two paradigms of emulator design, which mutually address their respective limitations. We introduce FaIRGP, a Bayesian energy balance model that \textit{(i)} maintains the robustness and interpretability of a simple climate model, \textit{(ii)} gains the flexibility of modern statistical machine learning models with the ability to learn from data, and \textit{(iii)} provides principled uncertainty quantification over the emulator design.

We demonstrate skilful emulation of global mean surface temperatures over realistic emission scenarios. Unlike GPs, FaIRGP has a robust physical grounding which allows it to provide reliable predictions even on out-of-sample scenarios. On the other hand, unlike SCMs, FaIRGP can learn complex non-linear relationships to deviate from an SCM and improve predictions. In particular, FaIRGP better accounts for the temperature response to anthropogenic aerosol emissions. We further show that these findings carry over to the task of emulating spatially-resolved surface temperature maps. In addition, we find that FaIRGP can also be used to produce estimates of top-of-atmosphere radiative forcing given temperature data.

The full mathematical tractability, with analytical expressions for probability distributions, provides great control over the modelling, and a rich framework to reason about probability distributions over temperatures. This is of great relevance to detection and attribution studies. Further, whilst our work focuses on temperature emulation --- which has already been thoroughly studied --- we envision FaIRGP as a foundation for the development of robust data-driven emulators for more complex climate variables, such as precipitations. Harnessing the mathematical properties of GPs, we believe that emulating climate impacts using FaIRGP will provide additional control over pure machine learning methods, whilst being able to capture complex non-linear relationships to forcing. 

We hope this work will contribute to building trust in data-driven models, and thereby allow the climate science community to benefit more widely from their potential.

\section*{Open Research}
\begin{itemize}
    \item Data --- The data used to run experiments in this paper is obtained from ClimateBenchv1.0~\cite{watsonparris2021climatebench} and available here \url{https://doi.org/10.5281/zenodo.5196512}.
    \item Software --- The model, evaluation metrics and all code used to run experiments and generate the plots is available here \url{https://doi.org/10.5281/zenodo.8164335}. All figures were made with Matplotlib version 3.6.2~\cite{caswell2020matplotlib, hunter2007matplotlib} and Cartopy~\cite{met2010cartopy}. Models are implemented with PyTorch version 1.12.1~\cite{paszke2019pytorch} and GPyTorch version 1.9.0~\cite{gardner2018gpytorch}.
\end{itemize}

\acknowledgments
We would like to thank Maybritt Schillinger for insightful discussions and helpful feedback which has greatly benefited this work. We would also like to thank three anonymous reviewers for their thorough reviews and comments, which contributed significantly to the quality of this work. Shahine Bouabid receives funding from the European Union’s Horizon 2020 research and innovation programme under Marie Skłodowska-Curie grant agreement No 860100.

\bibliography{references.bib}

\newpage

\appendix

\newpage
\section{Complementary material on Gaussian processes}\label{appendix:gps}

\subsection{Illustrated walk-through Gaussian process regression}\label{appendix:illustrations-gp}

We provide in this section an illustrated walk-through of how a regression task can be achieved using GPs. Begin by considering a GP specified by
\begin{equation}
    \left\{
    \begin{aligned}
        \begin{split}
            & \sfff(x) \sim \GP(m, k) \\
            & m(x) = 0 \\
            & k(x, x') = \sigma_\sfff^2 \exp\left(-\frac{|x - x|^2}{\ell}\right).
        \end{split}
    \end{aligned}
    \right.
\end{equation}
Further consider a set of $p = 500$ regularly spaced points between 0 and 1, $\texttt{grid} = \begin{bmatrix} z_1 & \ldots & z_p\end{bmatrix}$ where
\begin{equation}
    z_i = \frac{i}{p}, \enspace 1\leq i \leq p.
\end{equation}

Then, by definition of GPs, the evaluation of $\sfff$ on \texttt{grid} admits a joint multivariate normal distribution. Namely, if we define the covariance matrix
\begin{equation}
    \bK = \begin{bmatrix}
        k(z_1, z_1) & \ldots & k(z_1, z_p) \\
        \vdots & \ddots & \vdots \\
        k(z_p, z_1) & \ldots & k(z_p, z_p).
    \end{bmatrix},
\end{equation}
then we have
\begin{equation}
    \begin{bmatrix}
        \sfff(z_1) \\ \vdots \\ \sfff(z_p)
    \end{bmatrix} \sim \cN(\mathbf{0}, \bK).
\end{equation}

Therefore, the zero vector $\mathbf{0}$ provides the mean of $\sfff(x)$ at every grid point, and the diagonal elements of the covariance matrix provide the standard deviation of $\sfff(x)$ at every grid point following $\sqrt{k(z_i, z_i)} = \sigma_\sfff$. Figure~\ref{fig:gp-prior-only} plots the mean function and the 95\% credible interval of $\sfff(x)$ over the \texttt{grid}.
\begin{figure}[h]
    \centering
    \vspace*{-1em}
    \includegraphics[width=0.7\linewidth]{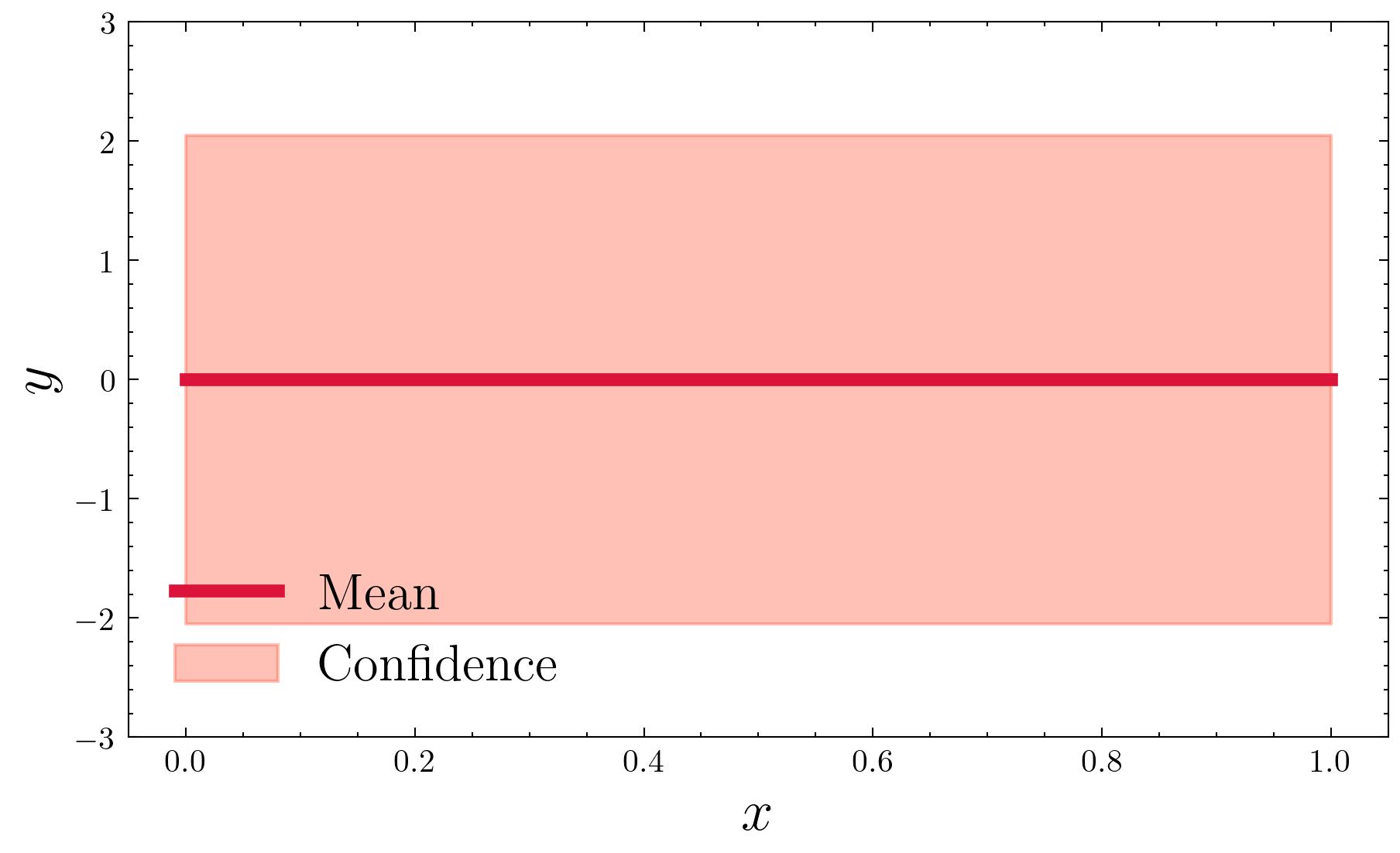}
    \vspace*{-1em}
    \caption{Plot of the $\sfff(x) \sim \GP(0, k)$ prior mean function and 95\% confidence region for $\sigma_\sfff^2 = 1$ and $\ell = 0.1$.}
    \vspace*{-1em}
    \label{fig:gp-prior-only}
\end{figure}

Furthermore, we are also able to draw samples from the multivariate normal distribution $\cN(\mathbf{0}, \bK)$ by taking
\begin{equation}
    \bff = \bK^{1/2} \bu, \enspace \bu\sim\cN(\mathbf{0}, \bI_p).
\end{equation}
Such a sample $\bff\in\RR^p$ corresponds to a sample path from the GP $\sfff(x)$ evaluated on the grid points. Figure~\ref{fig:gp-prior-samples} shows the plots of 10 sample paths, from this multivariate normal distribution.
\begin{figure}[h]
    \centering
    \vspace*{-1em}
    \includegraphics[width=0.7\linewidth]{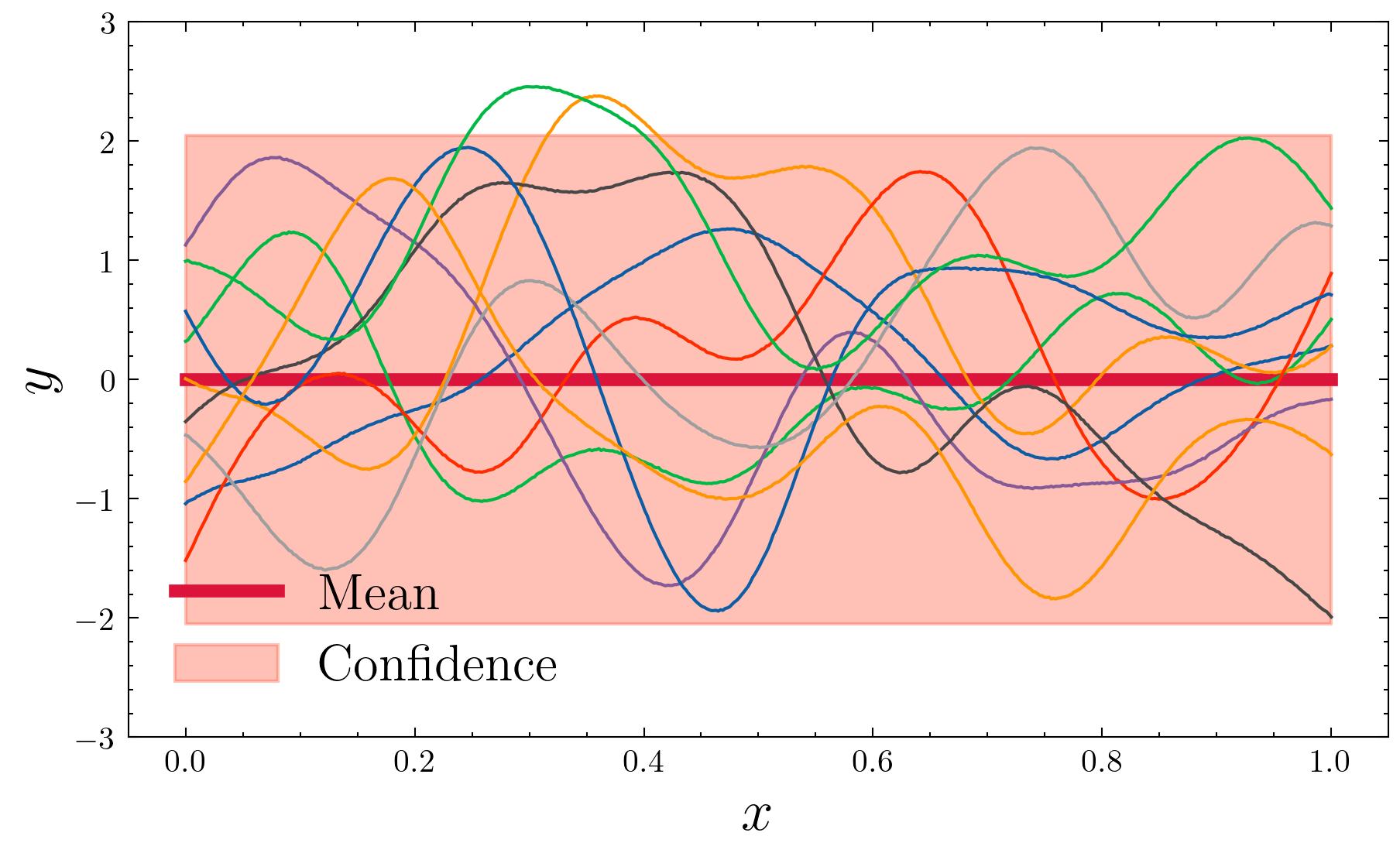}
    \vspace*{-1em}
    \caption{Same as Figure~\ref{fig:gp-prior-only}, but plotting on top 10 sample paths from $\bff \sim \cN(\mathbf{0}, \bK)$. $\sigma_\sfff^2 = 1$ and $\ell = 0.1$.}
    \vspace*{-1em}
    \label{fig:gp-prior-samples}
\end{figure}

Consider now the following data generating process
\begin{equation}
    y = \sin(2\pi x) + \epsilon, \enspace \epsilon\sim\cN(0, 0.04),
\end{equation}
from which we simulate $n = 10$ observations $\by = \begin{bmatrix}y_1 & \ldots & y_n\end{bmatrix}$ at inputs $\bx = \begin{bmatrix}x_1 & \ldots & x_n\end{bmatrix}$. Figure~\ref{fig:observed-data} shows a plot of the observed data.
\begin{figure}[h]
    \centering
    \vspace*{-1em}
    \includegraphics[width=0.6\linewidth]{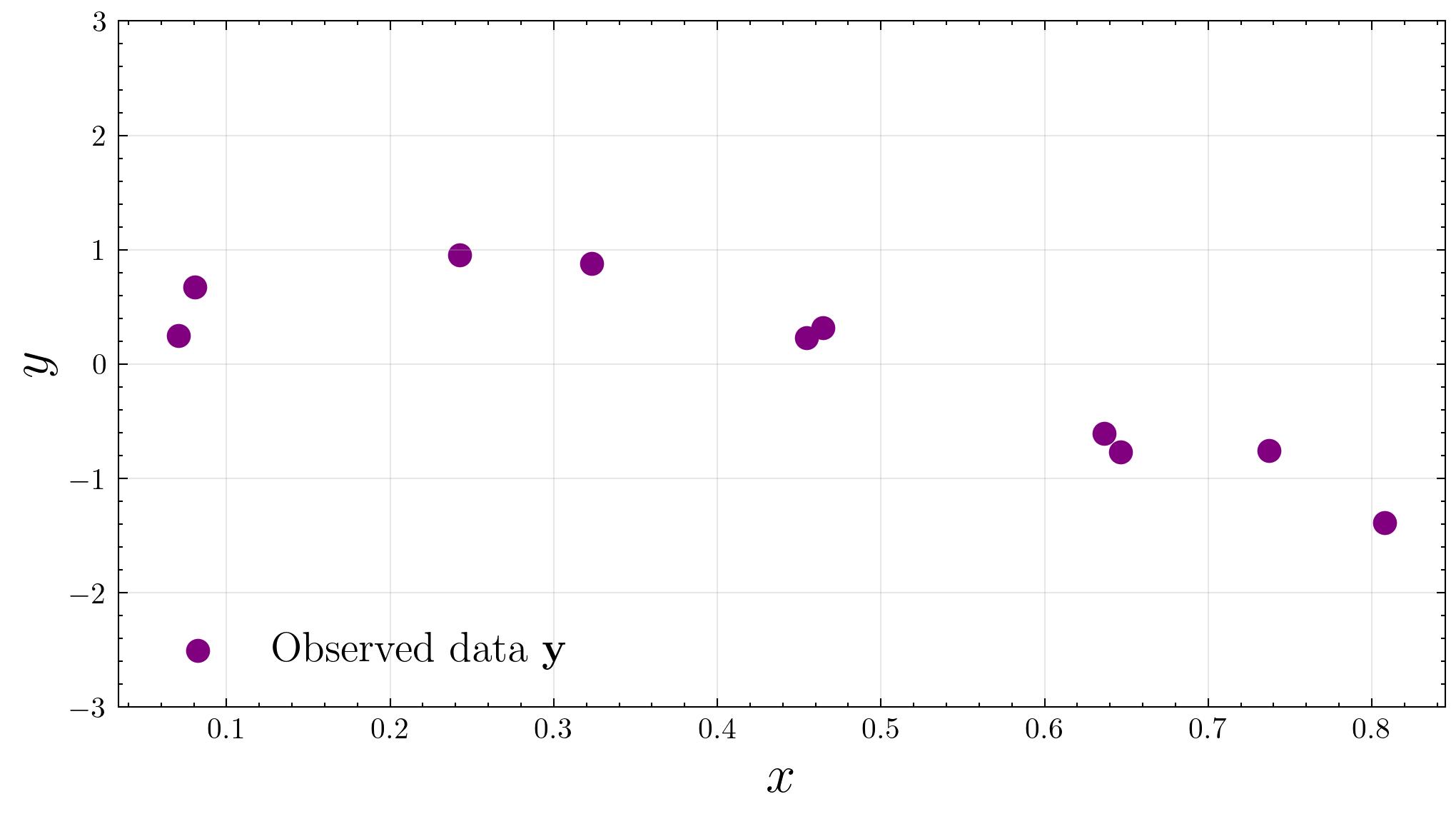}
    \vspace*{-1em}
    \caption{Scatter plot of $(\bx, \by)$}
    \vspace*{-1em}
    \label{fig:observed-data}
\end{figure}

Suppose that the form of the data generating process is unknown to a user who wishes to use these observations to regress $y$ onto $x$. One possibility is to use the previously defined GP $\sfff(x)$ as a Bayesian prior over the space of regression models $x\to y$ and assume the relationship
\begin{equation}
    y_i = \sfff(x_i) + \varepsilon_i, \enspace\varepsilon\sim\cN(0, \sigma_\varepsilon^2),
\end{equation}
for an unknown noise variance $\sigma^2_\varepsilon$.

By definition of GPs, we know that the evaluation of $\sfff$ on $\bx$ will again admit a joint multivariate normal distribution, i.e.\ by defining the covariance matrix
\begin{equation}
    \bK_\bx = \begin{bmatrix}
        k(x_1, x_1) & \ldots & k(x_1, x_n) \\
        \vdots & \ddots & \vdots \\
        k(x_n, x_1) & \ldots & k(x_n, x_n),
    \end{bmatrix}
\end{equation}
we get that
\begin{equation}
    \begin{bmatrix}
        \sfff(x_1) \\ \vdots \\ \sfff(x_n)
    \end{bmatrix} \sim \cN(\mathbf{0}, \bK_\bx),
\end{equation}
and therefore by incorporating the additive noise, it follows that
\begin{equation}
    \by|\bx \sim \cN(\mathbf{0}, \bK_\bx + \sigma^2_{\varepsilon}\bI_n).
\end{equation}

Firstly, using this distribution, we can find values of the hyperparameters $\theta = (\sigma_\sfff, \ell, \sigma_\varepsilon)$ that maximise $\log p(\by|\bx)$. This corresponds to finding their maximum likelihood estimates under $p(\by|\bx)$, and can be achieved for example following the gradient descent procedure outlined in the following algorithm.
\begin{algorithm}[h]
    \caption{Maximum likelihood calibration of $\theta = (\sigma_\sfff, \ell, \sigma_\varepsilon)$ with gradient descent}
\begin{algorithmic}[1]
    \STATE {\bfseries Input:} $\bx, \by, \theta_0, \eta > 0, N \in\NN$
    \FOR{$t\in\{1, \ldots, N\}$}
        \STATE Compute $\log p(\by|\bx) = \log \cN(\by|\mathbf{0}, \bK_\bx + \sigma^2_{\varepsilon}\bI_n)$
        \STATE Take $\theta_t \leftarrow \theta_{t-1} + \eta \nabla_\theta\log p(\by|\bx) $
    \ENDFOR
    \STATE {\bfseries Return:} $\theta_\tau$
\label{alg:mle}
\end{algorithmic}
\end{algorithm}
Such a procedure offers a principled and robust way to calibrate the model hyperparameters against observations. Figure~\ref{fig:calibrated-go} shows a plot of the GP mean and confidence region with sample paths that have been calibrated against observations. Calibrating the prior against $(\bx, \by)$ yields a smaller variance $\sigma_\sfff^2 = 0.43$, which translates in a tighter confidence region, and a larger lengthscale $\ell = 0.17$, which translates into slower variations of the sample paths.
\begin{figure}[h]
    \centering
    \vspace*{-1em}
    \includegraphics[width=0.7\linewidth]{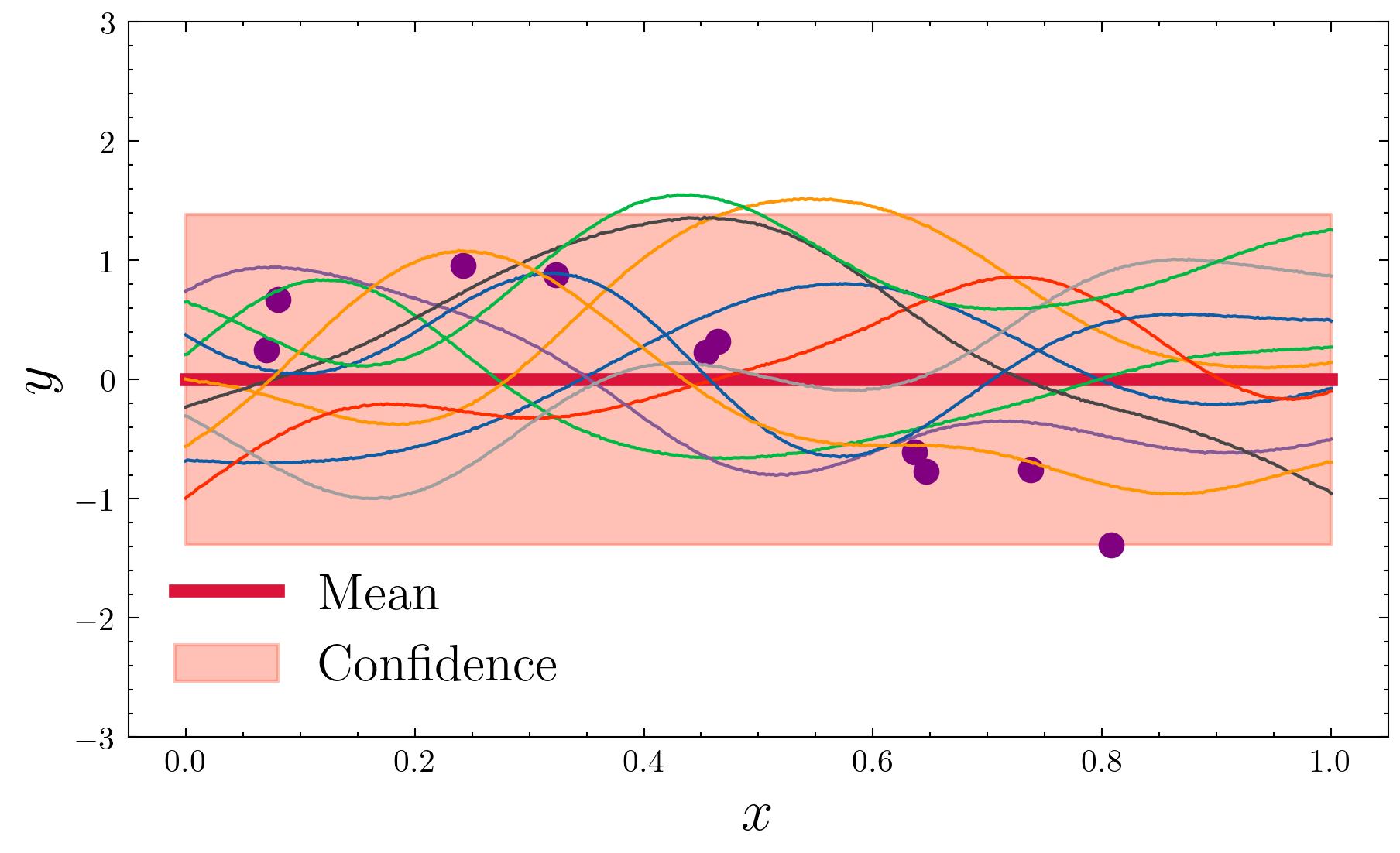}
    \vspace*{-1em}
    \caption{Plot of the $\sfff(x) \sim \GP(0, k)$ prior mean function and 95\% confidence region on the \texttt{grid} for calibrated hyperparameters $\sigma_\sfff^2 = 0.43$ and $\ell = 0.17$, with 10 sample paths and observed data $(\bx, \by)$.}
    \vspace*{-1em}
    \label{fig:calibrated-go}
\end{figure}

Second, we can update the prior $\sfff(x)$ placed over the predictive model into a posterior informed with the observed data $\sfff(x)|\by$. Formally, the mean and covariance of the posterior can be computed analytically following
\begin{equation}
    \left\{
    \begin{aligned}
        \begin{split}
            & \sfff(x)|\by \sim \GP(\bar m, \bar k) \\
            & \bar m(x) = k(x, \bx) (\bK_\bx + \sigma_\varepsilon^2\bI_n)^{-1}\by \\
            & \bar k(x, x') = k(x, x') - k(x, \bx) (\bK_\bx + \sigma_\varepsilon^2\bI_n)^{-1}k(\bx, x'),
        \end{split}
    \end{aligned}
    \right.
\end{equation}
where $k(x, \bx) = \begin{bmatrix} k(x, x_1) & \ldots & k(x, x_n)\end{bmatrix}$ and $k(\bx, x') = \begin{bmatrix} k(x_1, x') & \ldots & k(x_n, x')\end{bmatrix}^\top$.

Since this is still a GP, we can plot it on the \texttt{grid} using the exact same procedure as for the prior. Namely, define the posterior mean vector and posterior covariance matrix on \texttt{grid}
\begin{equation}
    \bar\bm = \begin{bmatrix}\bar m(z_1) \\ \vdots \\ \bar m(z_p)\end{bmatrix}\qquad     \bar\bK = \begin{bmatrix}
        \bar k(z_1, z_1) & \ldots & \bar k(z_1, z_p) \\
        \vdots & \ddots & \vdots \\
        \bar k(z_p, z_1) & \ldots & \bar k(z_p, z_p).
    \end{bmatrix}.
\end{equation}

Then
\begin{equation}
    \begin{bmatrix}\sfff(z_1)|\by \\ \vdots \\ \sfff(z_p)|\by\end{bmatrix}\sim \cN(\bar \bm, \bar\bK).
\end{equation}
Therefore, the posterior mean vector $\bar \bm$ provides the mean of $\sfff(x)|\by$ at every grid point and the diagonal elements of the covariance matrix $\bar\bK$ provide the standard deviation of $\sfff(x)|\by$ at every grid point (which does not simplify to $\sigma_\sfff$ this time). Figure~\ref{fig:gp-posterior} illustrates how updating the prior with 1, 2 and 10 observations affects the posterior mean and the posterior credible interval.
\begin{figure}[h]
    \centering
    \includegraphics[width=\linewidth]{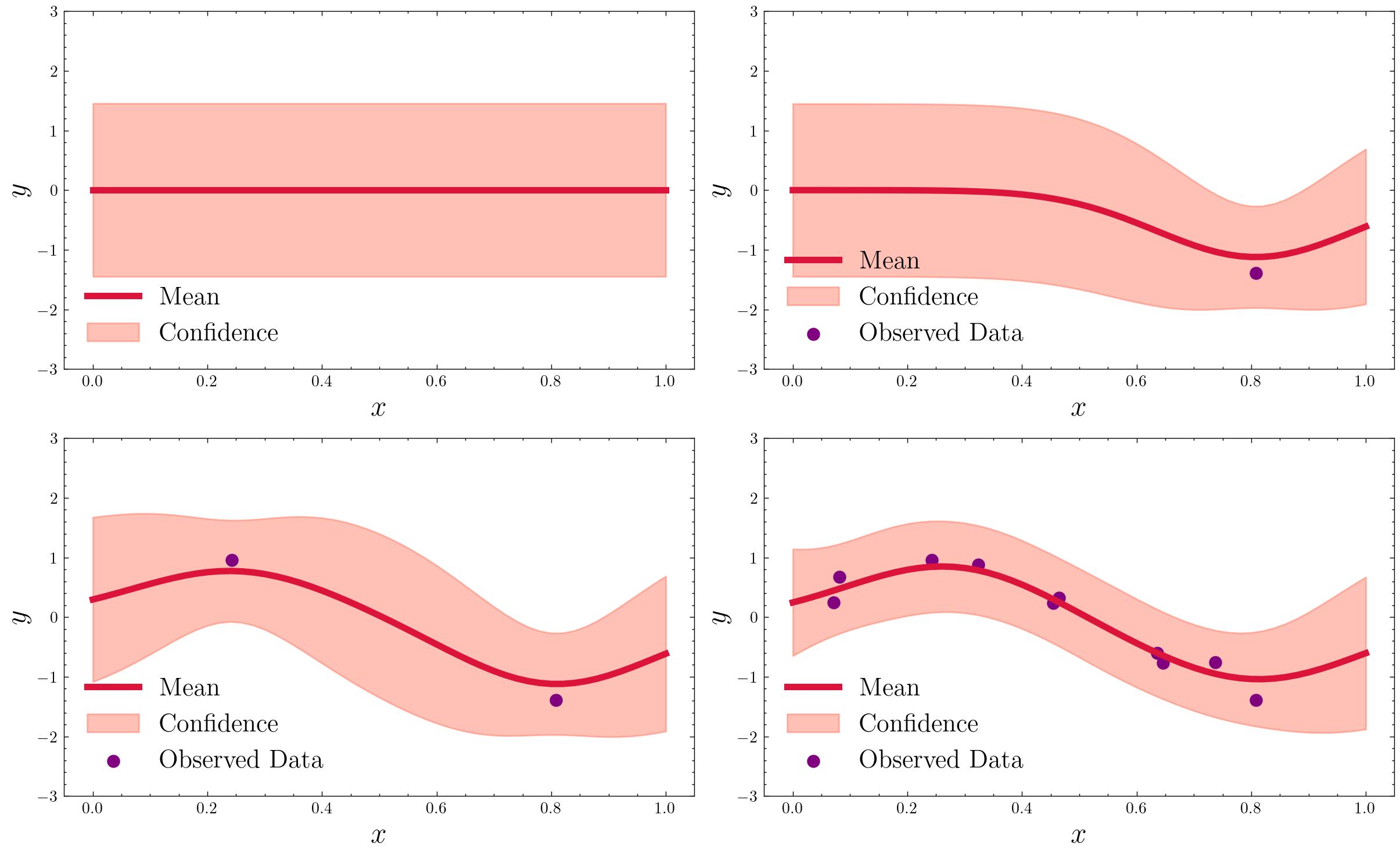}
    \caption{Plots of the posterior mean function and posterior 95\% credible interval when updating the prior with 0, (top-left), 1 (top-right), 2 (bottom-left) and 10 (bottom-right) observations.}
    \label{fig:gp-posterior}
    \vspace*{-2em}
\end{figure}

Furthermore, we are also able to draw samples from the multivariate normal distribution $\cN(\bar\bm, \bar\bK)$ by taking
\begin{equation}
    \bar \bff = \bar\bm + \bar\bK^{1/2}\bu, \enspace \bu\sim\cN(\mathbf{0}, \bI_p).
\end{equation}
Such a sample $\bar \bff\in\RR^p$ corresponds to a sample path from the posterior GP $\sfff(x)|\by$ evaluated on the grid points. Figure~\ref{fig:gp-posterior-samples} shows the plots of 10 sample paths from this multivariate normal distribution. Unlike for the prior, the sample paths are now concentrated around the observations and therefore provide reasonable candidates for a predictive model $x\to y$. In other words, the posterior GP induces a probability distribution over functions that provide a sound fit to observations, and are reasonable predictive models.
\begin{figure}[h]
    \centering
    \includegraphics[width=0.8\linewidth]{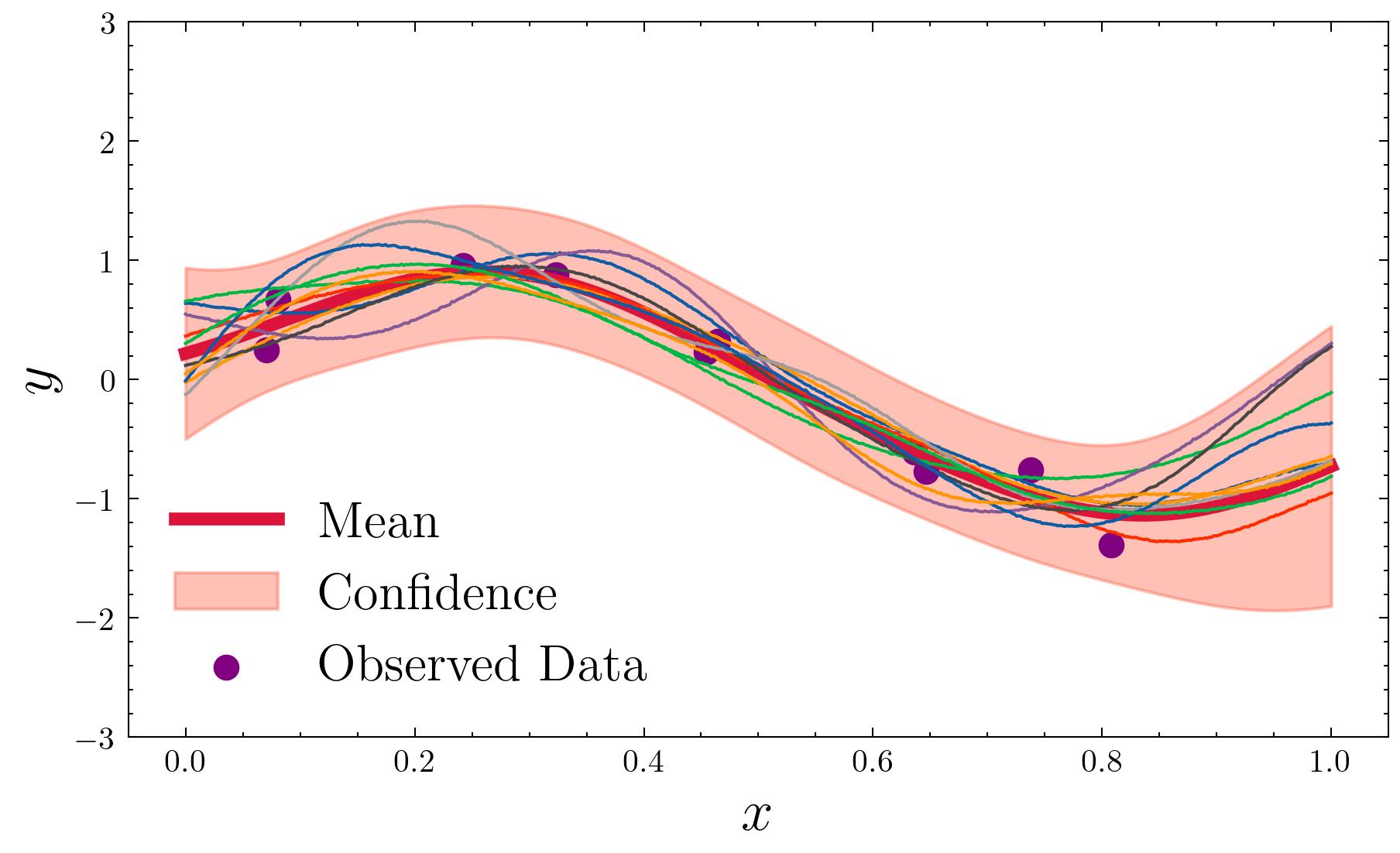}
    \caption{Plots of the posterior mean function and posterior 95\% credible interval + plots of 10 sample paths from $\bar\bff\sim\cN(\bar\bm, \bar\bK)$.}
    \label{fig:gp-posterior-samples}
\end{figure}

\subsection{Matérn covariance}\label{appendix:matern-covariance}

The Matérn covariances are a class of stationary covariance functions widely used in spatial statistics. The Matérn-$\nu$ covariance between two points $x, x'\in\RR$ is given by
\begin{equation}
    C_\nu(x, x') = \sigma^2 \frac{2^{1-\nu}}{\Gamma(\nu)}\left(\sqrt{2\nu}\frac{|x-x'|}{\ell}\right)^{\nu}K_\nu\left(\sqrt{2\nu}\frac{|x-x'|}{\ell}\right),
\end{equation}
where $\Gamma$ is the gamma function, $K_\nu$ is the modified Bessel function, $\sigma^2$ is a variance parameter and $\ell > 0$ is a lengthscale hyperparameter.

The covariance function expression considerably simplifies for $\nu = p + 1/2$ where $p\in\NN$. For example, for $\nu=1/2$ ($p=0$), $\nu=3/2$ ($p=1$) and $\nu = 5/2$ ($p=2$) we have
\begin{align}
    C_{1/2}(x, x') & = \sigma^2\exp\left(-\frac{|x-x'|}{\ell}\right) \\
    C_{3/2}(x, x') & = \sigma^2\left(1 + \sqrt{3}\frac{|x-x'|}{\ell}\right)\exp\left(-\sqrt{3}\frac{|x-x'|}{\ell}\right) \\
    C_{5/2}(x, x') & =  \sigma^2\left(1 + \sqrt{5}\frac{|x-x'|}{\ell} + \frac{5|x - x'|^2}{3\ell^2}\right)\exp\left(-\sqrt{5}\frac{|x-x'|}{\ell}\right)
\end{align}

When using the Matérn covariance as the covariance function of a GP, larger values of $\nu$ bestow greater smoothness on the GP sample paths. Namely, when $\nu = 1/2$ sample paths are continuous, when $\nu = 3/2$ sample paths are differentiable and when $\nu = 5/2$ sample paths are twice differentiable. More generally, for a Matérn-$p+1/2$ covariance function, sample paths from a GP are $p$ times continuously differentiable (with the convention that $0$ times means simply continuous). Figure~\ref{fig:matern-sample-paths} shows sample paths drawn from GPs with Matérn kernels for different values of $\nu$.

\begin{figure}
    \centering
    \includegraphics[width=\linewidth]{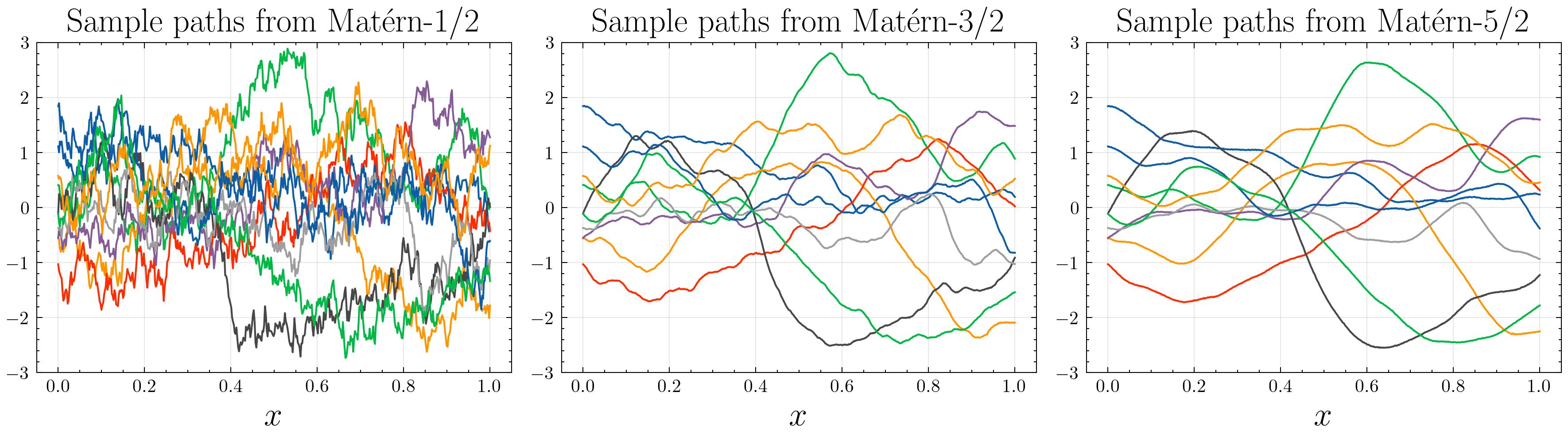}
    \caption{Example of sample paths drawn from a GP with kernel set as a Matérn-1/2, Matérn-3/2 and Matérn-5/2.}
    \label{fig:matern-sample-paths}
\end{figure}

In the limit where $\nu \to\infty$, the Matérn covariance converges to the squared exponential covariance given by
\begin{equation}
    C_\infty(x, x') = \sigma^2 \exp\left(-\frac{|x - x'|^2}{2\ell^2}\right).
\end{equation}
The sample paths of a GP with a squared exponential covariance function are infinitely differentiable.

When $x, x'\in\RR^d$, the distance $|x - x'|$ can be substituted by the norm $\|x - x'\| = \sqrt{\sum_{i=1}^d (x_i - x_i')^2}$. The covariance is called an anisotropic --- or automatic relevance determination --- kernel when each dimension has its own independent lengthscale parameter $\ell_i > 0$. For example, the Matérn-$1/2$ and Matérn-$3/2$ anisotropic kernels write
\begin{align}
    C_{1/2}(x, x') & = \sigma^2\exp\left(-\sqrt{\sum_{i=1}^d \frac{(x_i - x_i')^2}{\ell_i^2}}\right) \\
    C_{3/2}(x, x') & = \sigma^2\left(1 + \sqrt{3\sum_{i=1}^d \frac{(x_i - x_i')^2}{\ell_i^2}}\right)\exp\left(-\sqrt{3\sum_{i=1}^d \frac{(x_i - x_i')^2}{\ell_i^2}}\right).
\end{align}

\newpage
\section{Supporting materials for FaIRGP derivation}\label{eq:fairgp-derivation}

\subsection{Useful results}

\begin{lemma}\label{lemma:mintrick}
    Let $a, b > 0$ and $u, v \geq 0$. We have
    \begin{equation}
        (a + b) \min(u, v)  - (au + bv) = \begin{cases}
            -b|u - v| & \text{if } u \leq v \\
            -a|u - v| & \text{if } u \geq v.
        \end{cases}
    \end{equation}
\end{lemma}
\begin{proof}
If $u \leq v$,
\begin{align*}
    (a + b) \min(u, v)  - (au + bv) & = (a + b) u - au - bv \\
    & = bu - bv \\
    & = -b |u - v|.
\end{align*}

If $u \geq v$,
\begin{align*}
    (a + b) \min(u, v)  - (au + bv) & = (a + b) v - au - bv \\
    & = av - au \\
    & = -a |u - v|.
\end{align*}
\end{proof}

\subsection{Derivation of FaIRGP}\label{appendix:derivation-of-fairgp}

Consider a $k$-box EBM specified by the stochastic temperature response model
\begin{equation}
    \d\sfbX(t) = \bA\sfbX(t)\d t + \bb \sfF(t)\d t,
\end{equation}
where $\sfF(t) \sim \GP(F, K)$, a forcing feedback vector $\bb = \begin{bmatrix} 1/C_1 & 0 & \ldots & 0 \end{bmatrix}^\top$, and a forcing feedback matrix given by
\begin{equation*}
\resizebox{\linewidth}{!}{$
    \textbf{A} =  \begin{bmatrix}
\nicefrac{-(\kappa_1 + \kappa_2)}{C_1} & \nicefrac{\kappa_2}{C_1} & 0 & \dots & 0 & 0  & 0\\ 
\nicefrac{\kappa_2}{C_2} & \nicefrac{-(\kappa_2+\kappa_3)}{C_2} & \nicefrac{\kappa_3}{C_2} & \dots & 0 & 0 & 0\\
0 & \nicefrac{\kappa_3}{C_3} & \nicefrac{-(\kappa_3+\kappa_4)}{C_3} & \dots & 0 & 0 & 0\\
\vdots & \vdots & \vdots & \ddots & \vdots & \vdots & \vdots \\
0 & 0 & 0 & \dots & \nicefrac{-(\kappa_{k-2} + \kappa_{k-1})}{C_{k-2}} & \nicefrac{\kappa_{k-1}}{C_{k-2}} & 0 \\
0 & 0 & 0 & \dots &  \nicefrac{\kappa_{k-1}}{C_{k-1}} & \nicefrac{-(\kappa_{k-1} + \epsilon\kappa_k)}{C_{k-1}} & \nicefrac{\epsilon\kappa_k}{C_{k-1}} \\
0 & 0 & 0 & \dots & 0 &  \nicefrac{\kappa_k}{C_k} & \nicefrac{-\kappa_k}{C_k}.
\end{bmatrix}
$
}
\end{equation*}

Denote $\bA = \bPhi \bD\bPhi^{-1}$ the diagonalisation of the feedback matrix, where $\bD$ is a diagonal matrix with diagonal elements
\begin{equation}
    \bD = \operatorname{Diag}(D_1, \ldots, D_k).
\end{equation}

Then we have
\begin{align}
    & \d\sfbX(t)  = \bA\sfbX(t)\d t + \bb \sfF(t)\d t \\
    \Rightarrow & \bPhi^{-1}\d \sfbX(t) = \bD\bPhi^{-1}\sfbX(t)\d t + \bPhi^{-1}\bb \sfF(t)\d t\\
    \Rightarrow & \d \left[\bPhi^{-1}\sfbX(t)\right] = \bD \left[\bPhi^{-1}\sfbX(t)\right]\d t + \left[\bPhi^{-1}\bb \right]\sfF(t) \d t.
\end{align}

Therefore, by using the notations from \cite{millar2017modified}
\begin{align}
    D_i & = -1/d_i \\
    \bPhi^{-1}\sfbX(t) & = \begin{bmatrix} \sfS_1(t) & \ldots & \sfS_k(t)\end{bmatrix}^\top \\
    \bPhi^{-1}\bb & = \begin{bmatrix} q_1/d_1 & \ldots & q_k/d_k\end{bmatrix}^\top,
\end{align}
the diagonalised system can be written as an impulse response system, where for $i \in\{1, \ldots, k\}$ we have
\begin{equation}
    \boxed{
    \d \sfS_i(t)  = -\frac{1}{d_i}\sfS_i(t)\d t + \frac{q_i}{d_i}\sfF(t)\d t = \frac{1}{d_i}(q_i \sfF(t) - \sfS_i(t))\d t.
    }
\end{equation}

Consider now the thermal response $\sfS_i(t)$ of the $i$\textsuperscript{th} stochastic differential equation (SDE). This thermal response is given by
\begin{equation}
    \sfS_i(t) = \frac{q_i}{d_i}\int_0^t \sfF(s)e^{-(t-s)/d_i}\d s,
\end{equation}
where we recall that $\sfF(t) \sim \GP(F, K)$. Because GPs are closed under linear transformations, $\sfS_i(t)$ must also be GPs for all $i\in\{1, \ldots, k\}$. Their mean functions can be computed following
\begin{align}
    m_i(t) := \EE[\sfS_i(t)] & = \EE\left[\frac{q_i}{d_i}\int_0^t \sfF(s)e^{-(t-s)/d_i}\d s\right] \\
    & = \frac{q_i}{d_i}\int_0^t \EE[\sfF(s)]e^{-(t-s)/d_i}\d s \\
    & = \frac{q_i}{d_i}\int_0^t F(s)e^{-(t-s)/d_i}\d s.
\end{align}

Their cross-covariance functions can be computed following
\begin{align}
    k_{ij}(t, t') & := \operatorname{Cov}(\sfS_i(t), \sfS_j(t')) \\
    & = \operatorname{Cov}\left(\frac{q_i}{d_i}\int_0^t \sfF(s)e^{-(t-s)/d_i}\d s, \frac{q_j}{d_j}\int_0^{t'} \sfF(s')e^{-(t'-s')/d_j}\d s'\right) \\
    & = \frac{q_i q_j}{d_i d_j}\int_0^t\int_0^{t'} \operatorname{Cov}(\sfF(s), \sfF(s')) e^{-(t-s)/d_i}e^{-(t'-s')/d_j}\d s \d s' \\
    & = \frac{q_i q_j}{d_i d_j}\int_0^t\int_0^{t'} K(s, s') e^{-(t-s)/d_i}e^{-(t'-s')/d_j}\d s \d s'. \\
\end{align}

So we have $\sfS_i(t)\sim\GP(m_i, k_{ii})$ for any $i\in\{1, \ldots, k\}$. Finally, if we define $\sfT(t) = \sum_{i=1}^k \sfS_i(t)$, since the sum of Gaussians is still a Gaussian, we know that $\sfT(t)$ must also be a GP. And we can compute its mean function following
\begin{align}
    m_\sfT(t) := \EE[\sfT(t)]  = \EE\left[\sum_{i=1}^k\sfS_i(t)\right] = \sum_{i=1}^k\EE[\sfS_i(t)] = \sum_{i=1}^k m_i(t),
\end{align}
and its covariance function following
\begin{align}
    k_\sfT(t, t') := \operatorname{Cov}(\sfT(t), \sfT(t')) &  = \operatorname{Cov}\left(\sum_{i=1}^k \sfS_i(t), \sum_{j=1}^k \sfS_j(t')\right) \\
    & = \sum_{i=1}^k\sum_{j=1}^k \operatorname{Cov}(\sfS_i(t), \sfS_j(t')) \\
    & = \sum_{i=1}^k\sum_{j = 1}^k k_{ij}(t,t').
\end{align}

We conclude that
\begin{equation}
\boxed{
    \left\{
    \begin{aligned}
        \begin{split}
            & \sfT(t)  \sim \GP(m_\sfT, k_\sfT) \\
            & m_\sfT(t) = \sum_{i=1}^k m_i(t) \\
            & k_\sfT(t, t')  = \sum_{i=1}^k\sum_{j = 1}^k k_{ij}(t,t').
        \end{split}
    \end{aligned}
    \right.
}
\end{equation}

\subsection{Accounting for climate internal variability in FaIRGP}\label{appendix:internal-variability-proof}

Consider again the same problem, but introducing an additional white noise term in the temperature response SDE
\begin{equation}\label{eq:A25}
    \d\sfbX(t) = \bA\sfbX(t)\d t + \bb \sfF(t)\d t + \sigma \bb \d\sfB(t),
\end{equation}
where $\sfB(t)$ denotes the standard Brownian motion.

Following the same derivation steps we get
\begin{align}
    & \d\sfbX(t)  = \bA\sfbX(t)\d t + \bb \sfF(t)\d t + \sigma\bb\d\sfB(t) \\
    \Rightarrow & \bPhi^{-1}\d \sfbX(t) = \bD\bPhi^{-1}\sfbX(t)\d t + \bPhi^{-1}\bb F(t)\d t + \bPhi^{-1}\bc \d\sfB(t) \\
    \Rightarrow & \d \left[\bPhi^{-1}\sfbX(t)\right] = \bD \left[\bPhi^{-1}\sfbX(t)\right]\d t + \left[\bPhi^{-1}\bb \right]\sfF(t) \d t + \sigma \left[\bPhi^{-1}\bb \right] \d \sfB(t).
\end{align}
which gives an impulse response form
\begin{equation}
\boxed{
    \d \sfS_i(t)  = \frac{1}{d_i}(q_i \sfF(t) - \sfS_i(t))\d t + \sigma\frac{q_i}{d_i}\d \sfB(t),\enspace \forall i\in\{1, \ldots, k\}.}
\end{equation}

The solution to the $i$\textsuperscript{th} SDE is now given by
\begin{align}
    \sfS_i(t) & = \underbrace{\frac{q_i}{d_i}\int_0^t \sfF(s)e^{-(t-s)/d_i}\d s}_{\sfS_i^\circ(t)} + \sigma \underbrace{\frac{q_i}{d_i}\int_0^t e^{-(t-s)/d_i}\d\sfB(s)}_{\upxi_i(t)} \\
    & = \sfS_i^\circ(t) + \sigma\upxi_i(t).
\end{align}

$\sfS_i^\circ(t)$ corresponds to the GP we have obtained in the previous derivation without white noise, i.e.\ $\sfS_i^\circ(t)\sim\GP(m_i, k_{ii})$. $\upxi_i(t)$ is also a GP, with mean zero, and cross-covariance function given by
\begin{align*}
    \operatorname{Cov}(\upxi_i(t), \upxi_j(t')) & = \EE\left[\left(\upxi_i(t) - \EE[\upxi_i(t)]\right)\left(\upxi_j(t') - \EE[\upxi_j(t')]\right)\right] \\
    & = \EE[\upxi_i(t)\upxi_j(t')] \\
    & = \EE\left[\left(\frac{q_i}{d_i}\int_0^t e^{-(t-s)/d_i}\d\sfB(s)\right)\left(\frac{q_j}{d_j}\int_0^{t'} e^{-(t'-s')/d_j}\d\sfB(s')\right)\right] \\
    & = \frac{q_iq_j}{d_id_j}e^{-t/d_i - t'/d_j} \EE\left[\int_0^t e^{s/d_i}\d\sfB(s)\int_0^{t'}e^{s'/d_j}\d\sfB(s')\right] \\
    & = \frac{q_iq_j}{d_id_j} e^{-t/d_i - t'/d_j} \int_0^{\min(t, t')}e^{(1/d_i + 1/d_j)s} \d s \\
    & = \frac{q_iq_j}{d_id_j} e^{-t/d_i - t'/d_j} \frac{d_id_j}{d_i + d_j}\left(e^{\frac{d_i + d_j}{d_id_j}\min(t, t')} - 1\right) \\
    & = \frac{q_iq_j}{d_i + d_j} e^{-(d_j t + d_i t') / d_id_j} \left(e^{\frac{d_i + d_j}{d_id_j}\min(t, t')} - 1\right) \\
    & = \begin{cases}
        \frac{q_iq_j}{d_i + d_j} \left(e^{-|t-t'| / d_i} - e^{-(d_j t + d_i t') / d_id_j}\right) & \text{if } t \leq t' \\
        \frac{q_iq_j}{d_i + d_j} \left(e^{-|t-t'| / d_j} - e^{-(d_j t + d_i t') / d_id_j}\right) & \text{if } t \geq t'
    \end{cases} & (Lemma~\ref{lemma:mintrick})\\
    & \sim \begin{cases}
        \frac{q_iq_j}{d_i + d_j} e^{-|t-t'| / d_i} & \text{if } t \leq t' \\
        \frac{q_iq_j}{d_i + d_j}e^{-|t-t'| / d_j} & \text{if } t \geq t'
    \end{cases} \qquad \text{when } t\gg d_i \text{ or } t'\gg d_j.
\end{align*}

Therefore, if we define
\begin{equation}
    \gamma_{ij}(t, t') = \frac{q_iq_j}{d_i + d_j}\exp\left(\frac{-|t-t'|}{d_i \mathbbm{1}_{\{t \leq t'\}} + d_j \mathbbm{1}_{\{t > t'\}}}\right),
\end{equation}
which simplifies when $i = j$ to
\begin{equation}
    \gamma_i(t, t') = \gamma_{ii}(t, t') = \frac{q_i^2}{2d_i}\exp\left(-\frac{|t-t'|}{d_i}\right),
\end{equation}
we obtain that in the long time regime, we can approximate $\upxi_i(t)\sim\GP(0, \gamma_i)$. And because $\sfS_i^\circ(t)$ and $\upxi_i(t)$ are independent processes, we obtain that
\begin{equation}
    \boxed{\sfS_i(t) \sim \GP(m_i, k_{ii} + \sigma^2 \gamma_i).} 
\end{equation}

Finally, if we take again $\sfT(t) = \sum_{i=1}^k \sfS_i(t)$, then $\sfT(t)$ must be a GP. Its mean function is given by
\begin{align}
    m_\sfT(t) := \EE[\sfT(t)]  = \EE\left[\sum_{i=1}^k\sfS_i(t)\right] = \sum_{i=1}^k\EE[\sfS_i(t)] = \sum_{i=1}^k m_i(t),
\end{align}
and its covariance function is given by
\begin{align}
    k_\sfT(t, t') & := \operatorname{Cov}(\sfT(t), \sfT(t')) \\
    &  = \operatorname{Cov}\left(\sum_{i=1}^k \sfS_i(t), \sum_{j=1}^k \sfS_j(t')\right) \\
    & = \sum_{i, j = 1}^k \operatorname{Cov}(\sfS_i(t), \sfS_j(t')) \\
    & = \sum_{i, j = 1}^k \operatorname{Cov}\left(\sfS_i^\circ(t) + \sigma\eta_i(t), \sfS_j^\circ(t') + \sigma\eta_j(t')\right) \\
    & = \sum_{i, j = 1}^k \operatorname{Cov}(\sfS_i^\circ(t), \sfS_j^\circ(t')) + \sigma^2\operatorname{Cov}(\eta_i(t), \eta_j(t')) & (\sfS_i^\circ, \sfS_j^\circ \indep \upxi_i, \upxi_j) \\
    & = \sum_{i, j = 1}^k k_{ij}(t, t') + \sigma^2 \sum_{i, j = 1}^k \gamma_{ij}(t, t').
\end{align}

However, if $t \leq t'$ we have
\begin{align}
    \sum_{i,j=1}^k \gamma_{ij}(t, t') & = \sum_{i,j=1}^k \frac{q_iq_j}{d_i + d_j} \exp\left(-\frac{|t-t'|}{d_i}\right) \\
    & = \sum_{i=1}^k \left(\sum_{j=1}^k \frac{q_iq_j}{d_i + d_j}\right)\exp\left(-\frac{|t-t'|}{d_i}\right) \\
    & = \sum_{i=1}^k \underbrace{\frac{2d_i}{q_i^2}\left(\sum_{j=1}^k \frac{q_iq_j}{d_i + d_j}\right)}_{\nu_i} \frac{q_i^2}{2d_i}\exp\left(-\frac{|t-t'|}{d_i}\right) \\
    & = \sum_{i=1}^k \nu_i\gamma_i(t, t').
\end{align}

When $t \geq t'$ we can refactor terms similarly into a sum over $j$. By symmetry of the indices we conclude that
\begin{equation}
\boxed{
    \left\{
    \begin{aligned}
        \begin{split}
            & \sfT(t)  \sim \GP(m_\sfT, k_\sfT) \\
            & m_\sfT(t) = \sum_{i=1}^k m_i(t) \\
            & k_\sfT(t, t')  = \sum_{i,j=1}^k k_{ij}(t, t') + \sigma^2 \sum_{i=1}^k \nu_i \gamma_i(t, t') \\
            & \nu_i = \sum_{j=1}^k \frac{2d_iq_j}{q_i(d_i + d_j)} \\
        \end{split}
    \end{aligned}
    \right.
}
\end{equation}

\subsection{Accounting for unforced fluctuations in top-of-atmosphere radiative flux}

In addition to the variability arising from unforced transfers of energy modelled in the last section, \citeA{cummins2020optimal} propose to also account for unforced fluctuations in the top-of-atmosphere radiative flux. Let us denote by $\tilde \sfF(t)$ the radiative forcing model proposed by \citeA{cummins2020optimal}, to distinguish it from the FaIRGP prior $\sfF(t)$. They propose to model $\tilde \sfF(t)$ as an Ornstein-Uhlenbeck process,
\begin{equation}
    \d \tilde \sfF(t) = \theta(F(t) - \tilde \sfF(t))\d t + \varsigma \d \tilde\sfB(t)
\end{equation}
where $\theta > 0$ is the autocorrelation parameter, $\varsigma > 0$ is the white noise variance, $\tilde\sfB(t)$ is another standard Brownian motion (independent from $\sfB(t)$ in the temperature response SDE), and $F(t)$ is a deterministic forcing component.

Assuming that $F(0) = 0$, we know that the solution of this SDE is given by
\begin{align}
    \tilde \sfF(t) & = \underbrace{ \int_0^t \theta F(s) e^{-\theta (t - s)} \d s}_{\tilde F(t)} + \varsigma\underbrace{\int_0^t e^{-\theta (s - t)}\d \tilde\sfB(s)}_{\upeta(t)} \\
    & = \tilde F(t) + \varsigma\upeta(t).
\end{align}

$\tilde F(t)$ is a deterministic quantity, whereas $\upeta(t)$ is as before a GP with mean zero, and a covariance function $\gamma_\theta$ which can be approximated in the long time regime by
\begin{equation}
    \gamma_\theta(t, t') = \frac{1}{2\theta}\exp\left(-\theta |t - t'|\right).
\end{equation}
Therefore, we get that in the long term regime
\begin{equation}
    \tilde\sfF(t) \sim \GP\left(\tilde F, \varsigma\gamma_\theta\right).
\end{equation}
This can be combined with the FaIRGP prior by substituting the deterministic forcing $F(t)$ by the FaIRGP forcing prior $\sfF(t)\sim \GP(F, K)$. Namely, if we define
\begin{equation}
    \tilde K(t, t') = \int_0^t \int_0^{t'} \theta^2 K(s, s') e^{-\theta(t - s)}e^{-\theta(t'-s')}\d s\d s',
\end{equation}
then we obtain a new FaIRGP prior over the forcing
\begin{equation}
    \tilde\sfF(t) \sim \GP\left(\tilde F, \tilde K + \varsigma\gamma_\theta\right),
\end{equation}
which accounts for the unforced fluctuations in the top-of-atmosphere radiative flux like in \cite{cummins2020optimal}.

\subsection{Analytical expression for FaIRGP probability distribution}\label{appendix:analytical-expressions}

A key benefit of FaIRGP is its complete mathematical tractability. This tractability in particular includes access to the analytical expression of: (\textit{i}) the marginal probability distribution over the training data, which can be used as an objective to tune the model hyperparameters and (\textit{ii}) of the probability distribution predicted over emulated temperatures, which can be useful for downstream applications of emulation.

In what follows, we assume access to a training set $\cD = \{\bt, \bE, \bT\}$ of size $n$, where we follow the notational conventions from Section~\ref{subsection:posterior-distribution}.

\subsubsection{Analytical expression of the marginal log-likelihood}

Using the prior mean and covariance of the FaIRGP prior over $\sfT(t)$, we define a prior mean vector and a prior covariance matrix over $\bt, \bE$ given by
\begin{align}
    \bm & = m_\sfT(\bt)  = \begin{bmatrix} m_\sfT(t_1) \\ \vdots \\ m_\sfT(t_n) \end{bmatrix}\\
    \bK & = k_\sfT(\bE, \bE) = \begin{bmatrix} k_\sfT(E_i, E_j) \end{bmatrix}_{1 \leq i, j \leq n}.
\end{align}

Further, consider the temperatures internal variability covariance matrix defined by
\begin{equation}
    \bGamma = \gamma_\sfT(\bt, \bt) = \begin{bmatrix} \gamma_\sfT(t_i, t_j) \end{bmatrix}_{1 \leq i, j \leq n}.
\end{equation}

Then the prior distribution over temperatures is exactly given by the multivariate normal distribution $\cN(\bm, \bK + \sigma^2 \bGamma)$. Let us denote $\bK_{\sigma^2} = \bK + \sigma^2 \bGamma$ for conciseness. Then, we can exactly evaluate the marginal probability of $\bT\in\RR^n$ under the prior distribution following
\begin{equation}
    p(\bT | \bE, \bt)  = \frac{1}{\sqrt{(2\pi)^n \det(\bK_{\sigma^2})}}\exp\left(-\frac{1}{2} (\bT - \bm)^\top \bK_{\sigma^2}^{-1}(\bT - \bm)\right),
\end{equation}
which can be used as a maximisation objective to tune model hyperparameters such that the observed temperatures $\bT$ have the greatest possible likelihood under the prior. In practice, we prefer working with the marginal log-likelihood for computational stability. It is given in closed-form by
\begin{equation}
    \boxed{\log p(\bT|\bE, \bt) = -\frac{1}{2}\left\{n \log(2\pi) + \log \det(\bK_{\sigma^2}) + \left(\bT - \bm\right)^\top \bK_{\sigma^2}^{-1}\left(\bT - \bm\right)\right\}}.
\end{equation}

Let $\theta$ denote a set of model parameters we wish to tune against simulated data $\bt, \bE, \bT$. For example, we can choose the internal variability variance and GP kernel parameters $\theta = \{\sigma, \sigma_\sfF, \ell_1, \ldots, \ell_d\}$, FaIR temperature response parameters $\theta = \{d_1, \ldots, d_k, q_1, \ldots, q_k\}$, the radiative forcing model parameters $\theta = \{\alpha^\chi_{\log}, \alpha^\chi_\text{lin}, \alpha^\chi_\text{sqrt}\}$, or any combination of these options. Let $\theta_0$ be the initial values we arbitrarily choose for these parameters, and $\eta > 0$ is a typically small positive number called the \emph{learning rate}. Then the gradient descent algorithm outlined below proposes a principled procedure to tune $\theta$ such that it maximises the marginal log-likelihood $\log p(\bT|\bE, \bt)$.

\begin{algorithm}[h]
    \caption{Maximum likelihood calibration of $\theta$ with gradient descent}
\begin{algorithmic}[1]
    \STATE {\bfseries Input:} $\bt, \bE, \bT, \theta_0, \eta > 0, N\in\NN$
    \FOR{$t\in\{1, \ldots, N\}$}
        \STATE Compute $\log p(\bT|\bE, \bt) = \log \cN(\bT|\bm, \bK + \sigma^2\bGamma)$
        \STATE Take $\theta_t \leftarrow \theta_{t-1} + \eta \nabla_\theta\log p(\bT|\bE, \bt) $
    \ENDFOR
    \STATE {\bfseries Return:} $\theta_N$
\label{alg:mle}
\end{algorithmic}
\end{algorithm}

The learning rate $\eta > 0$ controls the size of the steps taken in the direction of the gradient. A higher learning rate results in larger steps, while a lower learning rate leads to smaller steps. If the learning rate is too high, the algorithm may overshoot the maximum and fail to converge. On the other hand, if the learning rate is too low, the algorithm may take a long time to converge.

In the presented gradient descent algorithm, the algorithm stops after the maximum number of iterations $N$ has been reached. More sophisticated stopping criteria can be devised, such as stopping if the absolute change in the marginal log-likelihood falls below a predefined threshold. 

Finally, it may be that the size $n$ of the dataset $\cD = \{\bt, \bE, \bT\}$ is so large that the evaluation of $\log p(\bT|\bE, \bt)$ is in practice computationally prohibitive. In this case, a popular variant of gradient descent called \emph{mini-batch stochastic gradient descent} can be employed. At each iteration, instead of computing the marginal log-likelihood over the entire dataset, we compute it over a randomly selected subset of the training data, called a mini-batch. This mini-batch typically has a size $m \ll n$. Because only a subset of the data is used for each update, the updates are noisier, introducing more variability in the optimisation process. It is however more computationally efficient, especially for large datasets, as it processes only a small subset of the data in each iteration.

\subsubsection{Analytical expression of the posterior distribution}

Suppose that we want to emulate the global temperature response for a different emission scenario where we have access to greenhouse gas and aerosol emission data $E_1^*, \ldots, E^*_m$ at times $t_1^* < \ldots < t^*_m$. We concatenate them into
\begin{equation}
    \bt^* = \begin{bmatrix}t_1^* \\ \vdots \\ t_m^* \end{bmatrix}\in\RR^m,\qquad \bE^* = \begin{bmatrix}E_1^* \\ \vdots \\ E_m^* \end{bmatrix}\in\RR^{m\times d},
\end{equation}
where $d \geq 0$ is the number of emission agents $\chi_1, \ldots, \chi_d$ we observe emissions from.

Using the posterior mean and covariance functions when $\sfT(t)$ is updated with the training set $\cD$, we define a posterior mean vector and posterior covariance matrix over $\bt^*, \bE^*$ given by
\begin{align}
    \bar \bm^* & = \bar m_\sfT(\bt^*) \\
    \bar \bK^* & = \bar k_\sfT(\bE^*, \bE^*).
\end{align}

Further, consider the emulated temperatures internal variability covariance matrix defined by
\begin{equation}
    \bGamma^* = \gamma_\sfT(\bt^*, \bt^*).
\end{equation}

Then, the posterior distribution over emulated temperatures is exactly given by the multivariate normal distribution $\cN(\bar\bm^*, \bar\bK^* + \sigma^2\bGamma^*)$. Let us denote $\bar\bK^*_{\sigma^2} = \bar\bK^* + \sigma^2\bGamma^*$ for conciseness. This means that for a given temperature vector $\bT^*\in\RR^m$, we can exactly evaluate the probability of $\bT^*$ under the predicted distribution following
\begin{equation}
\boxed{
    p(\bT^* | \cD, \bE^*, \bt^*)  = \frac{1}{\sqrt{(2\pi)^m \det(\bar\bK^*_{\sigma^2})}}\exp\left(-\frac{1}{2} (\bT^* - \bar \bm^*)^\top \left(\bar \bK_{\sigma^2}^*\right)^{-1}(\bT^* - \bar\bm^*)\right).
}
\end{equation}

The vector $\bT^*$ may be a retained test scenario, in which case computing $p(\bT^* | \cD, \bE^*, \bt^*)$ provides an evaluation of the emulated posterior $\cN(\bar\bm^*, \bar\bK^*_{\sigma^2})$. $\bT^*$ may also simply correspond to temperature for which we would like to assess the probability under emission scenario $\{\bt^*, \bE^*\}$. A natural application is in detection attribution studies. Indeed, we can take $\bT^*$ to be observed historical temperatures, and $\{\bt^*, \bE^*\}$ to be a counterfactual historical scenario without anthropogenic forcing. By evaluating $p(\bT^* | \cD, \bE^*, \bt^*)$ we would be able to assess the probability of historical temperature observations to occur in a scenario without anthropogenic forcing. 

The posterior distribution can also be used to draw a sample $\bV$ if needed following
\begin{equation}
    \bV =  \left(\bar\bK^*_{\sigma^2}\right)^{\nicefrac{1}{2}} \bU + \bar\bm^*, \quad \bU\sim\cN(0, \bI_m).
\end{equation}

Finally, whilst the above is formulated for global mean surface temperatures, we can also derive analytical forms for the posterior distribution over the radiative forcing and for spatially-resolved emulation.

\newpage
\section{Complementary experimental results}\label{appendix:experiments}

\subsection{FaIRGP posterior sample paths over SSPs}\label{appendix:sample-paths-quizz}

\begin{figure}[h]
    \centering
    \includegraphics[width=\linewidth]{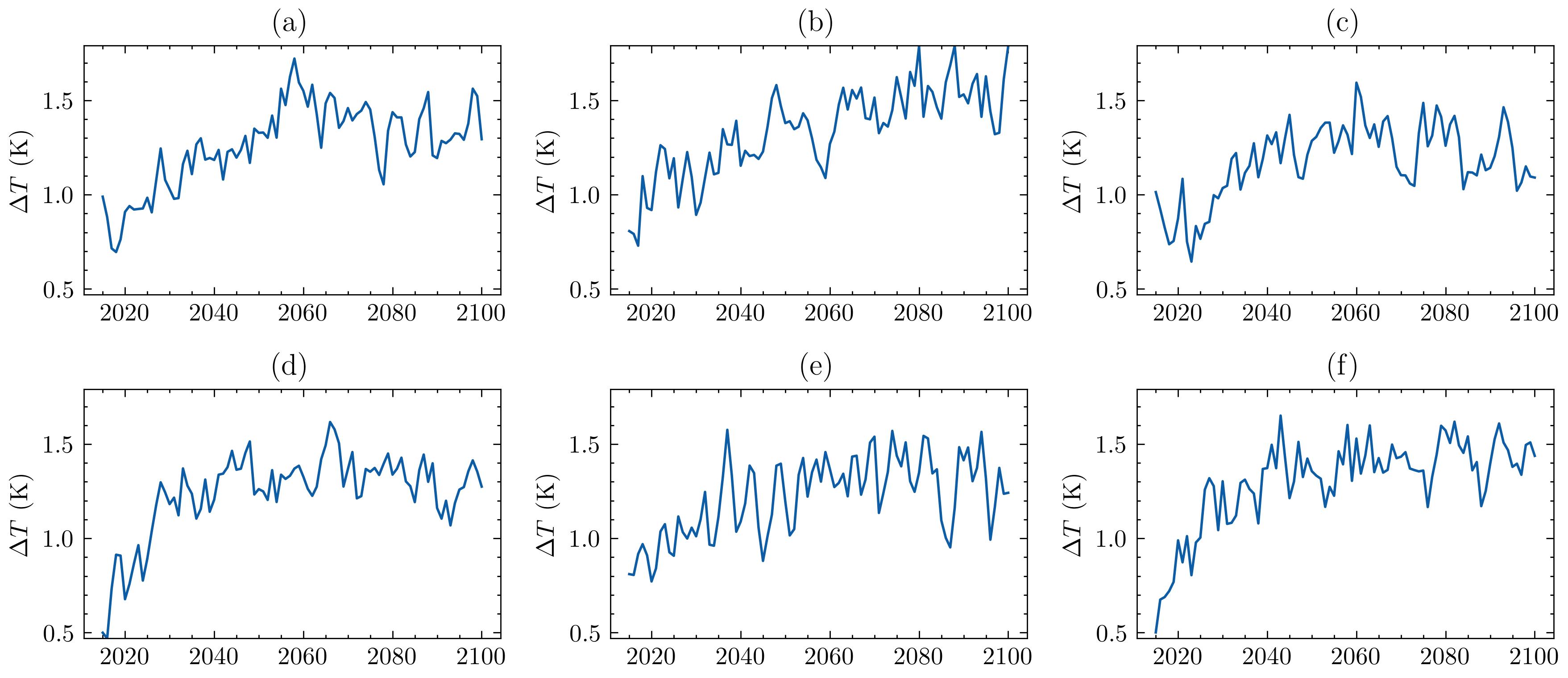}
    \caption{Example of 5 sample paths from the predicted posterior distribution over global annual mean surface temperature anomaly under \textit{SSP126} with FaIRGP + NorESM2-LM simulated global annual mean surface temperature anomaly. The posterior is conditioned on experiments $\cD_{\text{train}} = \{historical, SSP245, SSP370, SSP585\}$. The position of the NorESM2-LM simulated data is given at the end of the section.}
    \label{fig:ssp126-sample-paths}
\end{figure}
\begin{figure}[H]
    \centering
    \includegraphics[width=\linewidth]{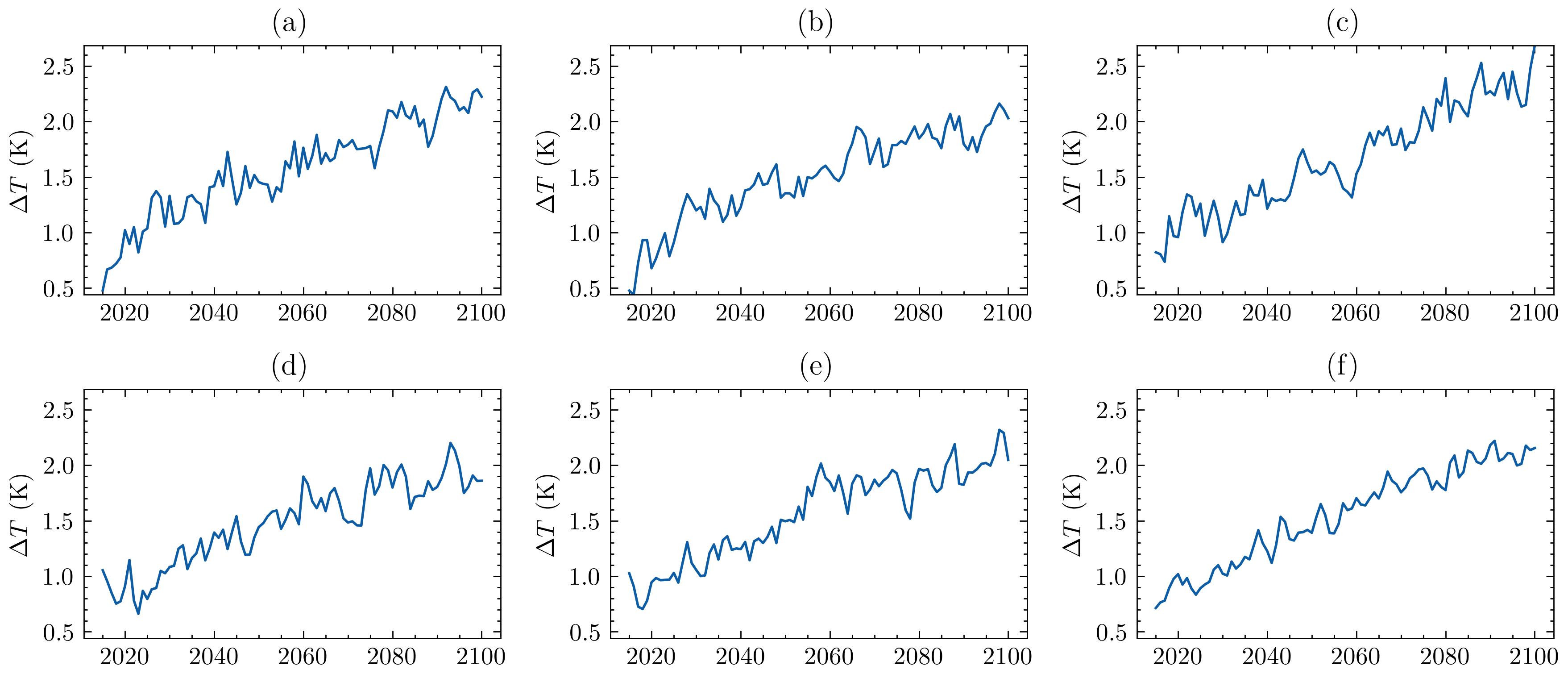}
    \caption{Example of 5 sample paths from the predicted posterior distribution over global annual mean surface temperature anomaly under \textit{SSP245} with FaIRGP + NorESM2-LM simulated global annual mean surface temperature anomaly. The posterior is conditioned on experiments $\cD_{\text{train}} = \{historical, SSP126, SSP370, SSP585\}$. The position of the NorESM2-LM simulated data is given at the end of the section.}
    \label{fig:ssp245-sample-paths}
\end{figure}
\begin{figure}[H]
    \centering
    \includegraphics[width=\linewidth]{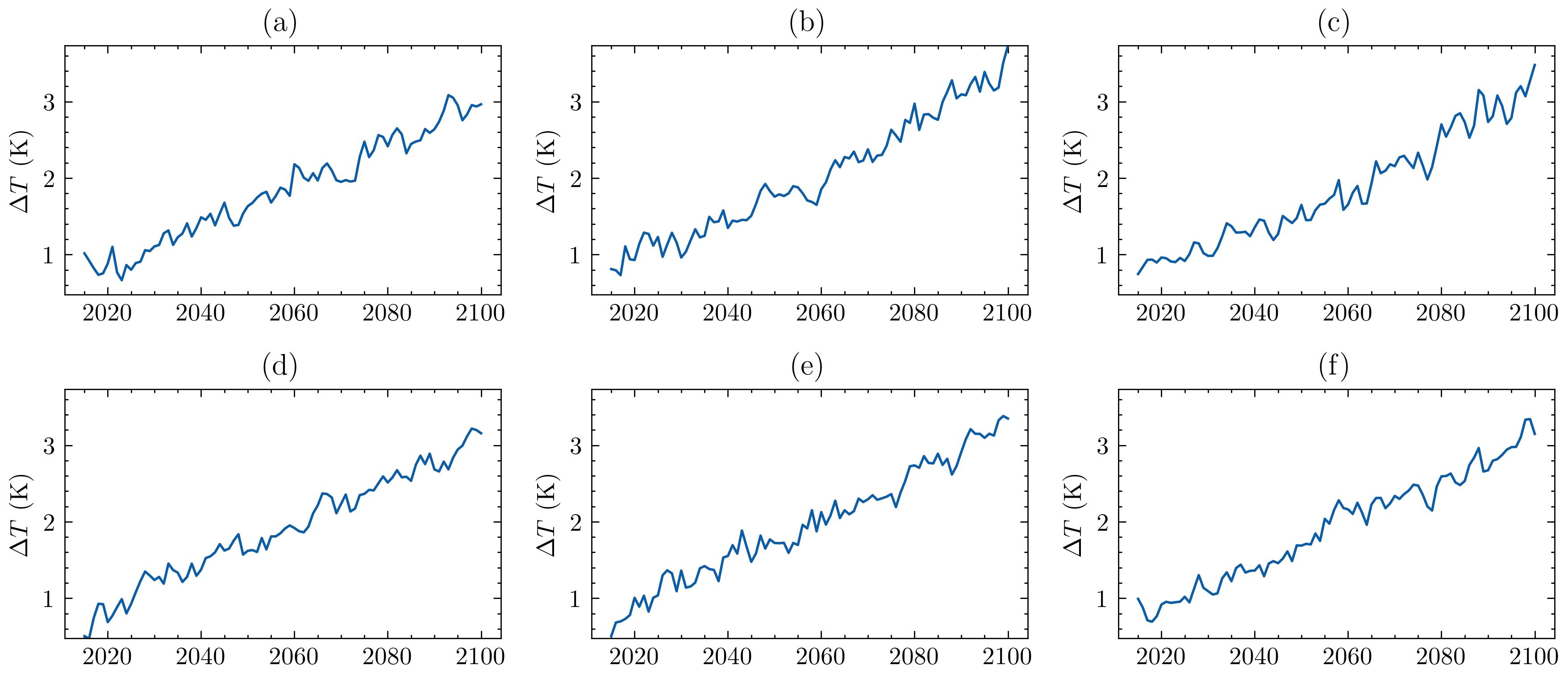}
    \caption{Example of 5 sample paths from the predicted posterior distribution over global annual mean surface temperature anomaly under \textit{SSP370} with FaIRGP + NorESM2-LM simulated global annual mean surface temperature anomaly. The posterior is conditioned on experiments $\cD_{\text{train}} = \{historical, SSP126, SSP245, SSP585\}$. The position of the NorESM2-LM simulated data is given at the end of the section.}
    \label{fig:ssp370-sample-paths}
\end{figure}
\begin{figure}[H]
    \centering
    \includegraphics[width=\linewidth]{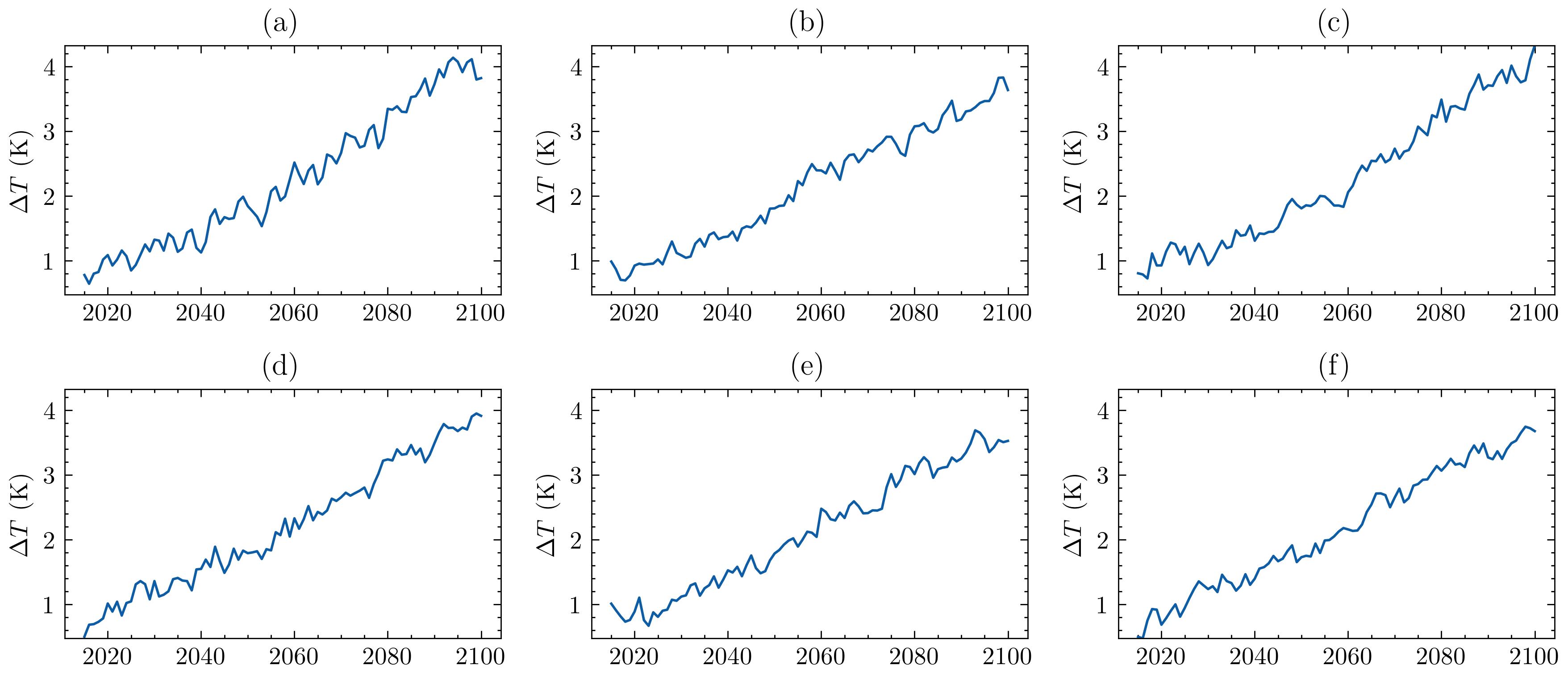}
    \caption{Example of 5 sample paths from the predicted posterior distribution over global annual mean surface temperature anomaly under \textit{SSP585} with FaIRGP + NorESM2-LM simulated global annual mean surface temperature anomaly. The posterior is conditioned on experiments $\cD_{\text{train}} = \{historical, SSP126, SSP245, SSP370\}$. The position of the NorESM2-LM simulated data is given at the end of the section.}
    \label{fig:ssp585-sample-paths}
\end{figure}

\noindent Positions of NorESM2-LM simulated data : (e) in Figure~\ref{fig:ssp126-sample-paths}, (f) in Figure~\ref{fig:ssp245-sample-paths}, (f) in Figure~\ref{fig:ssp370-sample-paths} and (a) in Figure~\ref{fig:ssp585-sample-paths}.

\newpage
\subsection{Comparison of plain GP baseline with ClimateBench GP emulator from \protect\citeA{watsonparris2021climatebench}}\label{appendix:comparison-with-climatebench}

\subsubsection{Modelling differences}

The plain GP baseline we use is analogous to the GP emulator from ClimateBench~\cite{watsonparris2021climatebench}, but differs in two aspects.

\textbf{First}, we use in our plain GP baseline an anisotropic covariance structure, denoted as $\rho(E, E') = C_{3/2}(E, E')$, where $C_{3/2}$ denotes the Matérn-3/2 covariance (see Appendix~\ref{appendix:matern-covariance}).

This kernel introduces a different lengthscale parameter $\ell_\chi$ for each atmospheric agent $\chi$, which is tuned through maximizing the marginal loglikelihood. It is explicitly given by
\begin{equation}
    \rho(E, E') = C_{3/2}(E, E') = \left(1 + \sqrt{3 \sum_\chi \frac{(E^\chi - E'^\chi)^2}{\ell_\chi^2}}\right)\exp\left(-\sqrt{3 \sum_\chi \frac{(E^\chi - E'^\chi)^2}{\ell_\chi^2}}\right),
\end{equation}
where $E^\chi$ denotes global emission level for agent $\chi \in \{\text{CH}_4, \text{SO}_2, \text{BC}\}$, while for $\chi = \text{CO}2$, it denotes global cumulative emission levels.

In contrast, \citeA{watsonparris2021climatebench} adopt an additive kernel structure given by
\begin{equation}
    \rho(E, E') = \sum_\chi \sigma_\chi C_{3/2}(E^\chi, E'^\chi) = \sum_\chi \sigma_\chi \left(1 + \sqrt{3}\frac{|E^\chi - E'^\chi|}{\ell_\chi}\right)\exp\left(-\sqrt{3}\frac{|E^\chi - E'^\chi|}{\ell_\chi}\right).
\end{equation}
This additive structure introduces additional complexity through the variance terms $\sigma_\chi$, which modulate the of each atmospheric agent to the covariance, and need to be tuned alongside the lengthscales $\ell_\chi$. From a functional perspective, choosing an additive kernel is equivalent to representing the temperature response as the sum of independent GPs, where each GP models the response to a single atmospheric agent. Further details on the covariance function be found in \cite[Appendix A2]{watsonparris2021climatebench}.

\textbf{Second}, in our plain GP baseline $E^{\text{SO}_2}$ and $E^{\text{BC}}$ correspond to global emission levels, whereas \citeA{watsonparris2021climatebench} used the 5 principal components of spatial emission maps for SO\textsubscript{2} and BC. We anticipate that spatially-resolved inputs for aerosol emissions should improve the model's ability to predict spatially-resolved surface temperature features.

\subsubsection{Predictive performance comparison}

To compare the predictive performance against our plain GP and FaIRGP, we evaluate the ClimateBench GP for the emulation of global and spatial mean surface temperature anomaly and report scores in Table~\ref{table:comparison-with-climatebench}. All emulators are trained on the same training data: \textit{historical, SSP126, SSP370, SSP585}. Predictive performance is evaluated for the emulation of \textit{SSP245} since this is the test data used in \cite{watsonparris2021climatebench}.

\begin{table}[h]
    \centering
    \caption{Scores of ClimateBench GP, our Plain GP baseline and FaIRGP for the task of emulating global and spatial mean surface temperatures on \textit{SSP245}; scores are computed over 2015-2100 period for global emulation and over 2080-2100 period for spatial emulation; the best emulator scores are highlighted in bold; $\uparrow\!/\!\downarrow$ indicates higher/lower is better.}
        \resizebox{\linewidth}{!}{
        \begin{tabular}{llcccccc}
        \toprule
         {} & Emulator &  RMSE\small{$\;\downarrow$}  &    MAE\small{$\;\downarrow$}  &   Bias  &     LL\small{$\;\uparrow$}  & Calib95 &   CRPS\small{$\;\downarrow$}  \\
        \midrule
        \multirow{3}{*}{\textit{Global}}  & ClimateBench GP &  0.12 &  0.10 &  \textbf{0.01} & -0.45 & \textbf{1.0} &   0.15 \\
                  & Plain GP &  \textbf{0.09} &  \textbf{0.07} &  0.03 & 0.67 &    \textbf{1.0} &  \textbf{0.06} \\
                  & FaIRGP &  \textbf{0.09} &  0.08 &  \textbf{-0.01} & \textbf{0.71} &    \textbf{1.0} &  \textbf{0.06} \\ \thinrule
        \multirow{3}{*}{\textit{Spatial}} & ClimateBench GP &  0.43 &  0.32 &  -0.13 &      -0.69 &  0.98 &  0.25 \\
                  & Plain GP &  0.56 &  0.40 &  -0.19 & -0.70 & \textbf{0.97} &  0.28 \\
                  & FaIRGP &  \textbf{0.36} &  \textbf{0.26} &  \textbf{-0.05} & \textbf{-0.34} &  0.99 & \textbf{0.20} \\ 
        \bottomrule
        \end{tabular}
        }
    \label{table:comparison-with-climatebench}
\end{table}

Overall, FaIRGP demonstrates improved scores compared to the ClimateBench baseline GP across various metrics, both in terms of global and spatial emulation. In terms of global emulation, the predictive performance of the Plain GP appears to be better than that of the ClimateBench GP. However, when it comes to spatial surface temperature emulation, the ClimateBench GP outperforms the Plain GP. This improvement can likely be attributed to the utilization of spatially-resolved aerosol information as input in the ClimateBench GP.

\newpage
\subsection{Tuned model parameters values}\label{appendix:model-parameters-values}

\begin{table}[h]
    \centering
        \begin{tabular}{lcccccc}
        \toprule
        Training data &  $\sigma^2$ &  $\sigma^2_\sfF$ &  $\ell_{CO_2}$ &  $\ell_{CH_4}$ & $\ell_{SO_2}$ & $\ell_{BC}$ \\
        \midrule
        All SSPs and historical        &    0.754 &                 0.244 &       1.60 &       3.79 &       2.98 &     3.797865 \\
        All but SSP126       &    0.679 &                 0.271 &       1.57 &       3.90 &       2.91 &     3.86 \\
        All but SSP245       &    0.845 &                 0.276 &       1.52 &       3.61 &       2.86 &     3.89 \\
        All but SSP370       &    0.709 &                 0.241 &       1.87 &       3.87 &       2.91 &     3.40 \\
        All but SSP585       &    0.647 &                 0.192 &       1.58 &       3.22 &       2.41 &     3.18 \\
        Historical only &    0.385 &                 0.263 &       3.33 &       2.57 &       1.02 &     3.48 \\
        \bottomrule
        \end{tabular}
    \caption{Values obtained for internal variability variance $\sigma^2$, forcing kernel variance $\sigma^2_\sfF$ and forcing kernel lengthscales $\ell_{CO_2}, \ell_{CH_4}, \ell_{SO_2}, \ell_{BC}$ after tuning using the gradient descent maximum likelihood procedure described in Section~\ref{section:parameters-tuning}.}
    \label{tab:my_label}
\end{table}

\newpage

\subsection{Emulating anthropogenic aerosols forcing}\label{appendix:aerosol-emulation-results}

\begin{figure}[H]
    \centering
    \includegraphics[width=\textwidth]{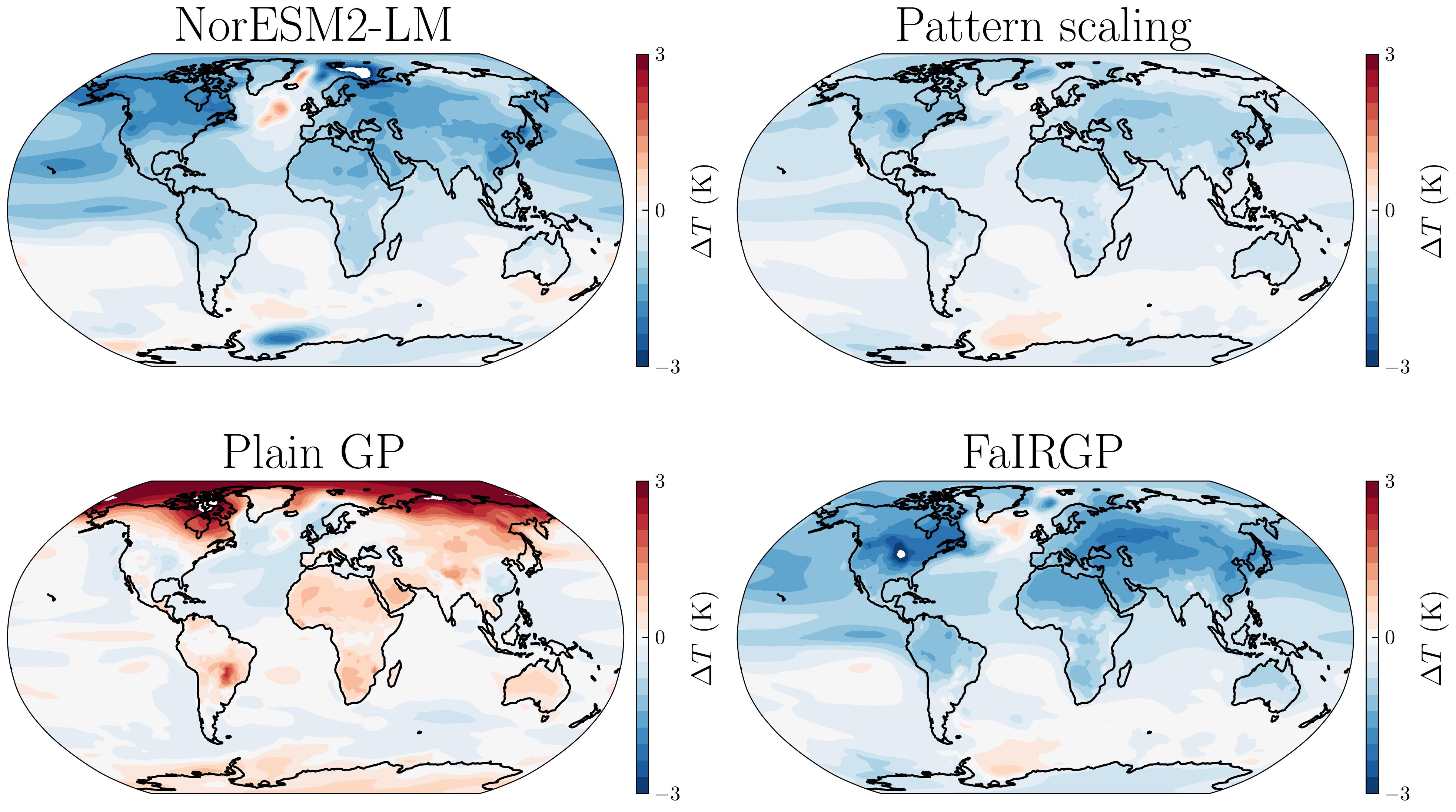}
    \caption{Maps of surface temperatures averaged for the \textit{hist-aer} experiment averaged over 1980-2014 period. \textbf{Top-Left:} Groundtruth NorESM2-LM surface temperatures map. \textbf{Top-Right:} Prediction with the FaIR pattern scaling baseline. \textbf{Bottom-Left:} Prediction with Plain GP posterior mean. \textbf{Bottom-Right:} Prediction with the FaIRGP posterior mean.}
\end{figure}

\end{document}

%% file: lettres_maths.tex
\def\NN{{\mathbb N}}    
\def\RR{{\mathbb R}}    
\def\EE{{\mathbb E}}    
\def\11{{\mathbf 1}}    

        \def\cN{{\mathcal N}}       \def\cD{{\mathcal D}}                 

 \def\bb{{\mathbf b}} \def\bc{{\mathbf c}}   \def\bff{{\mathbf f}}       \def\bm{{\mathbf m}}       \def\bt{{\mathbf t}} \def\bu{{\mathbf u}}   \def\bx{{\mathbf x}} \def\by{{\mathbf y}}  

\def\bA{{\mathbf A}}   \def\bD{{\mathbf D}} \def\bE{{\mathbf E}}    \def\bI{{\mathbf I}}  \def\bK{{\mathbf K}}         \def\bT{{\mathbf T}} \def\bU{{\mathbf U}} \def\bV{{\mathbf V}}  \def\bX{{\mathbf X}}  \def\bZ{{\mathbf Z}}

     \def\sfff{{\mathsf f}}                   \def\sfy{{\mathsf y}}  

 \def\sfB{{\mathsf B}}    \def\sfF{{\mathsf F}}             \def\sfS{{\mathsf S}} \def\sfT{{\mathsf T}}      \def\sfZ{{\mathsf Z}}

                       \def\sfbX{\boldsymbol{\mathsf X}}  \def\sfbZ{\boldsymbol{\mathsf Z}}

\newcommand{\indep}{\perp \!\!\!\!\!\; \perp}

\def\d{\,{\mathrm d}}

\def\GP{{\operatorname{GP}}}
\def\bPhi{{\boldsymbol{\Phi}}}
\def\bGamma{{\boldsymbol{\Gamma}}}

%% file: theoremes.tex
\theoremstyle{plain}
	\newtheorem{theorem}{Theorem}[section] 
	\newtheorem{lemma}[theorem]{Lemma}            

\theoremstyle{definition}

\theoremstyle{remark}
